\documentclass[12pt]{article}%
\usepackage{amssymb}
\usepackage{amsfonts}
\usepackage{amsmath}
\usepackage[dvips]{graphicx}
\usepackage{version}
\usepackage[dvipsnames,usenames]{color}%
\providecommand{\U}[1]{\protect\rule{.1in}{.1in}}
\newtheorem{theorem}{Theorem}

\newtheorem{condition}[theorem]{Condition}

\newtheorem{corollary}[theorem]{Corollary}

\newtheorem{definition}[theorem]{Definition}
\newtheorem{example}[theorem]{Example}

\newtheorem{lemma}[theorem]{Lemma}

\newtheorem{proposition}[theorem]{Proposition}

\newenvironment{proof}[1][Proof]{\textbf{#1.} }{\ \rule{0.5em}{0.5em}}
\makeatletter
\evensidemargin=\oddsidemargin
\newdimen\dummy
\dummy=\oddsidemargin
\def\tighttoc{\def\l@section{\@dottedtocline{1}{0em}{1.4em}}}
\tighttoc
\if@twoside
\def\ps@headings{\let\@mkboth\markboth
\def\@oddfoot{}\def\@evenfoot{}\def\@evenhead{\rm \thepage\hfil \sl
\leftmark}\def\@oddhead{\hbox{}\sl \rightmark \hfil
\rm\thepage}\def\chaptermark##1{\markboth {\uppercase{\ifnum \c@secnumdepth
>\m@ne
\@chapapp\ \thechapter. \ \fi ##1}}{}}\def\sectionmark##1{\markright
{\uppercase{\ifnum \c@secnumdepth >\z@
\thesection. \ \fi ##1}}}}
\else
\def\ps@headings{\let\@mkboth\markboth
\def\@oddfoot{}\def\@evenfoot{}\def\@oddhead{\hbox {}\sl \rightmark \hfil
\rm\thepage}\def\chaptermark##1{\markright {\uppercase{\ifnum \c@secnumdepth
>\m@ne
\@chapapp\ \thechapter. \ \fi ##1}}}}
\fi
\def\ps@myheadings{\let\@mkboth\@gobbletwo
\def\@oddhead{\hbox{}\sl\rightmark \hfil
\rm\thepage}\def\@oddfoot{}\def\@evenhead{\rm \thepage\hfil\sl\leftmark\hbox
{}}\def\@evenfoot{}\def\chaptermark##1{}\def\sectionmark##1{}\def\subsectionmark##1{}}
\makeatother
\begin{document}

\title{Lagrangian Framework for Systems Composed of High-Loss and Lossless Components}
\author{Alexander Figotin\\University of California at Irvine
\and Aaron Welters\\Massachusetts Institute of Technology}
\maketitle

\begin{abstract}
Using a Lagrangian mechanics approach, we construct a framework to study the
dissipative properties of systems composed of two components one of which is
highly lossy and the other is lossless. We have shown in our previous work
that for such a composite system the modes split into two distinct classes,
high-loss and low-loss, according to their dissipative behavior. A principal
result of this paper is that for any such dissipative Lagrangian system, with
losses accounted by a Rayleigh dissipative function, a rather universal
phenomenon occurs, namely, \textit{selective overdamping}: The high-loss modes
are all overdamped, i.e., non-oscillatory, as are an equal number of low-loss
modes, but the rest of the low-loss modes remain oscillatory each with an
extremely high quality factor that actually increases as the loss of the lossy
component increases. We prove this result using a new time dynamical
characterization of overdamping in terms of a virial theorem for dissipative
systems and the breaking of an equipartition of energy.

\end{abstract}

\section{Introduction}

In this paper we introduce a general Lagrangian framework to study the
dissipative properties of two component systems composed of a high-loss and
lossless components which can have gyroscopic properties. This framework
covers any linear Lagrangian system provided it has a finite number of degrees
of freedom, a nonnegative Hamiltonian, and losses accounted by the Rayleigh
dissipative function, \cite[Sec. 10.11, 10.12]{Pars}, \cite[Sec. 8, 9,
46]{Gantmacher}. Such physical systems include, in particular, many different
types of damped mechanical systems or electric networks.

\paragraph{Motivation.}

We are\ looking to design and study two-component composite dielectric media
consisting of a high-loss and low-loss component in which the lossy component
has a useful property (or functionality) such as magnetism. We want to
understand what is the trade-off between the losses and useful properties
inherited by the composite from its components. Our motivation comes from a
major problem in the design of such structures where a component which carries
a useful property, e.g., magnetism, has prohibitively strong losses in the
frequency range of interest. Often this precludes the use of such a lossy
component with otherwise excellent physically desirable properties. Then the
question stands: Is it possible to design a composite system having a useful
property at a level comparable to that of its lossy component but with
significantly reduced losses over a broad frequency range?

An important and guiding example of a two-component dielectric medium composed
of a high-loss and lossless components was constructed in \cite{FigVit8}. The
example was a simple layered structure which had magnetic properties
comparable with a natural bulk material but with 100 times lesser losses in a
wide frequency range. This example demonstrated that it is possible to design
a composite material/system which can have a desired property comparable with
a naturally occurring bulk substance but with significantly reduced losses.

In order to understand the general mechanism for this phenomenon,\ the authors
in \cite{FigWel1} considered a general dynamical system as the model for such
a medium where the high-loss component of the medium was represented as a
significant fraction of the entire system. We have found that for such a
system the losses of the entire structure become small provided that the lossy
component is sufficiently lossy. The general mechanism of this phenomenon,
which we proved in \cite{FigWel1}, is the \emph{modal dichotomy} when the
entire set of modes of the system splits into two distinct classes,
\emph{high-loss} and \emph{low-loss modes}, based on their dissipative
properties. The higher loss modes for a wide range of frequencies contribute
very little to losses and this is how the entire structure can be low loss,
whereas the useful property is still present. In that work the way we
accounted for the presence of a useful property was by simply demanding that
the lossy component was a significant fraction of the entire composite
structure without explicitly correlating the useful property and the losses.

In this paper, we continue our studies on the modal dichotomy and overdamping,
which began in \cite{FigWel1}, but now we focus are attention on
gyroscopic-dissipative systems. Although our results on modal dichotomy apply
to the full generality of the dynamical systems considered here, our results
on overdamping in this paper will be restricted to systems without gyroscopy.
Consideration of overdamping in systems with gyroscopy will be considered in a
future work.

\paragraph{Overview of results.}

Now in order to account in a general form for the physical properties of the
two-component composite system with a high-loss and lossless components we
introduce a Lagrangian framework with dissipation. This framework can be a
basis for studying the interplay of physical properties of interest and
dissipation. These general physical properties include in particular gyroscopy
(or gyrotropy), which is intimately related to magnetic properties, different
symmetries and other phenomena such as overdamping that are not present in
typical dynamical systems.

In the Lagrangian setting, the dissipation is taken into account by means of
the Rayleigh dissipation function \cite[Sec. 10.11, 10.12]{Pars}, \cite[Sec.
8, 9, 46]{Gantmacher}. In the case of two-component composite with lossy and
lossless components, the Rayleigh function is assumed to affect only a
fraction $0<\delta_{R}<1$ of the total degrees of freedom of the system. In a
rough sense, the loss fraction $\delta_{R}$ signifies the fraction of the
degrees of freedom that are effected by losses.

We now give a concise qualitative description of the main results of this
paper which are discussed next and organized under the following three topics:
Overdamping and selective overdamping in Section \ref{sodan}; Virial theorem
for dissipative systems and equipartition of energy in Section \ref{svir};
Standard spectral theory vs. Krein theory in Sections \ref{smodel} and
\ref{sspan}.

\textbf{Overdamping and selective overdamping}. Overdamping is a regime in
dissipative Lagrangian systems in which some of the systems eigenmodes are
overdamped, i.e., have exactly zero frequency or, in other words, are
non-oscillatory when losses are sufficiently large. The phenomenon of
overdamping has been studied thoroughly but only in the case when the entire
system is overdamped \cite{Duff55}, \cite[\S 7.6 \& Chap. 9]{Lan66},
\cite{BarLan92} meaning all its modes are non-oscillatory. The Lagrangian
framework and the methods we develop in this paper are applicable to the more
general case allowing for only a fraction of the modes to be susceptible to
overdamping, a phenomenon we call \emph{selective overdamping}. Our analysis
of overdamping and this selective overdamping phenomenon is carried out in
Section \ref{sodan}. Our main results are Theorems \ref{tcovd}, \ref{tpovd},
\ref{cpodr}, \ref{thmod}, and \ref{tseod}.

Our interest in selective overdamping is threefold. First, as we prove in this
paper the selective overdamping of the system is a rather \emph{universal
phenomenon} that can occur for a two-component system when the losses of the
lossy component are sufficiently large. Moreover, we show that \emph{the large
losses in the lossy component} cause not only the modal dichotomy, but also
\emph{results in overdamping of all of the high-loss modes while a positive
fraction of the low-loss modes remain oscillatory}. Second, since mode
excitation is efficient only when its frequency matches the frequency of the
excitation force, if these high-loss modes go into the overdamping regime
having exactly zero frequency these modes cannot be excited efficiently and
therefore their associated losses are essentially eliminated. This explains
the absorption suppression for systems composed of lossy and lossless
components. Third, as we will show, as the losses in the lossy component
increase the overdamped high-loss modes are more suppressed while all the
low-loss oscillatory modes are more enhanced with increasingly high quality
factor. \emph{This provides a mechanism for selective enhancement of these
high quality factor, low-loss oscillatory modes and selective suppression of
the high-loss non-oscillatory modes}.

In fact, we show that \emph{the fraction of all overdamped modes is exactly
the loss fraction }$\delta_{R}$\emph{\ which satisfies }$0<\delta_{R}<1$, and
half of these overdamped modes are the high-loss modes. Thus it is
\emph{exactly the positive fraction }$1-\delta_{R}>0$\emph{\ of modes which
are low-loss oscillatory modes}. In addition, \emph{these latter modes have
extremely high quality factor which increases as the losses in the lossy
component increase}. An example of this behavior using the electric circuit in
Section \ref{scircuits} was shown numerically in \cite{FigWel1}.

\textbf{Virial theorem for dissipative systems and equipartition of energy}.
Recall that the virial theorem from classical mechanics \cite[pp. 83-86;
\S 3.4]{Goldstein} (which we review in Appendix \ref{ApdxVirThm}) is about the
oscillatory transfer of energy from one form to another and its equipartition.
For conservative linear systems whose Lagrangian is the difference between the
kinetic and potential energy, e.g., a spring-mass system or a parallel LC
circuit, the virial theorem says that for any time-periodic state of the
system, the time-averaged kinetic energy equals the time-averaged potential
energy. This result is known as the equipartition of energy because of the
fact that the total energy of the system, i.e., the Hamiltonian, is equal to
the sum of the kinetic and potential energy. It is this result that we
generalize in this paper for dissipative Lagrangian systems and their damped
oscillatory modes.

In particular, we show that \emph{for any damped oscillatory mode of the
system the kinetic energy equals the potential energy} and so \emph{there is
an equipartition of the system energy for these damped oscillatory modes into
kinetic and potential energy}. Our precise statements of these results can be
found in Section \ref{svir} in Theorem \ref{tvir} and Corollary \ref{tvir_eqp}.

We have also found that the transition to overdamping is characterized by a
breakdown of the virial theorem. More specifically, as the losses of the lossy
component increase some of the modes become overdamped with complete ceasing
of any oscillations and with breaking of the equality between kinetic and
potential energy. Such a \emph{breaking of the equipartition of the energy can
viewed as a dynamical characterization of the overdamped modes! }So we now
have two different characterizations of overdamping: spectral and dynamical.
The effectiveness of this dynamical characterization is demonstrated in our
study of the selective overdamping phenomenon and in proving that a positive
fraction of modes always remain oscillatory no matter how large the losses
become in the lossy component of the composite system.

\textbf{Standard spectral theory vs. Krein spectral theory}. The study of the
eigenmodes of a Lagrangian system often relies on the fact the system
evolution can be transformed into the Hamiltonian form as the first-order
linear differential equations. However, it is not emphasized enough that the
Hamiltonian setting does not lead directly to the standard spectral theory of
self-adjoint or dissipative operators but that only it might be reduced to one
in some important and well known cases by a proper transformation as, for
instance, in the case of a simple oscillator.

The general spectral theory of Hamiltonian systems is known as the \emph{Krein
spectral theory} \cite[\S 42]{Arnold}, \cite{Yak}, \cite{YakSta}. This theory
is far more complex than the standard spectral theory and it is much harder to
apply. We have found, as discussed in Section \ref{smodel}, a general
transformation that reduces the evolution equation to the standard spectral
theory under the condition of positivity of the Hamiltonian. The standard
spectral theory is complete, well understood, and allows for an elaborate and
effective perturbation theory. We used that all in our analysis of the modal
dichotomy and symmetries of the spectrum in Section \ref{sspan} as well as the
selective overdamping phenomenon in Section \ref{sodan}.

Our main results on spectral symmetries are Proposition \ref{pssym} and
Corollary \ref{cssym}. Our principal results on the modal dichotomy can be
found in Theorems \ref{tmdic}, \ref{cpodr} and Corollaries \ref{cpfsp},
\ref{cmdic}. In fact, it is important to point out that the key result of this
paper that leads to the modal dichotomy which is essential in the analysis of
overdamping is Proposition \ref{pevbd} which relates to certain eigenvalue
bounds. Moreover, these bounds and the dichotomy come from some general
results that we derive in Appendix \ref{apxebd} on the perturbation theory for
matrices with a large imaginary part.

\paragraph{Organization of the paper.}

The rest of the paper is organized as follows. Section \ref{smodel} sets up
and discusses the Lagrangian framework approach and model used in this paper
to study the dissipative properties of two-component composite systems with a
high-loss and a lossless component. In Section \ref{scircuits} we apply the
developed approach to an electric circuit showing all key features of the
method. Sections \ref{svir}--\ref{sodan} are devoted to a precise formulation
of all significant results in the form of Theorems, Propositions, and so on
(with Section \ref{sproofs} providing the proofs of these results). More
specifically, Section \ref{svir} discusses our extension of the virial theorem
and equipartition of energy from classical mechanics for conservative
Lagrangian systems to dissipative systems. Section \ref{sspan} is on the
spectral analysis of the system eigenmodes including spectral symmetries,
modal dichotomy, and perturbation theory in the high-loss regime. Section
\ref{sodan} contains our analysis of the overdamping phenomena including
selective overdamping. Appendices \ref{apxsc} and \ref{apxebd} contain results
we need that we believe will be of use more generally in studying dissipative
dynamical systems, but especially for composite systems with high-loss and
lossless components. In particular, Appendix \ref{apxebd} is on the
perturbation theory for matrices with a large imaginary part. Finally,
Appendices \ref{apenerg} and \ref{ApdxVirThm} review the virial theorem from
classical mechanics from the energetic point of view.

\section{Model setup and discussion\label{smodel}}

\subsection{Integrating the dissipation into the Lagrangian
frameworks\label{slagrangian}}

The Lagrangian and Hamiltonian frameworks have numerous well known advantages
in describing the evolution of physical systems. These advantages include, in
particular, the universality of the mathematical structure, the flexibility in
the choice of variables and the incorporation of the symmetries with the
corresponding conservations laws. Consequently, it seems quite natural and
attractive to integrate dissipation into the Lagrangian and Hamiltonian
frameworks. We are particular interested in integrating dissipation into
systems with gyroscopic features such as dielectric media with magnetic
components and electrical networks.

We start with setting up the Lagrangian and Hamiltonian structures similarly
to those in \cite{FigVit1} and \cite{FigSch2}. In particular, we use matrices
to efficiently handle possibly large number of degrees of freedom. Since we
are interested in systems evolving linearly the Lagrangian $\mathcal{L}$ is
assumed to be a quadratic function (bilinear form) of the coordinates
$Q=\left[  q_{r}\right]  _{r=1}^{N}$ (column vector) and their time
derivatives $\dot{Q}$, that is
\begin{equation}
\mathcal{L}=\mathcal{L}\left(  Q,\dot{Q}\right)  =\frac{1}{2}\left[
\begin{array}
[c]{l}%
\dot{Q}\\
Q
\end{array}
\right]  ^{\mathrm{T}}M_{\mathrm{L}}\left[
\begin{array}
[c]{l}%
\dot{Q}\\
Q
\end{array}
\right]  ,\qquad M_{\mathrm{L}}=\left[
\begin{array}
[c]{ll}%
\alpha & \theta\\
\theta^{\mathrm{T}} & -\eta
\end{array}
\right]  , \label{dlag1}%
\end{equation}
where $\mathrm{T}$ denotes the matrix transposition operation, and
$\alpha,\eta$ and $\theta$ are $N\times N$-matrices with real-valued entries.
In addition to that, we assume $\alpha,\eta$ to be symmetric matrices which
are positive definite and positive semidefinite, respectively, and we assume
$\theta$ is skew-symmetric, that is
\begin{equation}
\alpha=\alpha^{\mathrm{T}}>0,\qquad\eta=\eta^{\mathrm{T}}\geq0,\qquad
\theta^{\mathrm{T}}=-\theta. \label{dlag2}%
\end{equation}
Notice that only the skew-symmetric part of $\theta$ in (\ref{dlag1}) matters
for the system dynamics because of the form of the Euler-Lagrange equations
(\ref{dlag4}) and so without loss of generality we can assume as we have that
the matrix $\theta$ is skew-symmetric. In fact, the symmetric part of the
matrix $\theta$ if any alters the Lagrangian by a complete time derivative and
consequently can be left out as we have done.

The Lagrangian (\ref{dlag1}) can be written in the energetic form%
\begin{equation}
\mathcal{L}=\mathcal{T}-\mathcal{V}, \label{dlag3}%
\end{equation}
where
\begin{equation}
\mathcal{T}=\mathcal{T}(\dot{Q},Q)=\frac{1}{2}\dot{Q}^{\mathrm{T}}\alpha
\dot{Q}+\frac{1}{2}\dot{Q}^{\mathrm{T}}\theta Q,\qquad\mathcal{V}%
=\mathcal{V}(\dot{Q},Q)=\frac{1}{2}Q^{\mathrm{T}}\eta Q-\frac{1}{2}\dot
{Q}^{\mathrm{T}}\theta Q \label{dlag3_1}%
\end{equation}
are interpreted as the kinetic and potential energy, respectively. By
Hamilton's principle the dynamics of the system are governed by the
Euler-Lagrange equations
\begin{equation}
\frac{d}{dt}\left(  \frac{\partial\mathcal{L}}{\partial\dot{Q}}\right)
-\frac{\partial\mathcal{L}}{\partial Q}=0, \label{dlag4}%
\end{equation}
which are the following second-order ordinary differential equation (ODEs)
\begin{equation}
\alpha\ddot{Q}+2\theta\dot{Q}+\eta Q=0, \label{dlag5}%
\end{equation}

The above second-order ODEs can be turned in the first-order ODEs with help of
the Hamiltonian function defined by the Lagrangian through the Legendre
transformation
\begin{equation}
\mathcal{H}=\mathcal{H}\left(  P,Q\right)  =P^{\mathrm{T}}\dot{Q}%
-\mathcal{L}\left(  Q,\dot{Q}\right)  ,\text{ where }P=\frac{\partial
\mathcal{L}}{\partial\dot{Q}}=\alpha\dot{Q}+\theta Q. \label{dlag6}%
\end{equation}
Hence,%
\begin{equation}
\dot{Q}=\alpha^{-1}\left(  P-\theta Q\right)  \label{dlag6a}%
\end{equation}
and, consequently
\begin{equation}
\mathcal{H}\left(  P,Q\right)  =\frac{1}{2}\left[  \left(  P-\theta Q\right)
^{T}\alpha^{-1}\left(  P-\theta Q\right)  +Q^{T}\eta Q\right]  =\frac{1}%
{2}\dot{Q}^{T}\alpha\dot{Q}+\frac{1}{2}Q^{T}\eta Q. \label{dlag7}%
\end{equation}
In particular, the Hamiltonian is the sum of the kinetic and potential energy,
that is%
\begin{equation}
\mathcal{H}=\mathcal{T}+\mathcal{V}. \label{dlag7c}%
\end{equation}
We remind that the Hamiltonian function $\mathcal{H}\left(  P,Q\right)  $ is
interpreted as the system energy which is a conserved quantity, that is%
\begin{equation}
\partial_{t}\mathcal{H}\left(  P,Q\right)  =0. \label{dlag7b}%
\end{equation}
The function $\mathcal{H}\left(  P,Q\right)  $ defined by (\ref{dlag7}) is a
quadratic form associated with a matrix $M_{\mathrm{H}}$ through the relation
\begin{equation}
\mathcal{H}\left(  P,Q\right)  =\frac{1}{2}\left[
\begin{array}
[c]{l}%
P\\
Q
\end{array}
\right]  ^{\mathrm{T}}M_{\mathrm{H}}\left[
\begin{array}
[c]{l}%
P\\
Q
\end{array}
\right]  , \label{dlag8}%
\end{equation}
where $M_{\mathrm{H}}$ is the $2N\times2N$ matrix having the block form
\begin{equation}
M_{\mathrm{H}}=\left[
\begin{array}
[c]{ll}%
\alpha^{-1} & -\alpha^{-1}\theta\\
-\theta^{\mathrm{T}}\alpha^{-1} & \theta^{\mathrm{T}}\alpha^{-1}\theta+\eta
\end{array}
\right]  =\left[
\begin{array}
[c]{ll}%
\mathbf{1} & 0\\
-\theta^{\mathrm{T}} & \mathbf{1}%
\end{array}
\right]  \left[
\begin{array}
[c]{ll}%
\alpha^{-1} & 0\\
0 & \eta
\end{array}
\right]  \left[
\begin{array}
[c]{ll}%
\mathbf{1} & -\theta\\
0 & \mathbf{1}%
\end{array}
\right]  , \label{dlag7a}%
\end{equation}
where $\mathbf{1}$ is the identity matrix. Notice that the representations
(\ref{dlag8}), (\ref{dlag7a}) combined with the inequalities (\ref{dlag2})
imply%
\begin{equation}
\mathcal{H}\left(  P,Q\right)  \geq0\text{ and }M_{\mathrm{H}}=M_{\mathrm{H}%
}^{\mathrm{T}}\geq0. \label{dlag8a}%
\end{equation}
The Hamiltonian evolution equations then take the form%
\begin{equation}
\partial_{t}\left[
\begin{array}
[c]{l}%
P\\
Q
\end{array}
\right]  =JM_{\mathrm{H}}\left[
\begin{array}
[c]{l}%
P\\
Q
\end{array}
\right]  ,\quad J=\left[
\begin{array}
[c]{ll}%
0 & -\mathbf{1}\\
\mathbf{1} & 0
\end{array}
\right]  \text{, } \label{dlag9}%
\end{equation}
Observe that the symplectic matrix $J$ satisfies the following relations%
\begin{equation}
J^{2}=-\mathbf{1},\quad J=-J^{\mathrm{T}}. \label{dlag9a}%
\end{equation}

\textbf{The Rayleigh dissipation function. }Up to now we have dealt with a
conservative system\ satisfying the energy conservation (\ref{dlag7b}). Let us
introduce now dissipative forces using Rayleigh's method described in
\cite[Sec. 10.11, 10.12]{Pars}, \cite[Sec. 8, 9, 46]{Gantmacher}. The
\emph{Rayleigh dissipation function} $\mathcal{R}$ is defined as a quadratic
function of the generalized velocities, namely
\begin{equation}
\mathcal{R}=\mathcal{R}\left(  \dot{Q}\right)  =\frac{1}{2}\dot{Q}%
^{\mathrm{T}}\beta R\dot{Q},\text{ \ \ }R\not =0,\text{ \ \ }R=R^{\mathrm{T}%
}\geq0,\text{ \ \ }\beta\geq0, \label{radis1}%
\end{equation}
where the scalar $\beta$ is a dimensionless \emph{loss parameter} which we
introduce to scale the intensity of dissipation. In particular, $R$ is an
$N\times N$ symmetric matrix with real-valued entries and, most importantly,
is positive semidefinite.

The dissipation is then introduced through the following general
Euler-Lagrange equations of motion with forces%
\begin{equation}
\frac{d}{dt}\left(  \frac{\partial\mathcal{L}}{\partial\dot{Q}}\right)
-\frac{\partial\mathcal{L}}{\partial Q}=-\frac{\partial\mathcal{R}}%
{\partial\dot{Q}}+F, \label{radis2}%
\end{equation}
where $\frac{\partial\mathcal{R}}{\partial\dot{Q}}$ are generalized
dissipative forces and $F=F\left(  t\right)  $ is an external force, yielding
the following second-order ODEs
\begin{equation}
\alpha\ddot{Q}+\left(  2\theta+\beta R\right)  \dot{Q}+\eta Q=F.
\label{radis2a}%
\end{equation}

\subsection{The Lagrangian system and two-component composite
model\label{sLagrModelEigenmQFac}}

Now by a linear (dissipative) \emph{Lagrangian system} we mean a system whose
state is described by a time-dependent $Q=Q(t)$ taking values in the Hilbert
space $%
%TCIMACRO{\U{2102} }%
%BeginExpansion
\mathbb{C}
%EndExpansion
^{N}$ with the standard inner product $(\cdot,\cdot)$ (i.e., $(a,b)=a^{\ast}b
$, where $\ast$ denotes the conjugate transpose, i.e., $a^{\ast}=\overline
{a}^{\mathrm{T}}$) whose dynamics are governed by the ODEs (\ref{radis2a}).

The energy balance equation for any such state, which follows from
(\ref{radis2a}), is%

\begin{equation}
\frac{d\mathcal{H}}{dt}=-2\mathcal{R}+\operatorname{Re}\left(  \dot
{Q},F\right)  , \label{radis4}%
\end{equation}
where $\mathcal{H=T+V}$\ as in (\ref{dlag7c}) but now instead of
(\ref{dlag3_1})\ we have%
\begin{gather}
\mathcal{T}=\mathcal{T}(\dot{Q},Q)=\frac{1}{2}\left(  \dot{Q},\alpha\dot
{Q}\right)  +\frac{1}{2}\operatorname{Re}\left(  \dot{Q},\theta Q\right)
,\text{ }\label{radis4_0}\\
\mathcal{V}=\mathcal{V}(\dot{Q},Q)=\frac{1}{2}\left(  Q,\eta Q\right)
-\frac{1}{2}\operatorname{Re}\left(  \dot{Q},\theta Q\right)  ,\nonumber\\
\mathcal{R}=\mathcal{R}\left(  \dot{Q}\right)  =\frac{1}{2}\left(  \dot
{Q},\beta R\dot{Q}\right)  .\nonumber
\end{gather}
We continue to interpret $\mathcal{T}$, $\mathcal{V}$ as the kinetic and
potential energies, $\mathcal{H}$ as the system energy, and $2\mathcal{R}$ as
the dissipated power. The term $\operatorname{Re}\left(  \dot{Q},F\right)  $
is interpreted as the rate of work done by the force $F$.

The physical significance of the various energetic terms $\mathcal{E}$, where
$\mathcal{E}\in\left\{  \mathcal{H},\mathcal{T},\mathcal{V},\mathcal{R}%
\right\}  $, is that for a complex-valued state $Q=Q_{1}+\mathrm{i}Q_{2}$ with
real-valued $Q_{1},Q_{2},F$ both $Q_{1}$ and $Q_{2}$ are also states
representing physical solutions and $\mathcal{E}\left(  \dot{Q},Q\right)
=\mathcal{E}\left(  \dot{Q}_{1},Q_{1}\right)  +\mathcal{E}\left(  \dot{Q}%
_{2},Q_{2}\right)  $. These latter two terms in the sum reduce to the previous
definitions of the energetic term $\mathcal{E}$ above in (\ref{dlag3_1}),
(\ref{dlag7}), or (\ref{radis1}) for real-valued vector quantities.

We now introduce an important quantity -- the \emph{loss fraction} $\delta
_{R}$. It is defined as the ratio of the rank of the matrix $R$ to the total
degrees of freedom $N$ of the system%
\begin{equation}
\text{loss fraction: \ \ }\delta_{R}=\frac{N_{R}}{N},\text{ \ \ }%
N_{R}=\operatorname{rank}R. \label{radis4a}%
\end{equation}
From our hypothesis $R\not =0$ it follows that $0<\delta_{R}\leq1$. We
consider the dissipative Lagrangian system to be a model of a two-component
composite with a lossy and a lossless component whenever the loss fraction condition%

\begin{equation}
\text{loss fraction condition: \ \ }0<\delta_{R}<1 \label{radis4b}%
\end{equation}
is satisfied. We then associate the range of the operator $R$, i.e.,
$\operatorname{Ran}R$, with the lossy component of the system and consider the
lossy component to be highly lossy when $\beta\gg1$, i.e., the high-loss
regime. Our paper is focused on this case.

\textbf{On the eigenmodes and the quality factor.} To study the dissipative
properties of the dissipative Lagrangian system (\ref{radis2a}) in the
high-loss regime $\beta\gg1$, one can consider the eigenmodes of the system
and their quality factor in this regime.

An eigenmode of this system (\ref{radis2a}) is defined as a nonzero solutions
of the ODEs (\ref{radis2a}) with no forcing, i.e., $F=0$, having the form
$Q\left(  t\right)  =qe^{-\mathrm{i}\zeta t}$, where $\operatorname{Re}\zeta$
and $-\operatorname{Im}\zeta$ are called its frequency and damping factor,
respectively. Since the system is dissipative, i.e., dissipates energy as
interpreted from the energy balance equation (\ref{radis4}), the damping
factor must satisfy
\begin{equation}
0\leq2\mathcal{R}=-\frac{d\mathcal{H}}{dt}=-2\operatorname{Im}\zeta\mathcal{H}
\label{radis4_1}%
\end{equation}
which implies that $-\operatorname{Im}\zeta\geq0$.

Thus, an eigenmode\ is a state of the dissipative Lagrangian system which is
in damped harmonic motion and the motion is oscillatory provided
$\operatorname{Re}\zeta\not =0$. An important quantity which characterizes the
quality of such a damped oscillation is the quality factor $Q_{\zeta}$ (i.e.,
$Q$-factor) of the eigenmode which is defined as the reciprocal of the
relative rate of energy dissipation per temporal cycle, i.e.,%

\begin{equation}
Q_{\zeta}=2\pi\frac{\text{energy stored in system}}{\text{energy lost per
cycle}}=\left\vert \operatorname{Re}\zeta\right\vert \frac{\mathcal{H}}%
{-\frac{d\mathcal{H}}{dt}}. \label{radis4_2}%
\end{equation}
It follows immediately from this definition and (\ref{radis4_1}) that
$Q_{\zeta}$ depends only on the complex frequency $\zeta$ of the eigenmode, in
particular,%
\begin{equation}
Q_{\zeta}=-\frac{1}{2}\frac{\left\vert \operatorname{Re}\zeta\right\vert
}{\operatorname{Im}\zeta}, \label{radis4_3}%
\end{equation}
with the convention $Q_{\zeta}=+\infty$ if $\operatorname{Im}\zeta=0$.

\subsection{The canonical system\label{scansys}}

Now the matrix Hamilton form of the Euler-Lagrange equation involving Rayleigh
dissipative forces and external force (\ref{radis2}) reads%
\begin{equation}
\partial_{t}u=\left(  J-\left[
\begin{array}
[c]{ll}%
\beta R & 0\\
0 & 0
\end{array}
\right]  \right)  M_{\mathrm{H}}u+\left[
\begin{array}
[c]{l}%
F\\
0
\end{array}
\right]  ,\quad u=\left[
\begin{array}
[c]{l}%
P\\
Q
\end{array}
\right]  . \label{radis3}%
\end{equation}

It is important to recognize that the evolution equation (\ref{radis3}) are
not quite yet of the desired form which is the most suitable for the spectral
analysis. Advancing general ideas of the Hamiltonian treatment of dissipative
systems developed in \cite{FigSch2} we can factor the matrix $M_{\mathrm{H}}$
as
\begin{equation}
M_{\mathrm{H}}=K^{\mathrm{T}}K, \label{radis5}%
\end{equation}
where the matrix $K$ is the block matrix%
\begin{gather}
K=\left[
\begin{array}
[c]{ll}%
K_{\mathrm{p}} & 0\\
0 & K_{\mathrm{q}}%
\end{array}
\right]  \left[
\begin{array}
[c]{ll}%
\mathbf{1} & -\theta\\
0 & \mathbf{1}%
\end{array}
\right]  =\left[
\begin{array}
[c]{ll}%
K_{\mathrm{p}} & -K_{\mathrm{p}}\theta\\
0 & K_{\mathrm{q}}%
\end{array}
\right]  ,\label{radis6}\\
K_{\mathrm{p}}=\sqrt{\alpha}^{-1},\text{ }K_{\mathrm{q}}=\sqrt{\eta}
\label{radis6a1}%
\end{gather}
which manifestly takes into account the gyroscopic term $\theta$. Here
$\sqrt{\alpha}$ and $\sqrt{\eta}$ denote the unique positive semidefinite
square roots of the matrices $\alpha$ and $\eta$, respectively. \ In
particular, it follows from the properties (\ref{dlag2}) and the proof of
\cite[\S VI.4, Theorem VI.9]{ReSi1} that $K_{\mathrm{p}}$, $K_{\mathrm{q}}$
are $N\times N$ matrices with real-valued entries with the properties
\begin{equation}
K_{\mathrm{p}}=K_{\mathrm{p}}^{\mathrm{T}}>0,\text{ \ \ }K_{\mathrm{q}%
}=K_{\mathrm{q}}^{\mathrm{T}}\geq0. \label{radis6b}%
\end{equation}

The introduction of the matrix $K$ according to \cite{FigSch2} is intimately
related to the introduction of force variables%
\begin{equation}
v=Ku. \label{radis7}%
\end{equation}
Consequently, recasting the evolution equation (\ref{radis3}) in terms of the
force variables $v$ yields the desired canonical form, namely
\begin{equation}
\partial_{t}v=-\mathrm{i}A\left(  \beta\right)  v+f,\text{ \ \ where }A\left(
\beta\right)  =\Omega-\mathrm{i}\beta B,\text{ \ \ }\beta\geq0, \label{radis8}%
\end{equation}
and%
\begin{gather}
\Omega=\mathrm{i}KJK^{\mathrm{T}}=\left[
\begin{array}
[c]{cc}%
\Omega_{\mathrm{p}} & -\mathrm{i}\Phi^{\mathrm{T}}\\
\mathrm{i}\Phi & 0
\end{array}
\right]  ,\text{ \ }B=K\left[
\begin{array}
[c]{ll}%
R & 0\\
0 & 0
\end{array}
\right]  K^{\mathrm{T}}=\left[
\begin{array}
[c]{ll}%
\mathsf{\tilde{R}} & 0\\
0 & 0
\end{array}
\right]  ,\text{ \ }f=\left[
\begin{array}
[c]{l}%
K_{\mathrm{p}}F\\
0
\end{array}
\right]  ,\label{radis8a}\\
\Omega_{\mathrm{p}}=-\mathrm{i}2K_{\mathrm{p}}\theta K_{\mathrm{p}%
}^{\mathrm{T}},\text{ \ \ }\Phi=K_{\mathrm{q}}K_{\mathrm{p}}^{\mathrm{T}%
},\text{ \ \ }\mathsf{\tilde{R}}=K_{\mathrm{p}}RK_{\mathrm{p}}^{\mathrm{T}}.
\label{radis8b}%
\end{gather}
Observe that the $2N\times2N$ matrices $\Omega$, $B$ in block form are
Hermitian and positive semidefinite, respectively, that is
\begin{equation}
\Omega=\Omega^{\ast},\quad B=B^{\ast}\geq0. \label{radis10}%
\end{equation}
Moreover, the matrix $B$ does not have full rank and its rank is that of $R$,
i.e.
\begin{equation}
\operatorname{rank}B=\operatorname{rank}R=N_{R}. \label{radis10a}%
\end{equation}
The \emph{system operator} $A\left(  \beta\right)  $ has the following
important properties
\begin{equation}
-\operatorname{Im}A\left(  \beta\right)  =\beta B\geq0,\text{ \ }%
\operatorname{Re}A\left(  \beta\right)  =\Omega,\text{ \ }A\left(
\beta\right)  ^{\ast}=-A\left(  \beta\right)  ^{\mathrm{T}}, \label{radis10b}%
\end{equation}
where the latter property comes from the fact that $\Omega^{\ast}%
=\Omega=-\Omega^{\mathrm{T}}$ and $B^{\ast}=B=B^{\mathrm{T}}$. Also, recall
that for any square matrix $M$, one can write $M=\operatorname{Re}%
M+\mathrm{i}\operatorname{Im}M$, where $\operatorname{Re}M=\frac{M+M^{\ast}%
}{2}$ and $\operatorname{Im}M=\frac{M-M^{\ast}}{2\mathrm{i}}$ denote the real
and imaginary parts of the matrix $M$, respectively.

Now by the \emph{canonical system} we mean the system whose state is described
by a time-dependent $v=v\left(  t\right)  $ taking values in the Hilbert space
$H=%
%TCIMACRO{\U{2102} }%
%BeginExpansion
\mathbb{C}
%EndExpansion
^{2N}$ with the standard inner product $\left(  \cdot,\cdot\right)  $ whose
dynamics are governed by the ODEs (\ref{radis8}). These ODEs will be referred
to as the \emph{canonical evolution equation}.

In this paper we will focus on the case when there is no forcing, i.e., $f=0$
(or, equivalently, \thinspace$F=0$). In this case, the canonical system
represents a dissipative dynamical system (since $-\operatorname{Im}A\left(
\beta\right)  =\beta B\geq0$) with the evolution governed by the semigroup
$e^{-\mathrm{i}A\left(  \beta\right)  t}$. As we shall see in this paper and
as mentioned in the introduction, there are some serious advantages in
studying this dissipative dynamical system, i.e., the canonical system, over
the Lagrangian or Hamiltonian systems.

\subsection{The energetic equivalence between the two systems\label{seneq}}

In a previous work \cite{FigWel1} of ours, we studied canonical systems whose
states are solutions of a linear evolution equation in the canonical form
(\ref{radis8}) with system operator $A\left(  \beta\right)  $\thinspace$=$
$\Omega-\mathrm{i}\beta B$, $\beta\geq0$ having exactly the properties
(\ref{radis10}) in which the matrix $B$ did not have full rank. In that paper,
the canonical system was a simplified version of an abstract model of an
oscillator damped retarded friction that modeled a two-component composite
system with a lossy and lossless component. We showed that the state
$v=v\left(  t\right)  $ of such a system satisfied the energy balance equation%
\begin{equation}
\partial_{t}U\left[  v\left(  t\right)  \right]  =-W_{\mathrm{dis}}\left[
v(t)\right]  +W\left[  v\left(  t\right)  \right]  , \label{radis11}%
\end{equation}
in which the system energy $U\left[  v\left(  t\right)  \right]  $, dissipated
power $W_{\mathrm{dis}}\left[  v(t)\right]  $, and $W\left[  v\left(
t\right)  \right]  $ the rate of work done by the force $f\left(  t\right)  $
were given by%

\begin{equation}
U\left[  v\left(  t\right)  \right]  =\frac{1}{2}\left(  v\left(  t\right)
,v\left(  t\right)  \right)  ,\text{ }W_{\mathrm{dis}}\left[  v(t)\right]
=\beta\left(  v\left(  t\right)  ,Bv\left(  t\right)  \right)  ,\text{
}W\left[  v\left(  t\right)  \right]  =\operatorname{Re}\left(  v\left(
t\right)  ,f\left(  t\right)  \right)  . \label{radis11a}%
\end{equation}

An important result of this paper which follows immediately from the block
form (\ref{radis8a}), (\ref{radis8b}) and the relation between the variables
$v$, $u$, $P,Q$ from (\ref{radis7}), (\ref{radis3}), (\ref{dlag6})] is that a
state $Q$ of the Lagrangian system, whose dynamics are governed by
second-order ODEs (\ref{radis2a}), and the corresponding state $v$ of the
canonical system, whose dynamics are governed by the canonical evolution
equation (\ref{radis8}), the energetics are equivalent in the sense%

\begin{equation}
U\left[  v\right]  =\mathcal{T}\left(  \dot{Q},Q\right)  +\mathcal{V}\left(
\dot{Q},Q\right)  ,\text{ \ }W_{\mathrm{dis}}\left[  v\right]  =2\mathcal{R}%
\left(  \dot{Q}\right)  ,\text{ \ }W\left[  v\right]  =\operatorname{Re}%
\left(  \dot{Q},F\right)  \text{,} \label{radis11b}%
\end{equation}
where, as you will recall from the energy balance equation (\ref{radis4}),
$\mathcal{H}=\mathcal{T}+\mathcal{V}$ was the system energy for the Lagrangian system.

\textbf{On the eigenmodes and the quality factor.} We will be interested in
the eigenmodes of the canonical system (\ref{radis8}) and their quality
factor. An eigenmode of the canonical system is defined as a nonzero solutions
of the ODEs (\ref{radis8}) with no forcing, i.e., $f=0$, having the form
$v\left(  t\right)  =we^{-\mathrm{i}\zeta t}$. The quality factor $Q[w]$ of
such an eigenmode is defined as in \cite{FigWel1} by%

\begin{equation}
Q[w]=2\pi\frac{\text{energy stored in system}}{\text{energy lost per cycle}%
}=\left\vert \operatorname{Re}\zeta\right\vert \frac{U\left[  v\left(
t\right)  \right]  }{W_{\mathrm{dis}}\left[  v(t)\right]  }. \label{radis12}%
\end{equation}
By the energy balance equation (\ref{radis11}) it follows that%

\begin{equation}
Q[w]=-\frac{1}{2}\frac{\left\vert \operatorname{Re}\zeta\right\vert
}{\operatorname{Im}\zeta}, \label{radis12_1}%
\end{equation}
with the convention $Q[w]=+\infty$ if $\operatorname{Im}\zeta=0$.

Now given an eigenmode $Q(t)=qe^{-\mathrm{i}\zeta t}$ of the Lagrangian system
\ref{radis2a}), it follows that the corresponding state $v$ of the canonical
system (\ref{radis8}), related by (\ref{radis7}), is an eigenmode of the form
$v(t)=we^{-\mathrm{i}\zeta t}$ (excluding the case $Ku=0$, which can only
occur if $\zeta=0$). Importantly, the energetic equivalence (\ref{radis11b})
holds and their quality factors (\ref{radis4_3}), (\ref{radis12_1}) are equal,
i.e.,%
\begin{equation}
Q\left[  w\right]  =Q_{\zeta}\text{.} \label{radis13}%
\end{equation}

A more detailed discussion on the eigenmodes of the two systems and the
relationship between them can be found in Section \ref{sspan}.

\section{Electric circuit example\label{scircuits}}

One of the important applications of our methods described above is electric
circuits and networks involving resistors representing losses. A general study
of electric networks with losses can be carried out with the help of the
Lagrangian approach, and that systematic study has already been carried out in
this paper. For Lagrangian treatment of electric networks and circuits we
refer to \cite[Sec. 9]{Gantmacher}, \cite[Sec. 2.5]{Goldstein}, \cite{Pars}.

We illustrate the idea and give a flavor of the efficiency of our methods by
considering below a rather simple example of an electric circuit as in Fig.
\ref{Figc1} with the assumptions%
\begin{equation}
L_{1},\text{ }L_{2},\text{ }C_{1},\text{ }C_{2},\text{ }C_{12}>0\text{ and
}R_{2}\geq0\text{.} \label{circ0}%
\end{equation}
This example\ has the essential features of two component systems
incorporating high-loss and lossless components. \begin{figure}[ptb]
\begin{center}
\includegraphics[height=2.6227in, width=4.7823in]{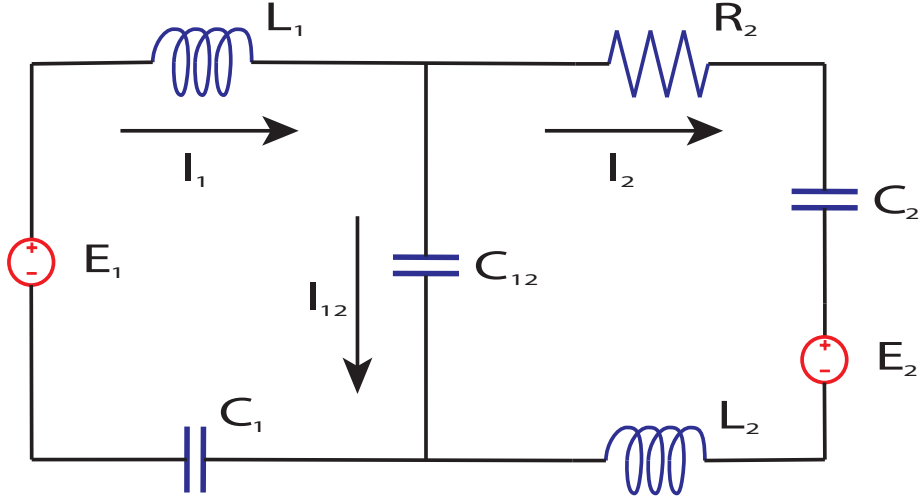}
\end{center}
\caption{An electric circuit involving three capacitances $C_{1}$, $C_{2}$,
$C_{12}$, two inductances $L_{1}$, $L_{2}$, a resistor $R_{2}$, and two
sources $E_{1}$, $E_{2}$. This electric circuit example fits within the
framework of our model.\ Indeed, since resistors represent losses, this two
component system consists of a lossy component and a lossless component -- the
right and left circuits, respectively.}%
\label{Figc1}%
\end{figure}

\textbf{The Lagrangian system.} To derive evolution equations for the electric
circuit in Fig. \ref{Figc1} we use a general method for constructing
Lagrangians for circuits, \cite[Sec. 9]{Gantmacher}, that yields%
\begin{equation}
\mathcal{T}=\frac{L_{1}}{2}\dot{q}_{1}^{2}+\frac{L_{2}}{2}\dot{q}_{2}%
^{2},\quad\mathcal{V}=\frac{1}{2C_{1}}q_{1}^{2}+\frac{1}{2C_{12}}\left(
q_{1}-q_{2}\right)  ^{2}+\frac{1}{2C_{2}}q_{2}^{2},\quad\mathcal{R}%
=\frac{R_{2}}{2}\dot{q}_{2}^{2}, \label{circ1}%
\end{equation}
where $\mathcal{T}$ and $\mathcal{V}$ are respectively the kinetic and the
potential energies, $\mathcal{L}=\mathcal{T}-\mathcal{V}$ is the Lagrangian,
and $\mathcal{R}$ is the Rayleigh dissipative function. Notice that
$I_{1}=\dot{q}_{1}$ and $I_{2}=\dot{q}_{2}$ are the currents. The general
Euler-Lagrange equations of motion with forces are, \cite[Sec. 8]{Gantmacher},%
\begin{equation}
\frac{\partial}{\partial t}\frac{\partial\mathcal{L}}{\partial\dot{Q}}%
-\frac{\partial\mathcal{L}}{\partial Q}=-\frac{\partial\mathcal{R}}%
{\partial\dot{Q}}+F, \label{circ2}%
\end{equation}
where $Q$ are the charges and $F$ the sources%
\begin{equation}
Q=\left[
\begin{array}
[c]{c}%
q_{1}\\
q_{2}%
\end{array}
\right]  ,\quad F=\left[
\begin{array}
[c]{c}%
E_{1}\\
E_{2}%
\end{array}
\right]  , \label{circ3}%
\end{equation}
yielding from (\ref{circ0})--(\ref{circ3}) the following second-order ODEs%
\begin{equation}
\alpha\ddot{Q}+\beta R\dot{Q}+\eta Q=F, \label{circ4}%
\end{equation}
with the dimensionless loss parameter%
\begin{equation}
\beta=\frac{R_{2}}{\ell}\text{ \ \ (where }\ell>0\text{ is fixed and has same
units as }R_{2}\text{)} \label{circ5}%
\end{equation}
that scales the intensity of losses in the system, and%
\begin{equation}
\alpha=\left[
\begin{array}
[c]{cc}%
L_{1} & 0\\
0 & L_{2}%
\end{array}
\right]  ,\text{ \ \ }R=\left[
\begin{array}
[c]{cc}%
0 & 0\\
0 & \ell
\end{array}
\right]  ,\text{ \ \ }\eta=\left[
\begin{array}
[c]{cc}%
\frac{1}{C_{1}}+\frac{1}{C_{12}} & -\frac{1}{C_{12}}\\
-\frac{1}{C_{12}} & \frac{1}{C_{2}}+\frac{1}{C_{12}}%
\end{array}
\right]  . \label{circ6}%
\end{equation}
Recall, the loss fraction $\delta_{R}$ defined in (\ref{radis4a}) is the ratio
of the rank of the matrix $R$ to the total degrees of freedom $N$ of the
system which in this case is
\begin{equation}
\text{loss fraction: \ \ }\delta_{R}=\frac{N_{R}}{N}=\frac{1}{2},\text{
\ \ }N=2,\text{ \ \ }N_{R}=\operatorname{rank}R=1. \label{circ7}%
\end{equation}
Thus the Lagrangian system (\ref{circ4}) has all the properties described in
Sections \ref{slagrangian}, \ref{sLagrModelEigenmQFac} and since the loss
fraction condition (\ref{radis4b}), i.e., $0<\delta_{R}<1$, is satisfied then
it is model of a two-component composite with a lossy and a lossless component.

\textbf{The canonical system.} \ Following the method in Section
\ref{scansys}, we introduce the variables $v$ defined by (\ref{dlag6}),
(\ref{radis3}), (\ref{radis7})\ in terms of the matrix $K$ defined in
(\ref{radis6}), (\ref{radis6a1}) as%
\begin{gather}
v=Ku,\text{ \ \ }u=%
\begin{bmatrix}
\alpha\dot{Q}\\
Q
\end{bmatrix}
,\label{circ8}\\
K=\left[
\begin{array}
[c]{ll}%
K_{\mathrm{p}} & 0\\
0 & K_{\mathrm{q}}%
\end{array}
\right]  ,\text{ \ \ }K_{\mathrm{p}}=\sqrt{\alpha}^{-1}=\left[
\begin{array}
[c]{ll}%
\frac{1}{\sqrt{L_{1}}} & 0\\
0 & \frac{1}{\sqrt{L_{2}}}%
\end{array}
\right]  ,\text{ \ \ }K_{\mathrm{q}}=\sqrt{\eta}.\nonumber
\end{gather}
As $\eta>0$ is a $2\times2$ matrix in this case, then we can compute its
square root explicitly using the formula%
\begin{equation}
\sqrt{M}=\left(  \sqrt{\operatorname{Tr}\left(  M\right)  +2\sqrt{\det\left(
M\right)  }}\right)  ^{-1}\left(  \sqrt{\det\left(  M\right)  }\mathbf{1}%
+M\right)  \text{,} \label{circ9}%
\end{equation}
where $\sqrt{\cdot}$ denotes the positive square root, which holds for any
$2\times2$ matrix $M\geq0$ with $M\not =0$.

Consequently, in the Lagrangian system (\ref{circ4}) in terms of the variable
$v$ is the desired canonical evolution equation from (\ref{radis8}), namely,%

\begin{equation}
\partial_{t}v=-\mathrm{i}A\left(  \beta\right)  v+f,\text{ \ \ where }A\left(
\beta\right)  =\Omega-\mathrm{i}\beta B,\text{ \ \ }\beta\geq0, \label{circ10}%
\end{equation}
and%
\begin{gather}
\Omega=\left[
\begin{array}
[c]{cc}%
0 & -\mathrm{i}\Phi^{\mathrm{T}}\\
\mathrm{i}\Phi & 0
\end{array}
\right]  ,\text{ \ }B=\left[
\begin{array}
[c]{ll}%
\mathsf{\tilde{R}} & 0\\
0 & 0
\end{array}
\right]  ,\text{ \ }f=\left[
\begin{array}
[c]{l}%
K_{\mathrm{p}}F\\
0
\end{array}
\right]  ,\label{circ11}\\
\Phi=K_{\mathrm{q}}K_{\mathrm{p}}^{\mathrm{T}},\text{ \ \ }\mathsf{\tilde{R}%
}=K_{\mathrm{p}}RK_{\mathrm{p}}^{\mathrm{T}}=\left[
\begin{array}
[c]{ll}%
0 & 0\\
0 & \frac{\ell}{L_{2}}%
\end{array}
\right]  .\nonumber
\end{gather}

\section{The virial theorem for dissipative systems and equipartition of
energy\label{svir}}

In this section we will introduce a new virial theorem for dissipative
Lagrangian systems in terms of the eigenmodes of the system. This result
generalizes the virial theorem from classical mechanics \cite[pp. 83-86;
\S 3.4]{Goldstein}, as discussed in Appendix \ref{ApdxVirThm}, for
conservative Lagrangian systems for these types of modes. Recall that the
virial theorem is about the oscillatory transfer of energy from one form to
another and its equipartition.

We can now state precisely and prove our generalization of the virial theorem
to dissipative ($\beta\geq0$) Lagrangian systems\ (\ref{radis2a}) and the
equipartition of energy for the oscillatory eigenmodes. In fact, the proof in
Section \ref{sproofs} shows that our theorem is true for more general system
of ODEs in the form (\ref{radis2a}) since essentially all that really matters
is that the energy balance equation (\ref{radis4}) holds for the eigenmodes.

\begin{theorem}
[virial theorem]\label{tvir}If $Q\left(  t\right)  =qe^{-\mathrm{i}\zeta t}$
is an eigenmode of the Lagrangian system (\ref{radis2a}) then the following
identity holds if $\operatorname{Re}\zeta\not =0$:
\begin{equation}
\mathcal{T}\left(  \dot{Q},Q\right)  =\mathcal{V}\left(  \dot{Q},Q\right)
-\left(  \frac{\operatorname{Im}\zeta}{\operatorname{Re}\zeta}\right)
^{2}\operatorname{Re}\left(  \dot{Q},\theta Q\right)  , \label{tvirs1}%
\end{equation}
where $\mathcal{T}\left(  \dot{Q},Q\right)  $ and $\mathcal{V}\left(  \dot
{Q},Q\right)  $ are the kinetic and potential energy, respectively, defined in
(\ref{radis4_0}). On the other hand, if $\operatorname{Re}\zeta=0\,$\ then the
identity (\ref{tvirs1}) no longer holds and either $\zeta=0$ or the identity
$\left(  \dot{Q},\theta Q\right)  =0$ must hold.
\end{theorem}

\begin{corollary}
[energy equipartition]\label{tvir_eqp}For systems with $\theta=0$, if
$Q\left(  t\right)  =qe^{-\mathrm{i}\zeta t}$ is an eigenmode of the
Lagrangian system (\ref{radis2a}) with $\operatorname{Re}\zeta\not =0$, then
the following identity holds%
\begin{equation}
\mathcal{T}\left(  \dot{Q},Q\right)  =\mathcal{V}\left(  \dot{Q},Q\right)  .
\label{tvir_eqps1}%
\end{equation}
In other words, for the oscillatory eigenmodes there is an equipartition of
the system energy, i.e., $\mathcal{H}=\mathcal{T}+\mathcal{V}$, between their
kinetic energy $\mathcal{T}$ and their potential energy $\mathcal{V}$.
\end{corollary}

\section{Spectral analysis of the system eigenmodes\label{sspan}}

In this section we study time-harmonic solutions to the Euler-Lagrange
equation (\ref{radis2a}) which constitutes a subject of the spectral theory.
An important subject of this section is the study of the relations between
standard spectral theory and the quadratic pencil formulation as it arises
naturally as the time Fourier transformation of the Euler-Lagrange evolution equation.

\subsection{Standard versus pencil formulations of the spectral problems
\label{stvspen}}

As introduced in Section \ref{sLagrModelEigenmQFac}, the eigenmodes of the
Lagrangian system are nonzero solutions of the ODEs (\ref{radis2a})\ with
$F=0$ having the form $Q\left(  t\right)  =qe^{-\mathrm{i}\zeta t}$. For a
fixed $\beta$, these modes correspond to the solutions of the quadratic
eigenvalue problem (QEP)%
\begin{equation}
C(\zeta,\beta)q=0,\qquad q\not =0 \label{pfsp1}%
\end{equation}
for the quadratic matrix pencil in $\zeta$%
\begin{equation}
C(\zeta,\beta)=\zeta^{2}\alpha+\left(  2\theta+\beta R\right)  \mathrm{i}%
\zeta-\eta. \label{pfsp2}%
\end{equation}
The set of eigenvalues (spectrum) of the pencil $C\left(  \cdot,\beta\right)
$ is the set%
\begin{equation}
\sigma\left(  C\left(  \cdot,\beta\right)  \right)  =\left\{  \zeta\in%
%TCIMACRO{\U{2102} }%
%BeginExpansion
\mathbb{C}
%EndExpansion
:\det C(\zeta,\beta)=0\right\}  , \label{pfsp3}%
\end{equation}
which are exactly those values $\zeta$ for which a solution to the QEP
(\ref{pfsp1}) exists. The spectral theory of polynomial operator pencils
\cite{Mar88} can be applied to study the eigenmodes but it has its
disadvantages such as being more complicated than standard spectral theory.
Thus an alternative approach to the spectral theory is desirable.

Often the alternative approach is to use the Hamiltonian system and consider
its eigenmodes, that is, the nonzero solutions of the Hamiltonian equations
(\ref{radis3}) with $F=0$ having the form $u\left(  t\right)  =ue^{-\mathrm{i}%
\zeta t}$. These modes correspond to the solutions of the eigenvalue problem%
\begin{equation}
Mu=-\mathrm{i}\zeta u,\qquad u\not =0 \label{pfsp4}%
\end{equation}
for the matrix%
\begin{equation}
M\left(  \beta\right)  =\left(  J-\left[
\begin{array}
[c]{ll}%
\beta R & 0\\
0 & 0
\end{array}
\right]  \right)  M_{\mathrm{H}}. \label{pfsp5}%
\end{equation}
An advantage to this approach is the simple correspondence via (\ref{dlag6}),
(\ref{radis3}) between the set of modes of the two systems, namely, the
eigenmodes of the Hamiltonian system are solutions of (\ref{radis3}) having
the block form $u\left(  t\right)  =\left[
\begin{array}
[c]{l}%
p\\
q
\end{array}
\right]  e^{-\mathrm{i}\zeta t}$ in which $Q\left(  t\right)  =qe^{-\mathrm{i}%
\zeta t}$ an eigenmode of the Lagrangian system and $p=\left(  -\mathrm{i}%
\zeta\alpha+\theta\right)  q$. In particular, this means the matrix
$\mathrm{i}M\left(  \beta\right)  $ and the pencil $C\left(  \cdot
,\beta\right)  $ have the same eigenvalues and hence the same spectrum, i.e.,%
\begin{equation}
\sigma\left(  \mathrm{i}M\left(  \beta\right)  \right)  =\sigma\left(
C\left(  \cdot,\beta\right)  \right)  . \label{pfsp6}%
\end{equation}
A major disadvantage of this approach is that $M\left(  \beta\right)  $ is a
non-self-adjoint matrix such that the standard theory of self-adjoint or
dissipative operators does not apply without further transformation of the
system, even in the absence of losses, i.e., $\beta=0$. Instead, in this case
it is the Krein spectral theory \cite[\S 42]{Arnold}, \cite{Yak},
\cite{YakSta} that is often used. But this theory is far more complex than the
standard spectral theory and much harder to apply.

Our approach to these spectral problems which overcomes the disadvantages in
the pencil or Krein spectral theory is to use the canonical system and
consider its eigenmodes, that is, the nonzero solutions of the canonical
evolution equation (\ref{radis8}) with $f=0$ having the form $v\left(
t\right)  =we^{-\mathrm{i}\zeta t}$. These modes correspond to the solutions
of the eigenvalue problem%
\begin{equation}
A\left(  \beta\right)  w=\zeta w,\text{ \ \ }w\not =0 \label{pfsp7}%
\end{equation}
for the system operator $A\left(  \beta\right)  =\Omega-\mathrm{i}\beta B$
with Hermitian matrices $\Omega$, $B$ with $B\geq0$. The advantage of this
approach is that the standard spectral theory can be used since $-\mathrm{i}%
A\left(  \beta\right)  $ is a dissipative operator and when losses are absent
$A(0)=\Omega$ is self-adjoint. This is a serious advantage since the spectral
theory is significantly simpler and allows the usage of the deep and effective
results from perturbation theory such as those developed in \cite{FigWel1} to
study the modes of such canonical evolution equations having the general form
(\ref{radis8}). This is the approach we take in this paper to study spectral
symmetries and the dissipative properties of the eigenmodes of the Lagrangian
system such as the modal dichotomy and overdamping phenomenon.

\textbf{Correspondence between spectral problems. }We now conclude this
section by summarizing the correspondence between the two main spectral
problems of this paper, namely, between the standard eigenvalue problem
(\ref{pfsp7}) and the quadratic eigenvalue problem (\ref{pfsp1}). We do this
in the next corollary which uses the following proposition that tells us the
characteristic matrix of the system operator $\zeta\mathbf{1}-A\left(
\beta\right)  $ can be factored in terms of the quadratic matrix pencil
$C\left(  \zeta,\beta\right)  $.

First we introduce some notation that will be useful. \ The Hilbert space $H=%
%TCIMACRO{\U{2102} }%
%BeginExpansion
\mathbb{C}
%EndExpansion
^{2N}$ with standard inner product $\left(  \cdot,\cdot\right)  $ can be
decomposed as $H=H_{\mathrm{p}}\oplus H_{\mathrm{q}}$ into the orthogonal
subspaces $H_{\mathrm{p}}=%
%TCIMACRO{\U{2102} }%
%BeginExpansion
\mathbb{C}
%EndExpansion
^{N}$, $H_{\mathrm{q}}=%
%TCIMACRO{\U{2102} }%
%BeginExpansion
\mathbb{C}
%EndExpansion
^{N}$ with orthogonal matrix projections%
\begin{equation}
P_{\mathrm{p}}=\left[
\begin{array}
[c]{cc}%
\mathbf{1} & 0\\
0 & 0
\end{array}
\right]  ,\text{ \ \ }P_{\mathrm{q}}=\left[
\begin{array}
[c]{cc}%
0 & 0\\
0 & \mathbf{1}%
\end{array}
\right]  . \label{pfsp8a}%
\end{equation}
In particular, the matrices $\Omega$, $B$, and $A\left(  \beta\right)  $
defined in (\ref{radis8})--(\ref{radis8b}) are block matrices already
partitioned with respect to the decomposition $H=H_{\mathrm{p}}\oplus
H_{\mathrm{q}}$ and any vector $w\in H$ can be represented uniquely in the
block form%
\begin{equation}
w=\left[
\begin{array}
[c]{c}%
\varphi\\
\psi
\end{array}
\right]  ,\text{ \ \ where }\varphi=P_{\mathrm{p}}w,\text{ \ \ }%
\psi=P_{\mathrm{q}}w. \label{pfsp8b}%
\end{equation}
Then with respect to this decomposition we have the following results.

\begin{proposition}
\label{ppfsp}If $\zeta\not =0$ then
\begin{equation}
\zeta\mathbf{1}-A\left(  \beta\right)  =\left[
\begin{array}
[c]{cc}%
K_{\mathrm{p}} & \zeta^{-1}\mathrm{i}\Phi^{\mathrm{T}}\\
0 & \mathbf{1}%
\end{array}
\right]  \left[
\begin{array}
[c]{cc}%
\zeta^{-1}\mathbf{1} & 0\\
0 & \zeta\mathbf{1}%
\end{array}
\right]  \left[
\begin{array}
[c]{cc}%
C(\zeta,\beta) & 0\\
0 & \mathbf{1}%
\end{array}
\right]  \left[
\begin{array}
[c]{cc}%
K_{\mathrm{p}}^{\mathrm{T}} & 0\\
-\zeta^{-1}\mathrm{i}\Phi & \mathbf{1}%
\end{array}
\right]  . \label{pfsp8}%
\end{equation}

\end{proposition}

\begin{corollary}
[spectral equivalence]\label{cpfsp}For any $\zeta\in%
%TCIMACRO{\U{2102} }%
%BeginExpansion
\mathbb{C}
%EndExpansion
$,
\begin{equation}
\det\left(  \zeta\mathbf{1}-A\left(  \beta\right)  \right)  =\frac{\det
C(\zeta,\beta)}{\det\alpha}. \label{cpfsp1}%
\end{equation}
In particular, the system operator $A\left(  \beta\right)  $ and quadratic
matrix pencil $C(\zeta,\beta)$ have the same spectrum, i.e.,
\begin{equation}
\sigma\left(  A\left(  \beta\right)  \right)  =\sigma\left(  C\left(
\cdot,\beta\right)  \right)  . \label{cpfsp2}%
\end{equation}
Moreover, if $\zeta\not =0$ then the following statements are true:

\begin{enumerate}
\item If $A\left(  \beta\right)  w=\zeta w$ and $w\not =0$ then
\begin{equation}
w=\left[
\begin{array}
[c]{c}%
-\mathrm{i}\zeta\sqrt{\alpha}q\\
\sqrt{\eta}q
\end{array}
\right]  ,\text{ \ \ where }C(\zeta,\beta)q=0,\text{ \ \ }q\not =0.
\label{cpfsp3}%
\end{equation}

\item If $C(\zeta,\beta)q=0$ and $q\not =0$ then
\begin{equation}
A\left(  \beta\right)  w=\zeta w\text{, \ \ where }w=\left[
\begin{array}
[c]{c}%
-\mathrm{i}\zeta\sqrt{\alpha}q\\
\sqrt{\eta}q
\end{array}
\right]  \not =0. \label{cpfsp4}%
\end{equation}

\end{enumerate}
\end{corollary}

\subsection{On the spectrum of the system operator\label{sspso}}

In Section \ref{stvspen} it was shown that the study of the eigenmodes of the
Lagrangian system (\ref{radis2a}) reduces to the quadratic eigenvalue problem
(\ref{pfsp1}) and this motivates a study of the spectrum $\sigma\left(
C\left(  \cdot,\beta\right)  \right)  $ of the quadratic matrix pencil
$C(\zeta,\beta)$ in (\ref{pfsp2}). Corollary \ref{cpfsp} makes it clear that
we can instead study the eigenmodes of the canonical system (\ref{radis8})
and, in particular, the system operator $A\left(  \beta\right)  $ spectrum
satisfies $\sigma\left(  A\left(  \beta\right)  \right)  =\sigma\left(
C\left(  \cdot,\beta\right)  \right)  $. The purpose of this section is to
give a detailed analysis of the set $\sigma\left(  A\left(  \beta\right)
\right)  $.

Recall, the system operator $A\left(  \beta\right)  =\Omega-\mathrm{i}\beta B
$, $\beta\geq0$ from (\ref{radis8}) with $2N\times2N$ matrices $\Omega$, $B$
has the fundamental properties (\ref{radis10a}), (\ref{radis10b})%
\begin{gather}
-\operatorname{Im}A\left(  \beta\right)  =\beta B\geq0,\text{ \ }%
\operatorname{Re}A\left(  \beta\right)  =\Omega,\text{ \ }A\left(
\beta\right)  ^{\ast}=-A\left(  \beta\right)  ^{\mathrm{T}},\label{sasp1}\\
0<\operatorname{rank}B=N_{R}\leq N, \label{ebmd1}%
\end{gather}
where $N_{R}=\operatorname{rank}R$. These properties are particularly
important in describing the spectrum $\sigma\left(  A\left(  \beta\right)
\right)  $ of the system operator $A\left(  \beta\right)  $.

For instance, the next proposition on the spectrum for nondissipative
($\beta=0$) Lagrangian systems (\ref{radis2a}) follows immediately from these
properties which would otherwise not be exactly obvious for gyroscopic systems
(i.e., $\theta\not =0$).

\begin{proposition}
[real eigenfrequencies]For nondissipative Lagrangian systems (\ref{radis2a}),
that is when $\beta=0 $, all the eigenfrequencies $\zeta$ are real, and
consequently any eigenmode evolution is of the form $Q\left(  t\right)
=qe^{-\mathrm{i}\zeta t} $ with $\operatorname{Im}\zeta=0$.
\end{proposition}

\begin{proof}
If $\beta=0$ (i.e., no dissipation) then $A\left(  0\right)  =\Omega$ is a
Hermitian matrix. Thus, the spectrum $\sigma\left(  A\left(  0\right)
\right)  $ is a subset of $%
%TCIMACRO{\U{211d} }%
%BeginExpansion
\mathbb{R}
%EndExpansion
$. Hence, if $Q\left(  t\right)  =qe^{-\mathrm{i}\zeta t}$ is an eigenmode of
the Lagrangian systems (\ref{radis2a}) with $\beta=0$ then by Corollary
\ref{cpfsp} we have $\zeta\in\sigma\left(  A\left(  0\right)  \right)  $ and
so $\operatorname{Im}\zeta=0$. This completes the proof.
\end{proof}

In the next few sections we will give a deeper analysis of the spectral
properties of $A\left(  \beta\right)  $ including spectral symmetries in
Section \ref{sspsy} and the modal dichotomy in Sections \ref{sevmd},
\ref{smdhlr}. In order to do so we must first introduce some notation. In the
Hilbert space $H=%
%TCIMACRO{\U{2102} }%
%BeginExpansion
\mathbb{C}
%EndExpansion
^{2N}$, denote by $b_{j}$, $j=1,\ldots,N_{R}$ the nonzero eigenvalues of $B$
(counting multiplicities) with the smallest denoted by%

\begin{equation}
b_{\text{min}}=\min_{1\leq j\leq N_{R}}b_{j}. \label{sasp2}%
\end{equation}
In particular, the spectrum of $B$ is
\begin{equation}
\sigma\left(  B\right)  =\left\{  b_{0},b_{1},\ldots,b_{N_{R}}\right\}  ,
\label{sasp2a}%
\end{equation}
where $b_{0}=0$.

Denote the largest eigenvalue of $\Omega$ by $\omega_{\text{max}}$. It follows
from the fact that $\Omega$ is a Hermitian matrix which is skew-symmetric that%

\begin{equation}
\omega_{\text{max}}=\left\Vert \Omega\right\Vert , \label{sasp3}%
\end{equation}
where $\left\Vert \cdot\right\Vert $ denotes the operator norm on square matrices.

\subsubsection{Spectral symmetry\label{sspsy}}

The next proposition describes the spectral symmetries of the system operator
$A$ which follow from the property $A\left(  \beta\right)  ^{\ast}=-A\left(
\beta\right)  ^{\mathrm{T}}$.

\begin{proposition}
[spectral symmetry]\label{pssym}The following statements are true:

\begin{enumerate}
\item The characteristic polynomial of $A\left(  \beta\right)  $ satisfies
\begin{equation}
\overline{\det\left(  -\overline{\zeta}\mathbf{1}-A\left(  \beta\right)
\right)  }=\det\left(  \zeta I-A\left(  \beta\right)  \right)  \label{pssym1}%
\end{equation}
for every $\zeta\in%
%TCIMACRO{\U{2102} }%
%BeginExpansion
\mathbb{C}
%EndExpansion
$. \ In particular, the spectrum $\sigma\left(  A\right)  $ of the system
operator $A$ has the symmetry%
\begin{equation}
\sigma\left(  A\left(  \beta\right)  \right)  =-\overline{\sigma\left(
A\left(  \beta\right)  \right)  }\text{.} \label{pssym2}%
\end{equation}

\item If $w$ is an eigenvector of the system operator $A$ with corresponding
eigenvalue $\zeta$ then $\overline{w}$ is an eigenvector of $A$ with
corresponding eigenvalue $-\overline{\zeta}$.

\item If $\beta=0$ (i.e., no dissipation) then $\det\left(  -\zeta
\mathbf{1}-A\left(  0\right)  \right)  =$ $\det\left(  \zeta I-A\left(
0\right)  \right)  $ for every $\zeta\in%
%TCIMACRO{\U{2102} }%
%BeginExpansion
\mathbb{C}
%EndExpansion
$.
\end{enumerate}
\end{proposition}

\subsubsection{Eigenvalue bounds and modal dichotomy\label{sevmd}}

We will denote the discs centered at the eigenvalues of $-\mathrm{i}\beta B$
with radius $\omega_{\text{max}}$ by%
\begin{equation}
D_{j}\left(  \beta\right)  =\left\{  \zeta\in%
%TCIMACRO{\U{2102} }%
%BeginExpansion
\mathbb{C}
%EndExpansion
:\left\vert \zeta-(-\mathrm{i}\beta b_{j})\right\vert \leq\omega_{\text{max}%
}\right\}  ,\text{ }0\leq j\leq N_{R}. \label{mdss0}%
\end{equation}
Two\ subsets of the spectrum $\sigma\left(  A\left(  \beta\right)  \right)  $
of the system operator $A=$ $\Omega-\mathrm{i}\beta B$ which play a central
role in our analysis are%
\begin{align}
\sigma_{0}\left(  A\left(  \beta\right)  \right)   &  =\sigma\left(  A\left(
\beta\right)  \right)  \cap D_{0}\left(  \beta\right)  ,\label{mdss}\\
\sigma_{1}\left(  A\left(  \beta\right)  \right)   &  =\sigma\left(  A\left(
\beta\right)  \right)  \cap\cup_{j=1}^{N_{R}}D_{j}\left(  \beta\right)
.\nonumber
\end{align}

\begin{proposition}
[eigenvalue bounds]\label{pevbd}The following statements are true:

\begin{enumerate}
\item The eigenvalues of the system operator $A\left(  \beta\right)  $ lie in
the union of the closed discs whose centers are the eigenvalues of
$-\mathrm{i}\beta B$ with radius $\omega_{\text{max}}$, that is,%
\begin{equation}
\sigma\left(  A\left(  \beta\right)  \right)  =\sigma_{0}\left(  A\left(
\beta\right)  \right)  \cup\sigma_{1}\left(  A\left(  \beta\right)  \right)  .
\label{pevbd1}%
\end{equation}

\item If $w\not =0$ and $A\left(  \beta\right)  w=\zeta w$ then%
\begin{equation}
\operatorname{Re}\zeta=\frac{\left(  w,\Omega w\right)  }{\left(  w,w\right)
},\text{ \ \ }-\operatorname{Im}\zeta=\beta\frac{\left(  w,Bw\right)
}{\left(  w,w\right)  }\geq0\text{.} \label{pevbd2}%
\end{equation}

\item If $\zeta$ is an eigenvalue of $A\left(  \beta\right)  $ and $\left\vert
\zeta\right\vert >\omega_{\text{max}}$ then
\begin{equation}
-\operatorname{Im}\zeta\geq\beta b_{\text{min}}-\omega_{\text{max}}\text{.}
\label{pevbd3}%
\end{equation}

\end{enumerate}
\end{proposition}

\begin{corollary}
[spectral clustering]\label{cssym}The eigenvalues of the system operator
$A\left(  \beta\right)  =\Omega-\mathrm{i}\beta B$, $\beta\geq0$ lie in the
closed lower half of the complex plane, are symmetric with respect to the
imaginary axis, and lie in the union of the closed discs whose centers are the
eigenvalues of $-\mathrm{i}\beta B$ with radius $\omega_{\text{max}}$.
\ Moreover, if $\beta=0$ (i.e., no dissipation) then the eigenvalues of
$A\left(  0\right)  =\Omega$ are real and symmetric with respect to the origin.
\end{corollary}

\begin{theorem}
[modal dichotomy I]\label{tmdic}If $\beta>2\frac{\omega_{\text{max}}%
}{b_{\text{min}}}$ then
\begin{equation}
\sigma\left(  A\left(  \beta\right)  \right)  =\sigma_{0}\left(  A\left(
\beta\right)  \right)  \cup\sigma_{1}\left(  A\left(  \beta\right)  \right)
,\text{ \ \ }\sigma_{0}\left(  A\left(  \beta\right)  \right)  \cap\sigma
_{1}\left(  A\left(  \beta\right)  \right)  =\emptyset. \label{tmdic1}%
\end{equation}
Furthermore, there exists unique\ invariant subspaces $H_{\ell\ell}\left(
\beta\right)  $, $H_{h\ell}\left(  \beta\right)  $ of the system operator
$A\left(  \beta\right)  =$ $\Omega-\mathrm{i}\beta B$ with the properties%
\begin{align}
&  (i)\text{ \ \ }H=H_{\ell\ell}\left(  \beta\right)  \oplus H_{h\ell}\left(
\beta\right)  ;\label{tmdic2}\\
&  (ii)\text{ \ \ }\sigma\left(  A\left(  \beta\right)  |_{H_{\ell\ell}\left(
\beta\right)  }\right)  =\sigma_{0}\left(  A\left(  \beta\right)  \right)
,\text{ \ \ }\sigma\left(  A\left(  \beta\right)  |_{H_{h\ell}\left(
\beta\right)  }\right)  =\sigma_{1}\left(  A\left(  \beta\right)  \right)
,\nonumber
\end{align}
where $H=%
%TCIMACRO{\U{2102} }%
%BeginExpansion
\mathbb{C}
%EndExpansion
^{2N}$. Moreover, the dimensions of these subspaces satisfy%
\begin{equation}
\dim H_{h\ell}\left(  \beta\right)  =N_{R},\ \ \dim H_{\ell\ell}\left(
\beta\right)  =2N-N_{R}. \label{tmdic3}%
\end{equation}

\end{theorem}

\begin{definition}
[high-loss susceptible subspace]\label{dmdic}For the system operator $A\left(
\beta\right)  =$ $\Omega-\mathrm{i}\beta B$ with $\beta>2\frac{\omega
_{\text{max}}}{b_{\text{min}}}$ we will call its $N_{R}$-dimensional invariant
subspace $H_{h\ell}\left(  \beta\right)  $ the high-loss susceptible subspace.
\ We will call its $(2N-N_{R})$-dimensional invariant subspace $H_{\ell\ell
}\left(  \beta\right)  $ the low-loss susceptible subspace.
\end{definition}

Our reasoning for the definitions of these subspaces is clarified with the
following corollary.

\begin{corollary}
[high-loss subspace: dissipative properties]\label{cmdic}If $\beta
>2\frac{\omega_{\text{max}}}{b_{\text{min}}}$ then%
\begin{align}
\sigma\left(  A\left(  \beta\right)  |_{H_{\ell\ell}\left(  \beta\right)
}\right)   &  =\left\{  \zeta\in\sigma\left(  A\left(  \beta\right)  \right)
:0\leq-\operatorname{Im}\zeta\leq\omega_{\text{max}}\right\}  , \label{cmdic1}%
\\
\text{\ }\sigma\left(  A\left(  \beta\right)  |_{H_{h\ell}\left(
\beta\right)  }\right)   &  =\left\{  \zeta\in\sigma\left(  A\left(
\beta\right)  \right)  :-\operatorname{Im}\zeta\geq\beta b_{\text{min}}%
-\omega_{\text{max}}>\omega_{\text{max}}\right\}  .\nonumber
\end{align}
Furthermore, the quality factor (\ref{radis12_1}) of any eigenmode of the
canonical system (\ref{radis8}) in the high-loss susceptible subspace
$H_{h\ell}\left(  \beta\right)  $ satisfies
\begin{equation}
0\leq\max\limits_{\substack{w\text{ an eigenvector}\\\text{of }A\left(
\beta\right)  \text{ in }H_{h\ell}\left(  \beta\right)  }}Q[w]\leq\frac{1}%
{2}\frac{\omega_{\text{max}}}{\beta b_{\text{min}}-\omega_{\text{max}}}%
<\frac{1}{2}. \label{cmdic2}%
\end{equation}
In particular, as the losses go to $\infty$ the damping factor and quality
factor of any such eigenmode goes to $+\infty$ and $0$, respectively, that is,%
\begin{equation}
\lim_{\beta\rightarrow\infty}\min_{\zeta\in\sigma\left(  A\left(
\beta\right)  |_{H_{h\ell}\left(  \beta\right)  }\right)  }\left(
-\operatorname{Im}\zeta\right)  =+\infty,\text{ \ \ }\lim_{\beta
\rightarrow\infty}\max\limits_{\substack{w\text{ an eigenvector}\\\text{of
}A\left(  \beta\right)  \text{ in }H_{h\ell}\left(  \beta\right)  }}Q[w]=0.
\label{cmdic3}%
\end{equation}

\end{corollary}

We conclude this section with the following remarks. The spectrum of $A\left(
\beta\right)  $ restricted to the low-loss susceptible subspace $H_{\ell\ell
}\left(  \beta\right)  $, that is, the set $\sigma\left(  A\left(
\beta\right)  |_{H_{\ell\ell}\left(  \beta\right)  }\right)  $ in
(\ref{cmdic1}), is close to the real axis and actually coalesces to a finite
set of real numbers as losses $\beta\rightarrow\infty$. This statement is made
precise in the following section with Theorem \ref{cpodr} and the asymptotic
expansions in (\ref{pod14}). On the other hand, results on quality factor for
the eigenmodes in $H_{\ell\ell}\left(  \beta\right)  $ is far more subtle than
the results in Corollary \ref{cmdic} for the eigenmodes in the high-loss
susceptible subspace $H_{h\ell}\left(  \beta\right)  $. For
gyroscopic-dissipative systems considered in this paper, Corollary \ref{cmdic}
above and Proposition \ref{pqaollm} below give a partial description of the
nature of the quality factor for large losses, i.e., $\beta\gg1$, in the lossy
component of the composite system. When gyroscopy is absent, i.e., $\theta=0$,
then a more complete analysis for quality factor can be carried out, which we
have done in Section \ref{sodan} of this paper in connection to our studies on
the overdamping phenomenon. Consideration of quality factor and overdamping in
dissipative systems with gyroscopy, however, requires a more subtle and
detailed analysis that will be considered in a future work.

\subsubsection{Modal dichotomy in the high-loss regime\label{smdhlr}}

We are interested in describing the spectrum $\sigma\left(  A\left(
\beta\right)  \right)  $ of the system operator $A\left(  \beta\right)
=\Omega-\mathrm{i}\beta B$, $\beta\geq0$ in the high-loss regime, i.e.,
$\beta\gg1$. We do this in this section by giving an asymptotic
characterization, as $\beta\rightarrow\infty$, of the modal dichotomy as
described in Theorem \ref{tmdic} and Corollary \ref{cmdic}. \ In order to do
so we need to give a spectral perturbation analysis of the matrix $A\left(
\beta\right)  $ as $\beta\rightarrow\infty$. \ Fortunately, this analysis has
already been carried out in \cite{FigWel1}. We now introduce the necessary
notion and describe the results.

The Hilbert space $H=%
%TCIMACRO{\U{2102} }%
%BeginExpansion
\mathbb{C}
%EndExpansion
^{2N}$ with standard inner product $\left(  \cdot,\cdot\right)  $ is
decomposed into the direct sum of orthogonal invariant subspaces of the
operator $B$, namely
\begin{equation}
H=H_{B}\oplus H_{B}^{\bot},\text{ \ \ }\dim H_{B}=N_{R}, \label{pod5}%
\end{equation}
where $H_{B}=\operatorname{Ran}B$ (the range of $B$) is the loss subspace of
dimension $N_{R}=\operatorname{rank}B$ with orthogonal projection $P_{B}$ and
its orthogonal complement, $H_{B}^{\bot}=\operatorname{Ker}B$ (the nullspace
of $B$), is the no-loss subspace of dimension $2N-N_{R}$ with orthogonal
projection $P_{B}^{\bot}$.

The operators $\Omega$ and $B$ with respect to the direct sum (\ref{pod5}) are
the $2\times2$ block operator matrices%
\begin{equation}
\Omega=\left[
\begin{array}
[c]{cc}%
\Omega_{2} & \Theta\\
\Theta^{\ast} & \Omega_{1}%
\end{array}
\right]  ,\quad B=\left[
\begin{array}
[c]{cc}%
B_{2} & 0\\
0 & 0
\end{array}
\right]  , \label{pod7}%
\end{equation}
where $\Omega_{2}=\left.  P_{B}\Omega P_{B}\right\vert _{H_{B}}:H_{B}%
\rightarrow H_{B}$ and $B_{2}=\left.  P_{B}BP_{B}\right\vert _{H_{B}}%
:H_{B}\rightarrow H_{B}$ are restrictions of the operators $\Omega$ and $B$
respectively to loss subspace $H_{B}$ whereas $\Omega_{1}=\left.  P_{B}^{\bot
}\Omega P_{B}^{\bot}\right\vert _{H_{B}^{\bot}}:H_{B}^{\bot}\rightarrow
H_{B}^{\bot}$ is the restriction of $\Omega$ to complementary subspace
$H_{B}^{\bot}$. \ Also, $\Theta:H_{B}^{\bot}\rightarrow H_{B}$ is the operator
$\Theta=\left.  P_{B}\Omega P_{B}^{\bot}\right\vert _{H_{B}^{\bot}}$ whose
adjoint is given by $\Theta^{\ast}=\left.  P_{B}^{\bot}\Omega P_{B}\right\vert
_{H_{B}}:H_{B}\rightarrow H_{B}^{\bot}$.

The perturbation analysis in the high-loss regime $\beta\gg1$ for the system
operator $A(\beta)=\Omega-\mathrm{i}\beta B$ described in \cite[\S VI.A,
Theorem 5 \& Proposition 11]{FigWel1} introduces an orthonormal basis
$\left\{  \mathring{w}_{j}\right\}  _{j=1}^{2N}$ diagonalizing the
self-adjoint operators $\Omega_{1}$ and $B_{2}>0$ from (\ref{pod7}) with%
\begin{equation}
B_{2}\mathring{w}_{j}=b_{j}\mathring{w}_{j}\text{ for }1\leq j\leq
N_{R};\text{ \ \ }\Omega_{1}\mathring{w}_{j}=\rho_{j}\mathring{w}_{j}\text{
for }N_{R}+1\leq j\leq2N, \label{pod8}%
\end{equation}
where%
\begin{align}
b_{j}  &  =\left(  \mathring{w}_{j},B_{2}\mathring{w}_{j}\right)  =\left(
\mathring{w}_{j},B\mathring{w}_{j}\right)  \text{ for }1\leq j\leq
N_{R};\label{pod9}\\
\rho_{j}  &  =\left(  \mathring{w}_{j},\Omega_{1}\mathring{w}_{j}\right)
=\left(  \mathring{w}_{j},\Omega\mathring{w}_{j}\right)  \text{ for }%
N_{R}+1\leq j\leq2N.\nonumber
\end{align}
Then for $\beta\gg1$ the system operator $A(\beta)$ is diagonalizable with
basis of eigenvectors $\left\{  w_{j}\left(  \beta\right)  \right\}
_{j=1}^{2N}$ satisfying%
\begin{equation}
A\left(  \beta\right)  w_{j}\left(  \beta\right)  =\zeta_{j}\left(
\beta\right)  w_{j}\left(  \beta\right)  ,\text{$\quad$}1\leq j\leq2N,\text{
\ \ }\beta\gg1 \label{pod10}%
\end{equation}
which split into two distinct classes
\begin{gather}
\text{high-loss}\text{:$\quad$}\zeta_{j}\left(  \beta\right)  ,\text{ }%
w_{j}\left(  \beta\right)  ,\text{$\quad$}1\leq j\leq N_{R};\label{pod11}\\
\text{low-loss}\text{:$\quad$}\zeta_{j}\left(  \beta\right)  ,\text{ }%
w_{j}\left(  \beta\right)  ,\text{$\quad$}N_{R}+1\leq j\leq2N,\nonumber
\end{gather}
with the following properties.

\textbf{The high-loss class:} the eigenvalues have poles at $\beta=\infty$
whereas their eigenvectors are analytic at $\beta=\infty$, having the
asymptotic expansions%
\begin{equation}
\zeta_{j}\left(  \beta\right)  =-\mathrm{i}b_{j}\beta+\rho_{j}+O\left(
\beta^{-1}\right)  ,\text{ }b_{j}>0,\text{ }\rho_{j}\in%
%TCIMACRO{\U{211d} }%
%BeginExpansion
\mathbb{R}
%EndExpansion
,\text{ }w_{j}\left(  \beta\right)  =\mathring{w}_{j}+O\left(  \beta
^{-1}\right)  ,\text{ }1\leq j\leq N_{R}. \label{pod12}%
\end{equation}
The vectors $\mathring{w}_{j}$, $1\leq j\leq N_{R}$ form an orthonormal basis
of the loss subspace $H_{B}$ and%
\begin{equation}
B\mathring{w}_{j}=b_{j}\mathring{w}_{j},\text{$\quad$}\rho_{j}=\left(
\mathring{w}_{j},\Omega\mathring{w}_{j}\right)  ,\text{ for }1\leq j\leq
N_{R}. \label{pod13}%
\end{equation}
In particular, $b_{j}$, $j=1,\ldots,N_{R}$ are all the nonzero eigenvalues of
$B$ (counting multiplicities).

\textbf{The low-loss class:} the eigenvalues and eigenvectors are analytic at
$\beta=\infty$, having the asymptotic expansions%
\begin{align}
\zeta_{j}\left(  \beta\right)   &  =\rho_{j}-\mathrm{i}d_{j}\beta
^{-1}+O\left(  \beta^{-2}\right)  ,\text{ \ \ }\rho_{j}\in%
%TCIMACRO{\U{211d} }%
%BeginExpansion
\mathbb{R}
%EndExpansion
,\text{ $\ \ d_{j}$}\geq0,\label{pod14}\\
w_{j}\left(  \beta\right)   &  =\mathring{w}_{j}+O\left(  \beta^{-1}\right)
,\text{ \ \ }N_{R}+1\leq j\leq2N.\nonumber
\end{align}
The vectors $\mathring{w}_{j}$, $N_{R}+1\leq j\leq2N$ form an orthonormal
basis of the no-loss subspace $H_{B}^{\perp}$ and%
\begin{equation}
B\mathring{w}_{j}=0,\text{$\quad$}\rho_{j}=\left(  \mathring{w}_{j}%
,\Omega\mathring{w}_{j}\right)  ,\text{$\quad$}d_{j}=\left(  \mathring{w}%
_{j},\Theta^{\ast}B_{2}^{-1}\Theta\mathring{w}_{j}\right)  \text{ for }%
N_{R}+1\leq j\leq2N. \label{pod15}%
\end{equation}

By \cite[\S VI.A, Proposition 7]{FigWel1} we know that the asymptotic formulas
for the real and imaginary parts of the complex eigenvalues $\zeta_{j}(\beta)$
as $\beta\rightarrow\infty$ are given by
\begin{align}
\text{high-loss}  &  \text{: \ }\operatorname{Re}\zeta_{j}\left(
\beta\right)  =\rho_{j}+O\left(  \beta^{-2}\right)  ,\text{ \ \ }%
\operatorname{Im}\zeta_{j}\left(  \beta\right)  =-b_{j}\beta+O\left(
\beta^{-1}\right)  ,\text{ \ \ }1\leq j\leq N_{R};\label{pod16}\\
\text{low-loss}  &  \text{: }\operatorname{Re}\zeta_{j}\left(  \beta\right)
=\rho_{j}+O\left(  \beta^{-2}\right)  ,\text{ }\operatorname{Im}\zeta
_{j}\left(  \beta\right)  =-d_{j}\beta^{-1}+O\left(  \beta^{-3}\right)
,\text{ }N_{R}+1\leq j\leq2N.\nonumber
\end{align}

Observe that the expansions (\ref{pod16}) imply%

\begin{equation}
\lim_{\beta\rightarrow\infty}\operatorname{Im}\zeta_{j}\left(  \beta\right)
=-\infty\text{ for }1\leq j\leq N_{R};\text{ \ \ }\lim_{\beta\rightarrow
\infty}\operatorname{Im}\zeta_{j}\left(  \beta\right)  =0\text{ for }%
N_{R}+1\leq j\leq2N, \label{pod17}%
\end{equation}
justifying the names high-loss and low-loss.

The following theorem is the goal of this section. It characterizes the
spectrum $\sigma\left(  A\left(  \beta\right)  \right)  $ of the system
operator $A\left(  \beta\right)  =\Omega-\mathrm{i}\beta B$ and the modal
dichotomy from Theorem \ref{tmdic} and Corollary \ref{cmdic} in the high-loss
regime $\beta\gg1$ in terms of the high-loss and low-loss eigenvectors.

\begin{theorem}
[modal dichotomy II]\label{cpodr}For $\beta$ sufficiently large, the modal
dichotomy occurs as in Theorem \ref{tmdic} and Corollary \ref{cmdic} with the
following equalities holding:
\begin{align}
\sigma\left(  A\left(  \beta\right)  |_{H_{\ell\ell}\left(  \beta\right)
}\right)   &  =\left\{  \zeta_{j}\left(  \beta\right)  :N_{R}+1\leq
j\leq2N\right\}  ,\label{pod18}\\
\sigma\left(  A\left(  \beta\right)  |_{H_{h\ell}\left(  \beta\right)
}\right)   &  =\left\{  \zeta_{j}\left(  \beta\right)  :1\leq j\leq
N_{R}\right\}  ,\nonumber
\end{align}
and%
\begin{align}
H_{\ell\ell}\left(  \beta\right)   &  =\operatorname*{span}\left\{
w_{j}\left(  \beta\right)  :N_{R}+1\leq j\leq2N\right\}  ,\label{pod19}\\
H_{h\ell}\left(  \beta\right)   &  =\operatorname*{span}\left\{  w_{j}\left(
\beta\right)  :1\leq j\leq N_{R}\right\}  .\nonumber
\end{align}
In particular, $H_{h\ell}\left(  \beta\right)  $ and $H_{\ell\ell}\left(
\beta\right)  $, the high-loss and low-loss susceptible subspaces of $A\left(
\beta\right)  $, respectively, have as a basis the high-loss eigenvectors
$\left\{  w_{j}\left(  \beta\right)  \right\}  _{j=1}^{N_{R}}$ and the
low-loss eigenvectors $\left\{  w_{j}\left(  \beta\right)  \right\}
_{j=N_{R}+1}^{2N}$, respectively.
\end{theorem}

In the next section we will study overdamping phenomena for Lagrangian systems
(\ref{radis2a}). An important role in our analysis will be played by the
eigenmodes $v_{j}\left(  t,\beta\right)  =w_{j}\left(  \beta\right)
e^{-\mathrm{i}\zeta_{j}\left(  \beta\right)  t}$, $1\leq j\leq2N$ of the
canonical system (\ref{radis8}) which by the modal dichotomy split into the
two distinct classes based on their dissipative properties:%
\begin{align}
\text{high-loss}  &  \text{: \ \ }v_{j}\left(  t,\beta\right)  =w_{j}\left(
\beta\right)  e^{-\mathrm{i}\zeta_{j}\left(  \beta\right)  t},\text{
\ \ }1\leq j\leq N_{R};\label{pod20}\\
\text{low-loss}  &  \text{: \ \ }v_{j}\left(  t,\beta\right)  =w_{j}\left(
\beta\right)  e^{-\mathrm{i}\zeta_{j}\left(  \beta\right)  t},\text{
\ \ }N_{R}+1\leq j\leq2N.\nonumber
\end{align}

It should be emphasized here that the classification of these modes into
high-loss and low-loss is based solely on the behavior of their damping factor
in (\ref{pod17}) as losses become large, i.e., as $\beta\rightarrow\infty$,
and not necessarily on their quality factor which we will discuss at the end
of this section. Moreover, the damping factor is related to energy loss by the
energy balance equation (\ref{radis11}), (\ref{radis11a}) satisfied by these
modes which implies the dissipated power is%
\begin{align}
-\partial_{t}U[v_{j}\left(  t,\beta\right)  ]  &  =-2\operatorname{Im}%
\zeta_{j}\left(  \beta\right)  U[v_{j}\left(  t,\beta\right)  ],\text{
\ \ }1\leq j\leq2N,\label{pod20a}\\
U[v_{j}\left(  t,\beta\right)  ]  &  =\frac{1}{2}\left(  w_{j}\left(
\beta\right)  ,w_{j}\left(  \beta\right)  \right)  e^{2\operatorname{Im}%
\zeta_{j}\left(  \beta\right)  t}\nonumber
\end{align}
where $U[v_{j}\left(  t,\beta\right)  ]$ defined in (\ref{radis11a}) was
interpreted as the system energy. In particular, by (\ref{pod17}),
(\ref{pod20a}), and the fact $\left(  \mathring{w}_{j},\mathring{w}%
_{j}\right)  =1$ it follows that their system energy for any \emph{fixed}
$t>0$ satisfies%
\begin{align}
\text{high-loss}  &  \text{: \ \ }\lim_{\beta\rightarrow\infty}U[v_{j}\left(
t,\beta\right)  ]=0,\text{ \ \ }1\leq j\leq N_{R};\label{pod20b}\\
\text{low-loss}  &  \text{: \ \ }\lim_{\beta\rightarrow\infty}U[v_{j}\left(
t,\beta\right)  ]=\frac{1}{2},\text{ \ \ }N_{R}+1\leq j\leq2N.\nonumber
\end{align}
This combined with (\ref{pod17}), (\ref{pod20a}) is the justification for the
names high-loss and low-loss modes.

\textbf{Asymptotic formulas for the quality factor.} \ The perturbation
analysis in the high-loss regime $\beta\gg1$ for the quality factor $Q\left[
w_{j}\left(  \beta\right)  \right]  $, $1\leq j\leq2N$ of the eigenmodes from
(\ref{pod20}) as given by (\ref{radis12}), (\ref{radis12_1}) has already been
carried out in \cite[\S IV A, Prop. 14]{FigWel1}. We will now describe the results.

The quality factor $Q\left[  w_{j}\left(  \beta\right)  \right]  $, $1\leq
j\leq N_{R}$ for each high-loss eigenmode has a series expansion containing
only odd powers of $\beta^{-1}$ implying for $\beta\gg1$ it is a nonnegative
decreasing function in the loss parameter $\beta$ which has the asymptotic
formula as $\beta\rightarrow\infty$%
\begin{equation}
\text{high-loss: \ \ }Q\left[  w_{j}\left(  \beta\right)  \right]  =\frac
{1}{2}\frac{\left\vert \rho_{j}\right\vert }{b_{j}}\beta^{-1}+O\left(
\beta^{-3}\right)  ,\ \ 1\leq j\leq N_{R}. \label{pod21}%
\end{equation}
It follows from this discussion that as $\beta\rightarrow\infty$
\begin{equation}
\text{high-loss: \ \ }Q\left[  w_{j}\left(  \beta\right)  \right]
\searrow0,\ \ 1\leq j\leq N_{R} \label{pod21_1}%
\end{equation}
(here we will use the notation $\searrow$ $0$ or $\nearrow+\infty$ to denote
that a function of $\beta$ is decreasing to $0$ or increasing to $+\infty$,
respectively, as $\beta\rightarrow\infty$).

The quality factor $Q\left[  w_{j}\left(  \beta\right)  \right]  $,
$N_{R}+1\leq j\leq2N$ for each low-loss eigenmode has two possibilities: (i)
$Q\left[  w_{j}\left(  \beta\right)  \right]  =+\infty$ for $\beta\gg1$; (ii)
$Q\left[  w_{j}\left(  \beta\right)  \right]  $ has a series expansion
containing only odd powers of $\beta^{-1}$ implying either $Q\left[
w_{j}\left(  \beta\right)  \right]  \searrow0$ or $Q\left[  w_{j}\left(
\beta\right)  \right]  \nearrow+\infty$ as $\beta\rightarrow\infty$. Finding
necessary and sufficient conditions for when the cases occur is a subtle
problem which is still open in general, but is solved completely in the next
section on overdamping when $\theta=0$ under the nondegeneracy condition
$\operatorname{Ker}\eta\cap\operatorname{Ker}R=\left\{  0\right\}  $. We will
now consider the cases $d_{j}\not =0$ [cf. (\ref{pod22}), (\ref{pod22_1})] or
$\rho_{j}\not =0$ (cf. Prop. \ref{pqaollm}). \ In the former case (which is
the typical case see \cite[\S IV A, Remark 9]{FigWel1}), we have as
$\beta\rightarrow\infty$ the asymptotic formula%
\begin{equation}
\text{low-loss: \ \ }Q\left[  w_{j}\left(  \beta\right)  \right]  =\frac{1}%
{2}\frac{\left\vert \rho_{j}\right\vert }{d_{j}}\beta+O\left(  \beta
^{-1}\right)  ,\ \ \text{if }d_{j}\not =0 \label{pod22}%
\end{equation}
and so it follows in this case that as $\beta\rightarrow\infty$
\begin{equation}
\text{low-loss: \ \ }Q\left[  w_{j}\left(  \beta\right)  \right]
\searrow0\text{, if }\rho_{j}=0\text{ and }Q\left[  w_{j}\left(  \beta\right)
\right]  \nearrow+\infty,\text{ if }\rho_{j}\not =0\text{.} \label{pod22_1}%
\end{equation}

\textbf{Asymptotic oscillatory low-loss modes. }Now the low-loss modes from
(\ref{pod20}) with $\rho_{j}\not =0$ will play a key role in the study of the
selective overdamping phenomenon in Section \ref{spad}. We call such modes the
asymptotic oscillatory low-loss modes since if $\rho_{j}\not =0$ then the
limiting function $\lim_{\beta\rightarrow\infty}v_{j}\left(  t,\beta\right)
=\mathring{w}_{j}e^{-\mathrm{i}\rho_{j}t}$ (for fix $t$) is an oscillatory
function in $t$ with period $2\pi/\rho_{j}$. Thus as $\beta\rightarrow\infty$
we denote this by%

\begin{equation}
\text{asymptotic oscillatory low-loss modes: \ }v_{j}\left(  t,\beta\right)
\sim\mathring{w}_{j}e^{-\mathrm{i}\rho_{j}t}\text{, }\rho_{j}\not =0.
\label{pod22_2}%
\end{equation}
From the quality factor discussion above for the low-loss modes\ the next
proposition follows immediately, regardless of whether $d_{j}\not =0$ or not.

\begin{proposition}
[Q-factor: asymptotic oscillatory low-loss modes]\label{pqaollm} Any low-loss
eigenmodes from (\ref{pod20}) with $\rho_{j}\not =0$ has a quality factor
$Q\left[  w_{j}\left(  \beta\right)  \right]  $ with the property that either
$Q\left[  w_{j}\left(  \beta\right)  \right]  =+\infty$ for all $\beta\gg1$ or
$Q\left[  w_{j}\left(  \beta\right)  \right]  \nearrow+\infty$ as
$\beta\rightarrow\infty$ (i.e., is an unbounded increasing function of the
loss parameter $\beta$).
\end{proposition}

\section{Overdamping analysis\label{sodan}}

The phenomenon of \emph{overdamping} (also called \emph{heavy damping}) is
best known for a simple damped oscillator. Namely, when the damping exceeds
certain critical value all oscillations cease entirely, see, for instance,
\cite[Sec. 2]{Pain}. In other words, if the damped oscillations are described
by the exponential function $e^{-\mathrm{i}\zeta t}$ with a complex constant
$\zeta$ then in the case of overdamping (heavy damping) $\operatorname{Re}%
\zeta=0$. Our interest to overdamping is motivated by the fact that if an
eigenmode becomes overdamped it can not resonate at any finite frequency.
Consequently, the contribution of such a mode to losses at finite frequencies
becomes minimal, and that provides a mechanism for the absorption suppression
for systems composed of lossy and lossless components.

The treatment of overdamping for systems with many degrees of freedom involves
a number of subtleties particularly in our case when the both lossy and
lossless degrees of freedom are present, i.e.,\ when the loss fraction
condition (\ref{radis4b}) is satisfied. We will show that any Lagrangian
system (\ref{radis2a}) with $\theta=0$ can be overdamped with selective
overdamping occurring whenever $R$ does not have full rank, i.e., $N_{R}<N$,
which corresponds to a model of a two-component composite with a high-loss and
a low-loss component.

Let us be clear that we study overdamping in this paper only for the case
$\theta=0$ as the case $\theta\not =0$ will differ significantly in the
analysis and it is not expected that there exists such a critical value for
damping when all oscillations cease entirely for any damped oscillation. A
study of a proper generalization of overdamping in the case $\theta\not =0$
will be carried out in a future publication.

We make here statements and provide arguments for overdamping for Lagrangian
systems. \ We will study the phenomenon known as overdamping in this section
under the following condition:

\begin{condition}
[no gyrotropy]\label{cndod}For our study of overdamping we assume henceforth
that%
\begin{equation}
\theta=0. \label{odan1}%
\end{equation}

\end{condition}

The following proposition is a key result that will be referenced often in our
study of overdamping. It describes the spectrum of the $2N\times2N$ matrices
$B$, $\Omega$ in terms of the spectrum of the $N\times N$ matrices
$\alpha^{-1}\eta$, $\alpha^{-1}R$.

\begin{proposition}
[spectra relations]\label{pspre}For the $2N\times2N$ matrices $B$, $\Omega$
and the $N\times N $ matrices $\alpha$, $\eta$, $R$ we have
$\operatorname{rank}B=\operatorname{rank}R$ ($=N_{R}$) and
\begin{align}
\det\left(  \zeta\mathbf{1}-\Omega\right)   &  =\det\left(  \zeta
^{2}\mathbf{1}-\alpha^{-1}\eta\right)  ,\label{pspre1}\\
\det\left(  \zeta\mathbf{1}-B\right)   &  =\zeta^{N}\det\left(  \zeta
\mathbf{1}-\alpha^{-1}R\right)  ,\nonumber
\end{align}
for every $\zeta\in%
%TCIMACRO{\U{2102} }%
%BeginExpansion
\mathbb{C}
%EndExpansion
$. \ In particular, if $b_{\text{min}}$ and $\omega_{\text{max}}$ denote the
smallest nonzero eigenvalue and the largest eigenvalue of $B$ and $\Omega$,
respectively, and $\omega_{\text{min}}$ denotes the smallest positive
eigenvalue of $\Omega$ then
\begin{equation}
\omega_{\text{max}}=\sqrt{\max\sigma\left(  \alpha^{-1}\eta\right)  },\text{
\ \ }\omega_{\text{min}}=\sqrt{\min_{\lambda\in\sigma\left(  \alpha^{-1}%
\eta\right)  \text{, }\lambda>0}\lambda}, \label{lod8}%
\end{equation}
and%
\begin{equation}
b_{\text{min}}=\min_{\lambda\in\sigma\left(  \alpha^{-1}R\right)  \text{,
}\lambda\not =0}\lambda. \label{lod9}%
\end{equation}

\end{proposition}

\subsection{Complete and partial overdamping\label{scpod}}

To study whether overdamping is possible we need to determine conditions under
which%
\begin{equation}
\det\left(  \zeta\mathbf{1}-A\left(  \beta\right)  \right)  =0,\text{
\ \ }\operatorname{Re}\zeta=0\text{.} \label{lod1}%
\end{equation}
has a solution for the system operator $A\left(  \beta\right)  =\Omega
-\mathrm{i}\beta B$. \ But we know from Corollary \ref{cpfsp} that
$\det\left(  \zeta\mathbf{1}-A\left(  \beta\right)  \right)  =0$ if and only
if $\det C\left(  \zeta,\beta\right)  =0.$ Hence the solutions to (\ref{lod1})
are the solution of%
\begin{equation}
\det C\left(  \zeta,\beta\right)  =0,\text{ \ \ }\zeta=-\mathrm{i}%
\lambda,\text{ \ \ }\lambda\in%
%TCIMACRO{\U{211d} }%
%BeginExpansion
\mathbb{R}
%EndExpansion
\text{.} \label{lod2}%
\end{equation}
Thus to study overdamping it suffices to study conditions under which the
quadratic eigenvalue problem%
\begin{equation}
C\left(  -\mathrm{i}\lambda,\beta\right)  q=0,\text{ \ \ }q\not =0,\text{
\ \ }\lambda\in%
%TCIMACRO{\U{211d} }%
%BeginExpansion
\mathbb{R}
%EndExpansion
\text{.} \label{lod3}%
\end{equation}
has a solution. \ To be explicit we have%
\begin{equation}
C\left(  -\mathrm{i}\lambda,\beta\right)  =-\left(  \lambda^{2}\alpha
-\lambda\beta R+\eta\right)  . \label{lod4}%
\end{equation}
As it turns out the key to studying overdamping will be to study zero sets of
the family of quadratic forms%
\begin{equation}
\mathfrak{q}\left(  q,\lambda\right)  =\left(  q,-C\left(  -\mathrm{i}%
\lambda,\beta\right)  q\right)  =\lambda^{2}\left(  q,\alpha q\right)
-\lambda\beta\left(  q,Rq\right)  +\left(  q,\eta q\right)  \label{lod5}%
\end{equation}
on certain subspaces of the Hilbert space $H_{\mathrm{p}}=%
%TCIMACRO{\U{2102} }%
%BeginExpansion
\mathbb{C}
%EndExpansion
^{N}$. In particular, we find that%
\begin{align}
q  &  \in\operatorname{Ker}C\left(  -\mathrm{i}\lambda,\beta\right)  ,\text{
}q\not =0\Rightarrow\mathfrak{q}\left(  q,\lambda\right)  =0\ \text{and}%
\label{lod6}\\
\lambda &  =\frac{\beta}{2}\frac{\left(  q,Rq\right)  }{\left(  q,\alpha
q\right)  }\pm\sqrt{\left[  \frac{\beta}{2}\frac{\left(  q,Rq\right)
}{\left(  q,\alpha q\right)  }\right]  ^{2}-\frac{\left(  q,\eta q\right)
}{\left(  q,\alpha q\right)  }}.\nonumber
\end{align}
Thus we find that a necessary and sufficient condition for an eigenvalue
$\zeta$ of the system operator $A\left(  \beta\right)  $ to satisfy
$\operatorname{Re}\zeta=0$ is%
\begin{equation}
\frac{\beta}{2}\frac{\left(  q,Rq\right)  }{\left(  q,\alpha q\right)  }%
\geq\sqrt{\frac{\left(  q,\eta q\right)  }{\left(  q,\alpha q\right)  }%
},\text{ for some }q\in\operatorname{Ker}C\left(  \zeta,\beta\right)  ,\text{
}q\not =0. \label{lod7}%
\end{equation}

Now there are some fundamental inequalities that play a key role in our
overdamping analysis which are described in the following proposition. We will
prove part of this proposition, the rest is proved in Section \ref{sproofs},
since the proof is enlightening showing how the energy conservation law
(\ref{radis4}) and the virial theorem \ref{tvir} can be used to derive these inequalities.

\begin{proposition}
[fundamental inequalities]\label{pfuineq}The following statements are true.

\begin{enumerate}
\item For any $q\in%
%TCIMACRO{\U{2102} }%
%BeginExpansion
\mathbb{C}
%EndExpansion
^{N}$ with $q\not =0$,%
\begin{equation}
\omega_{\text{max}}\geq\sqrt{\frac{\left(  q,\eta q\right)  }{\left(  q,\alpha
q\right)  }}\text{.} \label{lod7_0}%
\end{equation}
Moreover, if $\eta^{-1}$ exists then%
\begin{equation}
\sqrt{\frac{\left(  q,\eta q\right)  }{\left(  q,\alpha q\right)  }}\geq
\omega_{\text{min}}. \label{lod7_00}%
\end{equation}

\item If the matrix $R$ has full rank, i.e., $N_{R}=N$, then for any $q\in%
%TCIMACRO{\U{2102} }%
%BeginExpansion
\mathbb{C}
%EndExpansion
^{N}$ with $q\not =0$,
\begin{equation}
\frac{\beta}{2}\frac{\left(  q,Rq\right)  }{\left(  q,\alpha q\right)  }%
\geq\frac{\beta}{2}b_{\text{min}}\text{.} \label{lod7_0a}%
\end{equation}

\item For any $q\in\operatorname{Ker}C\left(  \zeta,\beta\right)  ,$
$q\not =0~$with $\operatorname{Re}\zeta\not =0$,%
\begin{equation}
\frac{\beta}{2}\frac{\left(  q,Rq\right)  }{\left(  q,\alpha q\right)
}=-\operatorname{Im}\zeta<\left\vert \zeta\right\vert =\sqrt{\frac{\left(
q,\eta q\right)  }{\left(  q,\alpha q\right)  }}\leq\omega_{\text{max}}.
\label{lod7_0b}%
\end{equation}

\item If $\det C\left(  \zeta,\beta\right)  =0$ [or equivalently, $\zeta
\in\sigma\left(  A\left(  \beta\right)  \right)  $] then $\operatorname{Re}%
\zeta=0$ whenever any one of the following two inequalities is satisfied: (i)
$-\operatorname{Im}\zeta\geq\omega_{\text{max}}$ or (ii) $\eta^{-1}$ exists
and $\left\vert \zeta\right\vert <\omega_{\text{min}}$.

\item If $Q\left(  t\right)  =qe^{-\mathrm{i}\zeta t}$ is an eigenmode of the
Lagrangian system (\ref{radis2a}) then energy equipartition, i.e.,
$\mathcal{T}\left(  \dot{Q},Q\right)  =\mathcal{V}\left(  \dot{Q},Q\right)  $,
can only hold if $\left\vert \zeta\right\vert \leq\omega_{\text{max}}$ or, in
the case $\eta^{-1}$ exists, if $\omega_{\text{min}}\leq\left\vert
\zeta\right\vert \leq\omega_{\text{max}}$.
\end{enumerate}
\end{proposition}

\begin{proof}
The first two statements are proved in Section \ref{sproofs}. The third
statement will now be proved using the energy conservation law (\ref{radis4})
and the virial theorem \ref{tvir}, or more specifically, Corollary
\ref{tvir_eqp} on the equipartition of energy. Suppose $q\in\operatorname{Ker}%
C\left(  \zeta,\beta\right)  ,$ $q\not =0~$with $\operatorname{Re}\zeta
\not =0$. Then $Q\left(  t\right)  =qe^{-\mathrm{i}\zeta t}$ is an eigenmode
of the Lagrangian system (\ref{radis2a}). It follows from (\ref{radis4}) and
Corollary \ref{tvir_eqp} that for the system energy $\mathcal{H}%
=\mathcal{T}\left(  \dot{Q},Q\right)  +\mathcal{V}\left(  \dot{Q},Q\right)  $
we have%
\begin{align}
2\mathcal{R}\left(  \dot{Q}\right)   &  =-\frac{d}{dt}\mathcal{H}%
=-2\operatorname{Im}\zeta\mathcal{H},\label{lod7_0c}\\
\mathcal{T}\left(  \dot{Q}\right)   &  =\mathcal{V}\left(  Q\right)  ,
\label{lod7_0d}%
\end{align}
where $\mathcal{T}\left(  \dot{Q},Q\right)  =\frac{1}{2}\left(  \dot{Q}%
,\eta\dot{Q}\right)  $ and $\mathcal{V}\left(  \dot{Q},Q\right)  =\frac{1}%
{2}\left(  Q,\eta Q\right)  $ are the kinetic and potential energy,
respectively, of this state $Q$ and $2\mathcal{R}\left(  \dot{Q}\right)
=\left(  \dot{Q},\beta R\dot{Q}\right)  $ is the dissipated power. Notice that
for $\mathcal{E}\in\left\{  \mathcal{H},\mathcal{T},\mathcal{V},\mathcal{R}%
\right\}  $ we have $\mathcal{E}\left(  t\right)  =e^{2\operatorname{Im}\zeta
t}\mathcal{E}\left(  0\right)  $. It now follows from this and (\ref{lod7_0c}%
), (\ref{lod7_0d}) that%
\begin{align}
\frac{\beta}{2}\frac{\left(  q,Rq\right)  }{\left(  q,\alpha q\right)  }  &
=\frac{\mathcal{R}\left(  \dot{Q}\right)  }{2\mathcal{T}\left(  \dot
{Q},Q\right)  }=\frac{-\operatorname{Im}\zeta\mathcal{H}}{\mathcal{H}%
}=-\operatorname{Im}\zeta,\label{lod7_0e}\\
\sqrt{\frac{\left(  q,\eta q\right)  }{\left(  q,\alpha q\right)  }}  &
=\sqrt{\left\vert \zeta\right\vert ^{2}\frac{\mathcal{V}\left(  \dot
{Q},Q\right)  }{\mathcal{T}\left(  \dot{Q},Q\right)  }}=\left\vert
\zeta\right\vert . \label{lod7_0f}%
\end{align}
The proof of the third statement now follows from these two equalities and the
inequality (\ref{lod7_0}). The fourth statement follows immediately from the
inequality (\ref{lod7_0b}) and the inequality (\ref{lod7_00}).

The fifth statement follows from the fact that if $\mathcal{T}\left(  \dot
{Q},Q\right)  =\mathcal{V}\left(  \dot{Q},Q\right)  $ then the equality
(\ref{lod7_0f}) holds, regardless of whether $\operatorname{Re}\zeta\not =0$
or not as long as $\zeta\not =0$, and so the result follows immediately from
the inequalities (\ref{lod7_0}) and (\ref{lod7_00}) if $\zeta\not =0$. And
since the existence of $\eta^{-1}$ prohibits $\zeta=0$, the proof now follows.
This completes the proof.
\end{proof}

From the inequalities (\ref{lod7_0a}), (\ref{lod7_0b}) one can see clearly
that if the matrix $R$ has full rank then all the eigenvalues of $A\left(
\beta\right)  $ are purely imaginary once
\begin{equation}
\beta\geq2\frac{\omega_{\text{max}}}{b_{\text{min}}}. \label{lod7_3}%
\end{equation}
In other words, (\ref{lod7_3}) is a sufficient condition for all of the
eigenmodes of the canonical system (\ref{radis8}) [or, equivalently, the
Lagrangian system (\ref{radis2a})] to be overdamped, provided $R$ has full rank.

But it is also clear that this argument must be refined if $R$ does not have
full rank, i.e., $N_{R}<N$, since then $\det R=0$ and hence $\min\sigma\left(
\alpha^{-1}R\right)  =0$. The refinement we use is perturbation theory. In
particular, we use our results in Sections \ref{sspsy}, \ref{smdhlr},
\ref{svir} on the modal dichotomy and the virial theorem to derive our main
results on overdamping. Our argument is essentially based on determining which
of the eigenmodes cannot maintain the energy equipartition in Corollary
\ref{tvir_eqp}, i.e., the equality between the kinetic and potential energy,
when $\beta$ is sufficiently large by either using the inequalities in
Corollary \ref{cmdic} or using the asymptotic expansions in Section
\ref{smdhlr} for the eigenmodes in (\ref{pod20}). This is where the fifth
statement in Proposition \ref{pfuineq} is relevant.

We derive the following results on overdamping using the inequalities in
Proposition \ref{pfuineq}.

\begin{theorem}
[complete overdamping]\label{tcovd}If $\beta\geq2\frac{\omega_{\text{max}}%
}{b_{\text{min}}}$ and $R $ has full rank, i.e., $N_{R}=N$, then all the
eigenmodes of the canonical system (\ref{radis8}) are overdamped, i.e., for
the system operator $A\left(  \beta\right)  =\Omega-\mathrm{i}\beta B$ we have%
\begin{equation}
\sigma\left(  A\left(  \beta\right)  \right)  =\left\{  \zeta\in\sigma\left(
A\left(  \beta\right)  \right)  :\operatorname{Re}\zeta=0\right\}  .
\label{tcovd1}%
\end{equation}

\end{theorem}

\begin{example}
[overdamped regime is optimal]\label{eodop}The following example shows for the
class of system operators $A=\Omega-\mathrm{i}\beta B$ under consideration in
this paper, the regime $\beta>2\frac{\omega_{\text{max}}}{b_{\text{min}}}$ is
optimal for guaranteeing overdamping in the sense that we can always find an
example of a system operator satisfying $\beta<2\frac{\omega_{\text{max}}%
}{b_{\text{min}}}$ with $\beta$ as close to $2\frac{\omega_{\text{max}}%
}{b_{\text{min}}}$ as we like such that the set $\left\{  \zeta\in
\sigma\left(  A\left(  \beta\right)  \right)  :\operatorname{Re}%
\zeta=0\right\}  $ is empty. \ In particular, it is proven using the following
family of system operators%
\begin{gather}
\Omega=\left[
\begin{array}
[c]{cc}%
0 & -\mathrm{i}\\
\mathrm{i} & 0
\end{array}
\right]  ,\text{ \ \ }B=\left[
\begin{array}
[c]{cc}%
1 & 0\\
0 & 0
\end{array}
\right]  ,\label{eodop1}\\
A\left(  \beta\right)  =\Omega-\mathrm{i}\beta B,\text{ \ \ }\beta
\geq0,\nonumber
\end{gather}
derived from the system with one degree of freedom having Lagrangian and
Rayleigh dissipation function
\begin{gather}
\mathcal{L}=\mathcal{L}\left(  Q,\dot{Q}\right)  =\frac{1}{2}\left[
\begin{array}
[c]{l}%
\dot{Q}\\
Q
\end{array}
\right]  ^{\mathrm{T}}M_{\mathrm{L}}\left[
\begin{array}
[c]{l}%
\dot{Q}\\
Q
\end{array}
\right]  ,\quad M_{\mathrm{L}}=\left[
\begin{array}
[c]{ll}%
1 & 0\\
0 & -1
\end{array}
\right]  ,\label{eodop2}\\
R=\frac{1}{2}\left(  \dot{Q},\beta R\dot{Q}\right)  ,\text{ \ \ }%
R=1\text{,}\nonumber
\end{gather}
which is just is the Lagrangian for a one degree-of-freedom spring-mass-damper
system with mass $1$, spring constant $1$, and viscous damping coefficient
$\beta$. Then we have%
\begin{align}
b_{\text{min}}  &  =1,\text{ \ \ }\omega_{\text{max}}=1,\label{eodop3}\\
\sigma\left(  A\left(  \beta\right)  \right)   &  =\left\{  \xi_{\pm}\left(
\beta\right)  =-\mathrm{i}\frac{\beta}{2}\pm\sqrt{1-\left(  \frac{\beta}%
{2}\right)  ^{2}}\right\}  ,\nonumber
\end{align}
and therefore%
\begin{align}
\text{if }\beta &  >2\text{ then }\left\{  \zeta\in\sigma\left(  A\left(
\beta\right)  \right)  :\operatorname{Re}\zeta=0\right\}  =\sigma\left(
A\left(  \beta\right)  \right)  ,\label{eodop4}\\
\text{if }\beta &  =2\text{ then }\sigma\left(  A\left(  \beta\right)
\right)  =\left\{  -\mathrm{i}\right\}  ,\nonumber\\
\text{if }\beta &  <2\text{ then }\left\{  \zeta\in\sigma\left(  A\left(
\beta\right)  \right)  :\operatorname{Re}\zeta=0\right\}  =\emptyset.\nonumber
\end{align}
In this case, at the boundary of the overdamping regime, where $\beta
=2\frac{\omega_{\text{max}}}{b_{\text{min}}}$, we have critical damping of the system.
\end{example}

Now we want to answer the question of which modes of the canonical system
(\ref{radis8}) are overdamped in the case when $R$ does not have full rank,
i.e., $N_{R}<N$. \ The next theorem gives us a partial answer to the question.

\begin{theorem}
[partial overdamping]\label{tpovd}If $\beta>2\frac{\omega_{\text{max}}%
}{b_{\text{min}}}$ and $N_{R}<N$ then the modal dichotomy occurs as in Theorem
\ref{tmdic} and Corollary \ref{cmdic}. In addition, any eigenmode of the
canonical system (\ref{radis8}) in the invariant subspace $H_{h\ell}\left(
\beta\right)  $ (the high-loss susceptible subspace) is overdamped, i.e.,
\begin{equation}
\sigma\left(  A\left(  \beta\right)  |_{H_{h\ell}\left(  \beta\right)
}\right)  \subseteq\left\{  \zeta\in\sigma\left(  A\left(  \beta\right)
\right)  :\operatorname{Re}\zeta=0\right\}  . \label{tpovd1}%
\end{equation}

\end{theorem}

Now we ask the important and natural question: Can any of the modes of the
system operator $A\left(  \beta\right)  $ in the low-loss susceptible subspace
$H_{\ell\ell}\left(  \beta\right)  $ be overdamped when $R$ from the Rayleigh
dissipative function $\mathcal{R}$ in (\ref{radis1}) does not have full rank,
i.e., in the case when the Lagrangian system is a model of a two-component
composite with a lossy and a lossless component? This question will be
investigated in the next section.

\subsection{Selective overdamping\label{spad}}

We are interested in solving the problem of describing all the modes of the
canonical system (\ref{radis8}) [or, equivalently, the Lagrangian system
(\ref{radis2a})], with system operator $A=$ $\Omega-\mathrm{i}\beta B$, that
are overdamped when losses are sufficiently large, i.e., $\beta\gg1$. \ In the
previous section we solved the problem for Lagrangian systems with complete
Rayleigh dissipative, i.e., $R$ has full rank, for then Theorem \ref{tcovd}
tells us that all the modes will become overdamped. On the other hand, for
Lagrangian systems with incomplete Rayleigh dissipative, i.e., $R$ does not
have full rank, then Theorem \ref{tpovd} is a partial solution to the problem
since it tells us that, due to sufficiently large losses, modal dichotomy
occurs splitting the modes into the two $A\left(  \beta\right)  $-invariant
subspaces $H_{\ell\ell}\left(  \beta\right)  $ and $H_{h\ell}\left(
\beta\right)  $ (the high-loss and low-loss susceptible subspaces,
respectively) with the modes belonging to $H_{h\ell}\left(  \beta\right)  $
overdamped. Thus to solve the problem completely we need to describe which of
the modes in the low-loss susceptible subspace $H_{\ell\ell}\left(
\beta\right)  $ can be overdamped when losses are sufficiently large. In this
section we solve the problem asymptotically as losses become extremely large,
i.e., as $\beta\rightarrow\infty$, using the virial theorem \ref{tvir} and, in
particular, energy equipartition \ref{tvir_eqp} for the eigenmodes and the
perturbation theory developed in \cite{FigWel1} which we described in Section
\ref{smdhlr}.

Our description of the selective overdamping phenomenon is based on the modal
dichotomy in the high-loss regime $\beta\gg1$ as described in Section
\ref{smdhlr}. In particular, we use Theorem \ref{cpodr} to characterize
overdamping in terms of the high-loss and low-loss eigenmodes $v_{j}\left(
t,\beta\right)  =w_{j}\left(  \beta\right)  e^{-\mathrm{i}\zeta_{j}\left(
\beta\right)  t}$, $1\leq j\leq2N$ of the canonical system (\ref{radis8})
which split into the two distinct classes based on their dissipative
properties (\ref{pod20}).

We say that such an eigenmode $v_{j}\left(  t,\beta\right)  $ is
\emph{overdamped} for $\beta\gg1$ if $\operatorname{Re}\zeta_{j}\left(
\beta\right)  =0$ for all $\beta$ sufficiently large. We say that such an
eigenmode \emph{remains oscillatory} for $\beta\gg1$ if $\operatorname{Re}%
\zeta_{j}\left(  \beta\right)  \not =0$ for all $\beta$ sufficiently large.

\subsubsection{On the overdamping of the high-loss eigenmodes\label{spad1}}

Our first result tells us that all of the high-loss eigenmodes $v_{j}\left(
t,\beta\right)  =w_{j}\left(  \beta\right)  e^{-\mathrm{i}\zeta_{j}\left(
\beta\right)  t}$, $1\leq j\leq N_{R}$ from (\ref{pod20}) are overdamped
regardless of whether $R$ has full rank or not. It is an immediate corollary
of Theorem \ref{tpovd} and Theorem \ref{cpodr}.

\begin{theorem}
[overdamped high-loss modes]\label{thmod}All of the high-loss eigenmodes from
(\ref{pod20}) are overdamped for $\beta\gg1$.
\end{theorem}

\subsubsection{On the overdamping of the low-loss eigenmodes\label{spad2}}

Now we will discuss and give a solution to the problem of determining which of
the low-loss eigenmodes $v_{j}\left(  t\right)  =w_{j}\left(  \beta\right)
e^{-\mathrm{i}\zeta_{j}\left(  \beta\right)  t}$, $N_{R}+1\leq j\leq2N$ in
(\ref{pod20}) will be overdamped. \ In order to do this we work under the
following nondegeneracy condition:

\begin{condition}
[nondegeneracy]\label{cndnd}For our study of the overdamping of the low-loss
modes in this section, we assume henceforth that
\begin{equation}
\operatorname{Ker}\eta\cap\operatorname{Ker}R=\left\{  0\right\}  .
\label{llod0}%
\end{equation}

\end{condition}

We have two main goals in this section. Our first is to find necessary and
sufficient conditions to have $\operatorname{Re}\zeta_{j}\left(  \beta\right)
=0$ for $\beta\gg1$. Let $\kappa$ be the number of low-loss eigenmodes that
are overdamped for $\beta\gg1$. Our second goal is to calculate this number
$\kappa$.

We begin by recalling some of the important properties of the low-loss
eigenpairs $\zeta_{j}\left(  \beta\right)  $, $w_{j}\left(  \beta\right)  $,
$N_{R}+1\leq j\leq2N$ in (\ref{pod11}) for the system operator $A\left(
\beta\right)  =\Omega-\mathrm{i}\beta B$ that we will need in this section.
\ By (\ref{pod14}) and (\ref{pod16}) the limits%
\begin{equation}
\lim_{\beta\rightarrow\infty}w_{j}\left(  \beta\right)  =\mathring{w}%
_{j},\text{ \ \ }\lim_{\beta\rightarrow\infty}\operatorname{Re}\zeta
_{j}\left(  \beta\right)  =\rho_{j},\text{ \ \ }\lim_{\beta\rightarrow\infty
}\operatorname{Im}\zeta_{j}\left(  \beta\right)  =0 \label{llod1}%
\end{equation}
exist and from (\ref{pod8}), (\ref{pod15}) the limiting vectors $\left\{
\mathring{w}_{j}\right\}  _{j=N_{R}+1}^{2N}$ are an orthonormal basis for
$\operatorname{Ker}B$ which also diagonalize the self-adjoint operator
$\Omega_{1}$ from (\ref{pod7}) satisfying%
\begin{equation}
B\mathring{w}_{j}=0,\text{ \ \ }\Omega_{1}\mathring{w}_{j}=\rho_{j}%
\mathring{w}_{j},\text{ \ \ }N_{R}+1\leq j\leq2N. \label{llod2}%
\end{equation}

\textbf{Overview of main results. }Before we proceed, let us now summarize the
results in this section. First, its obvious from this discussion that a
necessary condition for the low-loss eigenmode $v_{j}\left(  t,\beta\right)  $
to be overdamped for $\beta\gg1$ is $\rho_{j}=0$ (even though it does so with
a slow damping of order $\beta^{-1}$, as seen in (\ref{llod1}) and
(\ref{pod14})), but what is not obvious is that this is in fact a sufficient
condition. Although in the case $\eta^{-1}$ exists (i.e., $\operatorname{Ker}%
\eta=\left\{  0\right\}  $), it just follows from the fourth statement in
Proposition \ref{pfuineq} since if $\rho_{j}=0$ then by (\ref{llod1}) we have
$\zeta_{j}\left(  \beta\right)  \rightarrow0$ which implies $\left\vert
\zeta_{j}\left(  \beta\right)  \right\vert <\omega_{\text{min}}$ for $\beta
\gg1$ and hence $\operatorname{Re}\zeta_{j}\left(  \beta\right)  =0$ for
$\beta\gg1$. But this argument fails when\ $\operatorname{Ker}\eta
\not =\left\{  0\right\}  $. Thus we must find an alternative method to show
that $\rho_{j}=0$ is equivalent to $\operatorname{Re}\zeta_{j}\left(
\beta\right)  =0$ for $\beta\gg1$. The key idea that leads to an alternative
method is to realize that by the virial theorem \ref{tvir} and the fifth
statement in Proposition \ref{pfuineq}, when $\eta^{-1}$ exists the modes with
$\rho_{j}=0$ are exactly the modes that break the energy equipartition.

Hence our general method is to show, even in the case $\operatorname{Ker}%
\eta=\left\{  0\right\}  $, that for a low-loss modes with $\rho_{j}=0$, there
is a breakdown of the virial theorem \ref{tvir} since as losses increase,
i.e., as $\beta\rightarrow\infty$, the equality between the kinetic and
potential energy (see Corollary \ref{tvir_eqp}) of the corresponding eigenmode
of the Lagrangian system (\ref{radis2a}) can no longer be maintained and must
eventually break which forces the modes to be non-oscillatory for $\beta\gg1$,
i.e., overdamped. This proves that $\rho_{j}=0$ is equivalent to
$\operatorname{Re}\zeta_{j}\left(  \beta\right)  =0$ for $\beta\gg1$.

Thus it follows that the low-loss eigenmodes in (\ref{pod20}) of the canonical
system (\ref{radis8}) which are overdamped for $\beta\gg1$, their limiting
vectors (satisfy $\lim_{\beta\rightarrow\infty}v_{j}\left(  t,\beta\right)
=\mathring{w}_{j}\in$ $\operatorname{Ker}\Omega_{1}$ for $t$ fixed) form a
orthonormal basis for $\operatorname{Ker}\Omega_{1}$. This implies of course
that $\kappa=\dim\operatorname{Ker}\Omega_{1}$ is the number of low-loss
eigenmodes that are overdamped for $\beta\gg1$. A major result of this section
is that in fact $\kappa=N_{R}$, where $N_{R}=\operatorname{rank}%
B=\operatorname{rank}R$. \ But $N_{R} $ is also the number of high-loss
eigenpairs in (\ref{pod11}). \ Thus we conclude using these results and
Theorem \ref{thmod} that the number of eigenmodes in (\ref{pod20}) which are
overdamped in the high-loss regime $\beta\gg1$ is exactly $2N_{R}$ with
$N_{R}$ of these being low-loss eigenmodes and the rest being all the
high-loss eigenmodes. Moreover, in the case $N_{R}<N$, the remaining
$2N-2N_{R}>0$ modes are low-loss oscillatory modes with an extremely high
quality factor that actually increases as the losses increase, i.e., as
$\beta\rightarrow\infty$.

\textbf{Main results. }We now give the statements of our main results. \ The
proof of the first theorem we give in this section as it is enlightening. The
proofs of the rest of the statements below are found in Section \ref{sproofs}.

\begin{theorem}
[overdamped low-loss modes]\label{tlmod}A necessary and sufficient condition
for a low-loss eigenmode $v_{j}\left(  t,\beta\right)  =w_{j}\left(
\beta\right)  e^{-\mathrm{i}\zeta_{j}\left(  \beta\right)  t}$ from
(\ref{pod20}) to be overdamped for $\beta\gg1$ is $\rho_{j}=0$ (or
equivalently, $\mathring{w}_{j}\in\operatorname{Ker}\Omega_{1}$).
\end{theorem}

\begin{proof}
Let $v_{j}\left(  t,\beta\right)  =w_{j}\left(  \beta\right)  e^{-\mathrm{i}%
\zeta_{j}\left(  \beta\right)  t}$ be a low-loss eigenmode from (\ref{pod20}).
If it is overdamped for $\beta\gg1$ then $\operatorname{Re}\zeta_{j}\left(
\beta\right)  =0$ for $\beta\gg1$ (i.e., for all $\beta$ sufficiently large)
and so by (\ref{llod1}), (\ref{llod2}) we have $\rho_{j}=0$ and $\mathring
{w}_{j}\in\operatorname{Ker}\Omega_{1}$ (the nullspace of $\Omega_{1}$).

Lets now prove the converse. \ Suppose $\rho_{j}=0$. Then we must show that
$\operatorname{Re}\zeta_{j}\left(  \beta\right)  =0$ for $\beta\gg1$. Suppose
that this was not true. Then since $\zeta_{j}\left(  \beta\right)  $ was
analytic at $\beta=\infty$ we must have $\operatorname{Re}\zeta_{j}\left(
\beta\right)  \not =0$ for $\beta\gg1$. Thus by Corollary \ref{cpfsp} there is
a corresponding eigenmode $Q_{j}\left(  t\,,\beta\right)  =q_{j}\left(
\beta\right)  e^{-\mathrm{i}\zeta_{j}\left(  \beta\right)  t}$ of the
Lagrangian system (\ref{radis2a}) such that with respect to the block
representation (\ref{pfsp8b}) the vector $w_{j}\left(  \beta\right)  $ is
\begin{equation}
w_{j}\left(  \beta\right)  =\left[
\begin{array}
[c]{c}%
-\mathrm{i}\zeta_{j}\left(  \beta\right)  \sqrt{\alpha}q_{j}\left(
\beta\right) \\
\sqrt{\eta}q_{j}\left(  \beta\right)
\end{array}
\right]  . \label{llod3}%
\end{equation}
In particular, $v_{j}$ and $Q_{j}$ are related by (\ref{radis7}), namely,%
\begin{equation}
v_{j}=Ku_{j},\text{ where\ }u_{j}=%
\begin{bmatrix}
P_{j}\\
Q_{j}%
\end{bmatrix}
,\text{ }P_{j}=\alpha\dot{Q}_{j}. \label{llod4}%
\end{equation}
By Section \ref{seneq} on the energetic equivalence between the canonical
system (\ref{radis8}) and the Lagrangian system (\ref{radis2a}) it follows
that%
\begin{equation}
U\left[  v_{j}\right]  =\mathcal{T}\left(  \dot{Q},Q\right)  +\mathcal{V}%
\left(  \dot{Q},Q\right)  . \label{llod5}%
\end{equation}
The term on the left is the system energy of the eigenmode $v_{j}$ of the
canonical system defined in (\ref{radis11a}) whereas the sum on the right is
the system energy of the eigenmode $Q_{j}$ of the Lagrangian system
(\ref{radis2a}), i.e., the sum of its kinetic and potential energy. \ As
proved in Section \ref{pspad}, for any fixed time $t$ we have
\begin{gather}
\frac{1}{2}=\lim_{\beta\rightarrow\infty}U\left[  v_{j}\left(  t,\beta\right)
\right]  =\lim_{\beta\rightarrow\infty}\mathcal{V}\left(  \dot{Q}_{j}\left(
t\,,\beta\right)  ,Q_{j}\left(  t\,,\beta\right)  \right)  ,\label{llod6}\\
\lim_{\beta\rightarrow\infty}\mathcal{T}\left(  \dot{Q}_{j}\left(
t\,,\beta\right)  ,\dot{Q}_{j}\left(  t\,,\beta\right)  \right)  =0.\nonumber
\end{gather}
This proves that
\begin{equation}
\mathcal{T}\left(  \dot{Q},Q\right)  \not =\mathcal{V}\left(  \dot
{Q},Q\right)  \text{ for }\beta\gg1 \label{llod7}%
\end{equation}
and, consequently, by the energy equipartition (see Corollary \ref{tvir_eqp})
that%
\begin{equation}
\operatorname{Re}\zeta_{j}\left(  \beta\right)  =0\text{ for }\beta
\gg1\text{.} \label{llod8}%
\end{equation}
This completes the proof.
\end{proof}

\begin{corollary}
\label{clmod}The collection of limiting values as $\beta\rightarrow\infty$
(with $t$ fixed) of the low-loss eigenmodes in (\ref{pod20}) of the canonical
system (\ref{radis8}) which are overdamped for $\beta\gg1$ is $\left\{
\mathring{w}_{j}\right\}  _{j=N_{R}+1}^{2N}\cap$ $\operatorname{Ker}\Omega
_{1}$ and is an orthonormal basis for $\operatorname{Ker}\Omega_{1}$.
\end{corollary}

\begin{proposition}
\label{plmod}Let $\kappa$ denote the number low-loss eigenmodes in
(\ref{pod20}) of the canonical system (\ref{radis8}) which are overdamped for
$\beta\gg1$. Then
\begin{equation}
\kappa=\dim\operatorname{Ker}\Omega_{1}=N_{R}, \label{plmod1}%
\end{equation}
where $N_{R}=\operatorname{rank}R=\operatorname{rank}B$.
\end{proposition}

\begin{theorem}
[number of overdamped modes]\label{tnlm}The total number of eigenmodes of the
canonical system (\ref{radis8}) from (\ref{pod20}) which are overdamped for
$\beta\gg1$ is $2N_{R}$, where $N_{R}=\operatorname{rank}R=\operatorname{rank}%
B$. Furthermore, $N_{R}$ of these are low-loss eigenmodes and the rest are all
of the high-loss eigenmodes. Moreover, if $N_{R}<N$ then the $2N-2N_{R}$ modes
which are not overdamped as $\beta\gg1$ are low-loss eigenmodes which remain
oscillatory for $\beta\gg1$.
\end{theorem}

\begin{theorem}
[selective overdamping]\label{tseod}If $N_{R}<N$ then there is exactly
$2N-2N_{R}$ low-loss eigenmodes of the canonical system (\ref{radis8}) from
(\ref{pod20}) which remain oscillatory for $\beta\gg1$ (i.e.,
$\operatorname{Re}\zeta_{j}\left(  \beta\right)  \not =0$ for $\beta\gg1$) and
their quality factors as given by (\ref{radis12}), (\ref{radis12_1}) have the
following property
\begin{equation}
\text{low-loss oscill. modes: }Q\left[  w_{j}\left(  \beta\right)  \right]
=+\infty\text{ for }\beta\gg1\text{ or }Q\left[  w_{j}\left(  \beta\right)
\right]  \nearrow+\infty\text{ as }\beta\rightarrow\infty. \label{tseod1}%
\end{equation}

\end{theorem}

\section{Proof of results\label{sproofs}}

This section contains the proofs of the results from Sections \ref{svir},
\ref{sspan} and \ref{sodan}. Proofs are organized in the order in which
statements appeared and grouped according to the sections and subsections in
which they are located in the paper. \ All assumptions, notation, and
conventions used here will adhere to that previously introduced in this paper.

\subsection{The virial theorem for dissipative systems and equipartition of
energy}

The proofs of statements in Section \ref{svir} are contained here.

\begin{proof}
[Proof of Theorem \ref{tvir}]For $Q\left(  t\right)  =qe^{-\mathrm{i}\zeta t}$
an eigenmode of the Lagrangian system (\ref{radis2a}), we consider the virial%
\begin{equation}
\mathrm{G}=\left(  \alpha\dot{Q},Q\right)  . \label{vir1}%
\end{equation}
In our proof we will use the real quantities $\mathcal{T}_{0}\left(  \dot
{Q}\right)  =\frac{1}{2}\left(  \dot{Q},\alpha\dot{Q}\right)  $ and
$\mathcal{V}_{0}\left(  Q\right)  =\frac{1}{2}\left(  Q,\eta Q\right)  $. If
$\zeta=0$ we are done. Thus, suppose that $\zeta\not =0$. The virial
(\ref{vir1}) satisfies%
\begin{equation}
\mathrm{G}=\frac{2}{-\mathrm{i}\zeta}\mathcal{T}_{0}\left(  \dot{Q}\right)
;\text{ \ \ }\mathrm{\dot{G}}\left(  t\right)  =e^{2\operatorname{Im}\zeta
t}\mathrm{G}(0). \label{vir9}%
\end{equation}
Taking the time-derivative of (\ref{vir9}) and using (\ref{radis2a}) yields
the identity%
\begin{equation}
4\frac{\operatorname{Im}\zeta}{-\mathrm{i}\zeta}\mathcal{T}_{0}\left(  \dot
{Q}\right)  =\mathrm{\dot{G}}=2\mathcal{T}_{0}\left(  \dot{Q}\right)
-2\mathcal{V}_{0}\left(  Q\right)  -\frac{2}{-\mathrm{i}\zeta}\mathcal{R}%
\left(  \dot{Q}\right)  -(2\theta\dot{Q},Q), \label{vir10}%
\end{equation}
where $\mathcal{R}\left(  \dot{Q}\right)  =\frac{1}{2}\left(  \dot{Q},\beta
R\dot{Q}\right)  $. The energy balance equation (\ref{radis4}) implies%
\begin{equation}
-2\mathcal{R}\left(  \dot{Q}\right)  =\partial_{t}\left(  \mathcal{T}%
_{0}\left(  \dot{Q}\right)  +\mathcal{V}_{0}\left(  Q\right)  \right)
=2\operatorname{Im}\zeta\left(  \mathcal{T}_{0}\left(  \dot{Q}\right)
+\mathcal{V}_{0}\left(  Q\right)  \right)  . \label{vir11}%
\end{equation}
Combining the identities (\ref{vir10}) and (\ref{vir11}) yields immediately
the identity%
\begin{equation}
\frac{\operatorname{Re}\zeta}{\zeta}\mathcal{T}_{0}\left(  \dot{Q}\right)
=\frac{\operatorname{Re}\zeta}{\zeta}\mathcal{V}_{0}\left(  Q\right)
+(\theta\dot{Q},Q). \label{vir12}%
\end{equation}
From (\ref{vir12}) and since $\theta^{\mathrm{T}}=-\theta$ we derive the
identity%
\begin{equation}
\mathcal{T}_{0}\left(  \dot{Q}\right)  =\mathcal{V}_{0}\left(  Q\right)
-\frac{\zeta}{\operatorname{Re}\zeta}(\dot{Q},\theta Q) \label{vir14}%
\end{equation}
if $\operatorname{Re}\zeta\not =0$ or $(\dot{Q},\theta Q)=0$ if
$\operatorname{Re}\zeta=0$. Suppose now that $\operatorname{Re}\zeta\not =0$.
Then, by taking the imaginary and real parts of the identity (\ref{vir14}), we
find that%
\begin{equation}
\mathcal{T}_{0}\left(  \dot{Q}\right)  =\mathcal{V}_{0}\left(  Q\right)
-\operatorname{Re}\left(  \dot{Q},\theta Q\right)  -\left(  \frac
{\operatorname{Im}\zeta}{\operatorname{Re}\zeta}\right)  ^{2}\operatorname{Re}%
\left(  \dot{Q},\theta Q\right)  . \label{vir16}%
\end{equation}
The rest of the proof of Theorem \ref{tvir} now immediately follows from
(\ref{vir16}) and (\ref{radis4_0}).
\end{proof}

\begin{proof}
[Proof of Corollary \ref{tvir_eqp}]Corollary \ref{tvir_eqp} follows
immediately from the virial theorem \ref{tvir}.
\end{proof}

\subsection{Spectral analysis of the system eigenmodes}

The proofs of statements in Section \ref{sspan} are contained here.

\subsubsection{Standard versus pencil formulations of the spectral problems}

\begin{proof}
[Proof of Proposition \ref{ppfsp}]The factorization (\ref{pfsp8}) follows from
formula (\ref{mab2}) in Appendix \ref{apxsc} and (\ref{radis6a1}),
(\ref{radis6b}), (\ref{radis8})--(\ref{radis8b}), (\ref{pfsp2})\ since%
\begin{gather*}
\zeta I-A\left(  \beta\right)  =\left[
\begin{array}
[c]{cc}%
\zeta\mathbf{1}-\Omega_{\mathrm{p}}+\mathrm{i}\mathsf{\tilde{R}} &
\mathrm{i}\Phi^{\mathrm{T}}\\
-\mathrm{i}\Phi & \zeta\mathbf{1}%
\end{array}
\right] \\
=\left[
\begin{array}
[c]{cc}%
\mathbf{1} & \zeta^{-1}\mathrm{i}\Phi^{\mathrm{T}}\\
0 & \mathbf{1}%
\end{array}
\right]  \left[
\begin{array}
[c]{cc}%
\zeta\mathbf{1}-\Omega_{\mathrm{p}}+\mathrm{i}\mathsf{\tilde{R}}-\zeta
^{-1}\mathrm{i}\Phi^{\mathrm{T}}\left(  -\mathrm{i}\Phi\right)  & 0\\
0 & \zeta\mathbf{1}%
\end{array}
\right]  \left[
\begin{array}
[c]{cc}%
\mathbf{1} & 0\\
-\zeta^{-1}\mathrm{i}\Phi & \mathbf{1}%
\end{array}
\right] \\
=\left[
\begin{array}
[c]{cc}%
K_{\mathrm{p}} & \zeta^{-1}\mathrm{i}\Phi^{\mathrm{T}}\\
0 & \mathbf{1}%
\end{array}
\right]  \left[
\begin{array}
[c]{cc}%
\zeta^{-1}\mathbf{1} & 0\\
0 & \zeta\mathbf{1}%
\end{array}
\right]  \left[
\begin{array}
[c]{cc}%
C(\zeta,\beta) & 0\\
0 & \mathbf{1}%
\end{array}
\right]  \left[
\begin{array}
[c]{cc}%
K_{\mathrm{p}}^{\mathrm{T}} & 0\\
-\zeta^{-1}\mathrm{i}\Phi & \mathbf{1}%
\end{array}
\right]  .
\end{gather*}
This completes the proof.
\end{proof}

\begin{proof}
[Proof of Corollary \ref{cpfsp}]The proof of this corollary follows
immediately from the factorization (\ref{pfsp8}) using the facts
$K_{\mathrm{p}}^{\mathrm{T}}=K_{\mathrm{p}}=\sqrt{\alpha}^{-1}$, $\Phi
\sqrt{\alpha}=\sqrt{\eta}$ and%
\[
\left[
\begin{array}
[c]{cc}%
K_{\mathrm{p}}^{\mathrm{T}} & 0\\
-\zeta^{-1}\mathrm{i}\Phi & \mathbf{1}%
\end{array}
\right]  ^{-1}=\left[
\begin{array}
[c]{cc}%
\sqrt{\alpha} & 0\\
\zeta^{-1}\mathrm{i}\Phi\sqrt{\alpha} & \mathbf{1}%
\end{array}
\right]  .
\]

\end{proof}

\subsubsection{On the spectrum of the system operator}

\begin{proof}
[Proof of Proposition \ref{pssym}]The proof of this proposition follows
immediately from elementary properties of the determinant and using the
fundamental property of the system operator $A\left(  \beta\right)
=\Omega-\mathrm{i}\beta B$ in (\ref{sasp1}), namely, $A\left(  \beta\right)
^{\ast}=-A\left(  \beta\right)  ^{\mathrm{T}}$ which implies $-\overline
{A\left(  \beta\right)  }=A\left(  \beta\right)  $.
\end{proof}

\subsubsection{Eigenvalue bounds and modal dichotomy}

\begin{proof}
[Proof of Proposition \ref{pevbd}]The second statement follows immediately
from (\ref{sasp1}) since if $A\left(  \beta\right)  w=\zeta w$ and $w\not =0$
then%
\[
\operatorname{Re}\zeta=\frac{\operatorname{Re}\left(  w,A\left(  \beta\right)
w\right)  }{\left(  w,w\right)  }=\frac{\left(  w,\Omega w\right)  }{\left(
w,w\right)  },\text{ }-\operatorname{Im}\zeta=-\frac{\operatorname{Im}\left(
w,A\left(  \beta\right)  w\right)  }{\left(  w,w\right)  }=\frac{\left(
w,\beta Bw\right)  }{\left(  w,w\right)  }\geq0.
\]
From this and using (\ref{sasp1})--(\ref{sasp3}), the first and third
statements follow immediately from Proposition \ref{apxpebd} in Appendix
\ref{apxebd}. This completes the proof.
\end{proof}

\begin{proof}
[Proof of Corollary \ref{cssym}]The proof of this corollary follows
immediately from Propositions \ref{pssym} and \ref{pevbd}.
\end{proof}

\begin{proof}
[Proof of Theorem \ref{tmdic}]Suppose that $\beta>2\frac{\omega_{\text{max}}%
}{b_{\text{min}}}$. It follows from Theorem \ref{apxtmd} in Appendix
\ref{apxebd} that the $2N\times2N$ matrix $A\left(  \beta\right)  $ satisfies
the hypotheses of the theorem and together with (\ref{sasp1}), (\ref{ebmd1})
implies $\operatorname{rank}\left(  \operatorname{Im}A\left(  \beta\right)
\right)  =\operatorname{rank}\left(  B\right)  =N_{R}$ and the sets
$\sigma_{0}\left(  A\left(  \beta\right)  \right)  $, $\sigma_{1}\left(
A\left(  \beta\right)  \right)  $ in (\ref{mdss}) satisfy%
\[
\sigma\left(  A\left(  \beta\right)  \right)  =\sigma_{0}\left(  A\left(
\beta\right)  \right)  \cup\sigma_{1}\left(  A\left(  \beta\right)  \right)
,\text{ \ \ }\sigma_{0}\left(  A\left(  \beta\right)  \right)  \cap\sigma
_{1}\left(  A\left(  \beta\right)  \right)  =\emptyset.
\]
Moreover, Theorem \ref{apxtmd} gives the existence of unique projection
matrices $P_{0}\left(  \beta\right)  $, $P_{1}\left(  \beta\right)  $ such
that in the space $H=%
%TCIMACRO{\U{2102} }%
%BeginExpansion
\mathbb{C}
%EndExpansion
^{2N}$ the subspaces $H_{h\ell}\left(  \beta\right)  =\operatorname{Ran}%
P_{1}\left(  \beta\right)  $ and $H_{\ell\ell}\left(  \beta\right)
=\operatorname{Ran}P_{0}\left(  \beta\right)  $ have exactly the desired
properties in Theorem \ref{tmdic} with their dimensions satisfying%
\[
\dim H_{h\ell}\left(  \beta\right)  =\operatorname{rank}P_{1}\left(
\beta\right)  =N_{R},\text{ }\dim H_{\ell\ell}\left(  \beta\right)
=\operatorname{rank}P_{0}\left(  \beta\right)  =2N-N_{R}.
\]
This completes the proof.
\end{proof}

\begin{proof}
[Proof of Corollary \ref{cmdic}]First, if $w$ is an eigenvector of $A\left(
\beta\right)  $ with corresponding eigenvalue $\zeta$ then Proposition
\ref{pevbd} and (\ref{sasp3}) imply%
\begin{equation}
-\operatorname{Im}\zeta\geq0,\text{ \ \ }\left\vert \operatorname{Re}%
\xi\right\vert =\left\vert \frac{\left(  w,\Omega w\right)  }{\left(
w,w\right)  }\right\vert \leq\left\Vert \Omega\right\Vert =\omega_{\text{max}%
}. \label{pcmdic1}%
\end{equation}
Assume now that $\beta>2\frac{\omega_{\text{max}}}{b_{\text{min}}}$. Then by
(\ref{sasp1})--(\ref{sasp3}), (\ref{pcmdic1}), and Theorem \ref{tmdic} the
proof of the first part of Corollary \ref{cmdic} follows immediately from
Corollary \ref{apxcmd} in Appendix \ref{apxebd}. The rest of the proof now
follows immediately from (\ref{pcmdic1}), formula (\ref{radis12_1}), and the
first part of Corollary \ref{cmdic}. This completes the proof.
\end{proof}

\begin{proof}
[Proof of Theorem \ref{cpodr}]Consider the high-loss and low-loss eigenpairs
$\zeta_{j}\left(  \beta\right)  $, $w_{j}\left(  \beta\right)  $, $1\leq
j\leq2N$ for $\beta\gg1$ introduced in Section \ref{smdhlr}. The have the
properties that $\left\{  w_{j}\left(  \beta\right)  \right\}  _{j=1}^{2N}$
are a basis of eigenvectors for the system operator $A\left(  \beta\right)  $
satisfying
\begin{gather*}
A\left(  \beta\right)  w_{j}\left(  \beta\right)  =\zeta_{j}\left(
\beta\right)  w_{j}\left(  \beta\right)  ,\ \ 1\leq j\leq2N;\\
\lim_{\beta\rightarrow\infty}-\operatorname{Im}\zeta_{j}\left(  \beta\right)
=\infty,\text{ \ \ }1\leq j\leq N_{R};\text{ }\lim_{\beta\rightarrow\infty
}-\operatorname{Im}\zeta_{j}\left(  \beta\right)  =0,\text{ \ \ }N_{R}+1\leq
j\leq2N.
\end{gather*}
The proof of Theorem \ref{cpodr} follows from these facts, Theorem
\ref{tmdic}, and Corollary \ref{cmdic}.
\end{proof}

\subsection{Overdamping analysis}

The proofs of statements in Section \ref{sodan} are contained here. Recall,
that we are assuming from now on that Condition \ref{cndod} is true i.e.,
$\theta=0.$

\begin{proof}
[Proof of Proposition \ref{pspre}]Using the block representation of the
matrices $\Omega$, $B$ in (\ref{radis8a}) and the result (\ref{mab5}) in
Appendix \ref{apxsc} on the determinant of block matrices with commuting block
entries we find that for any $\zeta\in%
%TCIMACRO{\U{2102} }%
%BeginExpansion
\mathbb{C}
%EndExpansion
$ we have%
\begin{align*}
\det\left(  \zeta\mathbf{1}-\Omega\right)   &  =\det\left(  \left[
\begin{array}
[c]{cc}%
\zeta\mathbf{1} & \mathrm{i}\Phi^{\mathrm{T}}\\
-\mathrm{i}\Phi & \zeta\mathbf{1}%
\end{array}
\right]  \right)  =\det\left(  \zeta^{2}\mathbf{1-}\Phi^{\mathrm{T}}%
\Phi\right)  =\det\left(  \zeta^{2}\mathbf{1}-\alpha^{-1}\eta\right)  ,\\
\det\left(  \zeta\mathbf{1}-B\right)   &  =\det\left(  \left[
\begin{array}
[c]{cc}%
\zeta\mathbf{1-}\mathsf{\tilde{R}} & 0\\
0 & \zeta\mathbf{1}%
\end{array}
\right]  \right)  =\zeta^{N}\det\left(  \zeta\mathbf{1-}\mathsf{\tilde{R}%
}\right)  =\zeta^{N}\det\left(  \zeta\mathbf{1-}\alpha^{-1}R_{\mathrm{H}%
}\right)  ,
\end{align*}
since by (\ref{radis6a1}), (\ref{radis6b}), (\ref{radis8b}) we have
$\Phi^{\mathrm{T}}\Phi=\sqrt{\alpha}^{-1}\eta\sqrt{\alpha}^{-1}$ and
$\mathsf{\tilde{R}}=\sqrt{\alpha}^{-1}R_{\mathrm{H}}\sqrt{\alpha}^{-1}$. The
proof now follows from these facts.
\end{proof}

\subsubsection{Complete and partial overdamping}

\begin{proof}
[Proof of Proposition \ref{pfuineq}]The last three statements have already
been proved. The first two statements are proved using Proposition \ref{pspre}
and the following series of inequalities which just use elementary properties
of the operator norm $\left\Vert \cdot\right\Vert $ for Hermitian matrices
\begin{gather}
0\leq\frac{\left(  q,\eta q\right)  }{\left(  q,\alpha q\right)  }%
=\frac{\left(  \sqrt{\alpha}q,\sqrt{\alpha}^{-1}\eta\sqrt{\alpha}^{-1}%
\sqrt{\alpha}q\right)  }{\left(  \sqrt{\alpha}q,\sqrt{\alpha}q\right)  }%
\leq\left\Vert \sqrt{\alpha}^{-1}\eta\sqrt{\alpha}^{-1}\right\Vert =\max
\sigma\left(  \alpha^{-1}\eta\right)  ,\label{lod7_1}\\
\frac{\left(  q,Mq\right)  }{\left(  q,\alpha q\right)  }=\frac{\left(
\sqrt{\alpha}q,\sqrt{\alpha}^{-1}M\sqrt{\alpha}^{-1}\sqrt{\alpha}q\right)
}{\left(  \sqrt{\alpha}q,\sqrt{\alpha}q\right)  }\geq\left\Vert \left(
\sqrt{\alpha}^{-1}M\sqrt{\alpha}^{-1}\right)  ^{-1}\right\Vert ^{-1}%
=\min\sigma\left(  \alpha^{-1}M\right)  , \label{lod7_1a}%
\end{gather}
where $0\leq M=\eta$ or $R$ and the inequality (\ref{lod7_1a})\ requires the
hypothesis that $\eta$ or $R$, respectively, be invertible so that $M>0$. The
completes the proof.
\end{proof}

\begin{proof}
[Proof of Theorem \ref{tcovd}]Suppose $\beta\geq2\frac{\omega_{\text{max}}%
}{b_{\text{min}}}$and $N_{R}=N $, where $N_{R}=\operatorname{rank}R$. If
$\zeta$ was an eigenvalue of $A\left(  \beta\right)  $ with $\operatorname{Re}%
\zeta\not =0$ then by Proposition \ref{pfuineq} we have
\[
-\operatorname{Im}\zeta=\frac{\beta}{2}\frac{\left(  q,Rq\right)  }{\left(
q,\alpha q\right)  }\geq\frac{\beta}{2}b_{\text{min}}\geq\omega_{\text{max}%
}\geq\left\vert \zeta\right\vert >-\operatorname{Im}\zeta,
\]
a contradiction. Therefore, all the eigenvalues of the system operator
$A\left(  \beta\right)  $ are real. \ This completes the proof.
\end{proof}

\begin{proof}
[Proof of Theorem \ref{tpovd}]Suppose $\beta\geq2\frac{\omega_{\text{max}}%
}{b_{\text{min}}}$and $N_{R}<N $. Then by Theorem \ref{tmdic} and Corollary
\ref{cmdic}, if $\zeta\in\sigma\left(  A\left(  \beta\right)  |_{H_{h\ell
}\left(  \beta\right)  }\right)  $ then $-\operatorname{Im}\zeta
>\omega_{\text{max}}$ and so by Proposition \ref{pfuineq} we must have
$\operatorname{Re}\zeta=0$. Therefore, all the eigenvalues of $A\left(
\beta\right)  $ in $\sigma\left(  A\left(  \beta\right)  |_{H_{h\ell}\left(
\beta\right)  }\right)  $ are real.
\end{proof}

\subsubsection{Selective overdamping\label{pspad}}

\begin{proof}
[Proof of Theorem \ref{thmod}]If $v_{j}\left(  t,\beta\right)  =w_{j}\left(
\beta\right)  e^{-\mathrm{i}\zeta_{j}\left(  \beta\right)  t}$ is a high-loss
eigenmode from (\ref{pod20}) then by the modal dichotomy (\ref{cpodr}) we must
have $\zeta_{j}\left(  \beta\right)  \in\sigma\left(  A\left(  \beta\right)
|_{H_{h\ell}\left(  \beta\right)  }\right)  $ for $\beta\gg1$ (i.e., for
$\beta$ sufficiently large). By Theorem \ref{tpovd} it follows that
$\operatorname{Re}\zeta_{j}\left(  \beta\right)  =0$ for $\beta\gg1$.
Therefore, $v_{j}\left(  t,\beta\right)  $ is overdamped for $\beta\gg1$. This
completes the proof.
\end{proof}

We will now assume in the rest of the proofs that Condition \ref{cndnd} is
true i.e., $\operatorname{Ker}\eta\cap\operatorname{Ker}R=\left\{  0\right\}
$.

\begin{proof}
[Proof of Theorem \ref{tlmod}]The proof of this theorem was started in Section
\ref{spad2}. It remains to show that (\ref{llod6}) holds. First, we have that
$\lim_{\beta\rightarrow\infty}w_{j}\left(  \beta\right)  =\mathring{w}_{j}$
with $\Omega_{1}\mathring{w}_{j}=0$ and $\left\Vert \mathring{w}%
_{j}\right\Vert =1$ by the properties of the low-loss eigenmodes in
(\ref{pod20}) and since (\ref{llod1}), (\ref{llod2}) is true for these modes
together with the hypothesis $\rho_{j}=0$.

Now we use the representation (\ref{llod3}) and (\ref{plmodp6}) for
$w_{j}\left(  \beta\right)  $ and $\operatorname{Ker}\Omega_{1}$,
respectively, to imply that as $\beta\rightarrow\infty$
\begin{equation}
w_{j}\left(  \beta\right)  =\left[
\begin{array}
[c]{c}%
-\mathrm{i}\zeta_{j}\left(  \beta\right)  \sqrt{\alpha}q_{j}\left(
\beta\right) \\
\sqrt{\eta}q_{j}\left(  \beta\right)
\end{array}
\right]  \rightarrow\mathring{w}_{j}=%
\begin{bmatrix}
0\\
\mathring{\psi}%
\end{bmatrix}
,\text{ }-\mathrm{i}\zeta_{j}\left(  \beta\right)  \sqrt{\alpha}q_{j}\left(
\beta\right)  \rightarrow0,\text{ }\sqrt{\eta}q_{j}\left(  \beta\right)
\rightarrow\mathring{\psi} \label{tlmod1}%
\end{equation}
for some $\mathring{\psi}\in%
%TCIMACRO{\U{2102} }%
%BeginExpansion
\mathbb{C}
%EndExpansion
^{N}$. It follows from (\ref{tlmod1}) and the fact $\lim_{\beta\rightarrow
\infty}\operatorname{Im}\zeta_{j}\left(  \beta\right)  =0$ that for any fixed
time $t$ we have as $\beta\rightarrow\infty$%
\begin{align*}
\mathcal{T}\left(  \dot{Q}_{j}\left(  t\,,\beta\right)  ,Q_{j}\left(
t\,,\beta\right)  \right)   &  =\frac{1}{2}\left(  \dot{Q}_{j}\left(
t\,,\beta\right)  ,\alpha\dot{Q}_{j}\left(  t\,,\beta\right)  \right) \\
&  =\frac{1}{2}\left(  -\mathrm{i}\zeta_{j}\left(  \beta\right)  \sqrt{\alpha
}q_{j}\left(  \beta\right)  ,-\mathrm{i}\zeta_{j}\left(  \beta\right)
\sqrt{\alpha}q_{j}\left(  \beta\right)  \right)  e^{2\operatorname{Im}%
\zeta_{j}\left(  \beta\right)  t}\rightarrow0,\\
\mathcal{V}\left(  \dot{Q}_{j}\left(  t\,,\beta\right)  ,Q_{j}\left(
t\,,\beta\right)  \right)   &  =\frac{1}{2}\left(  Q_{j}\left(  t\,,\beta
\right)  ,\eta Q_{j}\left(  t\,,\beta\right)  \right) \\
&  =\frac{1}{2}\left(  \sqrt{\eta}q_{j}\left(  \beta\right)  ,\sqrt{\eta}%
q_{j}\left(  \beta\right)  \right)  e^{2\operatorname{Im}\zeta_{j}\left(
\beta\right)  t}\rightarrow\frac{1}{2}\left(  \mathring{\psi},\mathring{\psi
}\right)  ,\\
U\left[  v_{j}\left(  t,\beta\right)  \right]   &  =\frac{1}{2}\left(
w_{j}\left(  \beta\right)  ,w_{j}\left(  \beta\right)  \right)
e^{2\operatorname{Im}\zeta_{j}\left(  \beta\right)  t}\rightarrow\frac{1}%
{2}\left(  \mathring{w}_{j},\mathring{w}_{j}\right)  =\frac{1}{2}.
\end{align*}
This proves (\ref{llod6}) since by (\ref{tlmod1}) we have $\left(
\mathring{w}_{j},\mathring{w}_{j}\right)  =\left(  \mathring{\psi}%
,\mathring{\psi}\right)  $. This completes the proof.
\end{proof}

\begin{proof}
[Proof of Corollary \ref{clmod}]This corollary follows immediately from
Theorem \ref{tlmod} and the fact that the $\left\{  \mathring{w}_{j}\right\}
_{j=N_{R}+1}^{2N}$ are an orthonormal basis for $\operatorname{Ker}B$ which
diagonalize the self-adjoint operator $\Omega_{1}$ such that (\ref{llod2}) is satisfied.
\end{proof}

\begin{proof}
[Proof of Theorem \ref{tnlm}]This theorem follows immediately from Theorems
\ref{thmod} and \ref{tlmod}, Corollary \ref{clmod}, and Proposition
\ref{pqaollm} (which is proved below).
\end{proof}

\begin{proof}
[Proof of Theorem \ref{tseod}]This theorem follows immediately from Theorems
\ref{tnlm}, \ref{tlmod} and Proposition \ref{pqaollm}.
\end{proof}

\begin{proof}
[Proof of Proposition \ref{plmod}]Let $\kappa$ denote the number of low-loss
eigenmodes of the canonical system (\ref{radis8}) in (\ref{pod20}) which are
overdamped for $\beta\gg1$. \ Then by Corollary \ref{clmod} we know that%
\[
\kappa=\dim\operatorname{Ker}\Omega_{1}\text{.}%
\]
We will now prove that
\begin{equation}
\kappa=N_{R} \label{plmodp0}%
\end{equation}
where by (\ref{radis10a}), $N_{R}=\operatorname{rank}R=\operatorname{rank}B$.
We will use the elementary facts from linear algebra that since $\eta$,
$R\geq0$ are $N\times N$ matrices which are positive semidefinite then
\begin{gather}%
%TCIMACRO{\U{2102} }%
%BeginExpansion
\mathbb{C}
%EndExpansion
^{N}=\operatorname{Ker}\eta\oplus\operatorname{Ran}\eta,\text{ \ \ }%
\operatorname{Ker}\eta=\operatorname{Ker}\sqrt{\eta},\text{ \ \ }%
\operatorname{Ran}\eta=\operatorname{Ran}\sqrt{\eta},\label{plmodp1}\\
\left(  \operatorname{Ran}R\right)  ^{\bot}=\operatorname{Ker}R\text{ and
}\left(  \operatorname{Ran}\eta\right)  ^{\bot}=\operatorname{Ker}%
\eta\label{plmodp3}%
\end{gather}
Another elementary fact from linear algebra which we need is that if $S_{1}$
and $S_{2}$ are subspaces in the same Hilbert space then%
\[
\left(  S_{1}\cap S_{2}\right)  ^{\bot}=S_{1}^{\bot}+S_{2}^{\bot},
\]
where $S_{1}^{\bot}+S_{2}^{\bot}=\left\{  u+v:u\in S_{1}^{\bot}\text{ and
}v\in S_{2}^{\bot}\right\}  $. Thus with $S_{1}=\operatorname{Ran}\eta$,
$S_{2}=\operatorname{Ran}R$ we have%
\begin{equation}
(\operatorname{Ran}\eta\cap\operatorname{Ran}R)^{\bot}=\operatorname{Ker}%
\eta\oplus\operatorname{Ker}R, \label{plmodp4}%
\end{equation}
where the sum is direct by Condition \ref{cndnd}, i.e., the hypothesis
$\operatorname{Ker}\eta\cap\operatorname{Ker}R=\left\{  0\right\}  $. In
particular, this implies
\begin{equation}
\dim\left(  \operatorname{Ran}\eta\cap\operatorname{Ran}R\right)
=N-\dim\operatorname{Ker}\eta-\dim\operatorname{Ker}R=N_{R}-\dim
\operatorname{Ker}\eta. \label{plmodp5}%
\end{equation}

The proof of (\ref{plmodp0}) will be carried out in a series of steps which we
outline now before we begin. First, we will prove that
\begin{equation}
\operatorname{Ker}\Omega_{1}=\left\{  \left[
\begin{array}
[c]{c}%
0\\
\psi
\end{array}
\right]  :\sqrt{\eta}\psi\in\operatorname{Ran}\eta\cap\operatorname{Ran}%
R\right\}  , \label{plmodp6}%
\end{equation}
where the block form is with respect to the decomposition $H=H_{\mathrm{p}%
}\oplus H_{\mathrm{q}}$, $H=%
%TCIMACRO{\U{2102} }%
%BeginExpansion
\mathbb{C}
%EndExpansion
^{2N}$, $H_{\mathrm{p}}=H_{\mathrm{q}}=%
%TCIMACRO{\U{2102} }%
%BeginExpansion
\mathbb{C}
%EndExpansion
^{N}$ described in Section \ref{stvspen}, i.e., the form (\ref{pfsp8b}). Next,
it follows immediately from (\ref{plmodp1}) that
\begin{equation}
\left\{  \psi\in%
%TCIMACRO{\U{2102} }%
%BeginExpansion
\mathbb{C}
%EndExpansion
^{N}:\sqrt{\eta}\psi\in\operatorname{Ran}\eta\cap\operatorname{Ran}R\right\}
=\operatorname{Ker}\eta\oplus\left\{  \phi\in\operatorname{Ran}\eta:\sqrt
{\eta}\phi\in\operatorname{Ran}\eta\cap\operatorname{Ran}R\right\}  .
\label{plmodp7}%
\end{equation}
Finally, we prove that
\begin{equation}
\dim\left\{  \phi\in\operatorname{Ran}\eta:\sqrt{\eta}\phi\in
\operatorname{Ran}\eta\cap\operatorname{Ran}R\right\}  =N_{R}-\dim
\operatorname{Ker}\eta. \label{plmodp8}%
\end{equation}
The proof of (\ref{plmodp0}) follows immediately from (\ref{plmodp6}%
)--(\ref{plmodp8}).

We begin by computing the representation (\ref{plmodp6}) for
$\operatorname{Ker}\Omega_{1}$. \ First, by its definition from (\ref{pod7}),
$\Omega_{1}$ is the restriction of the operator $\Omega$ to the no-loss
subspace $H_{B}^{\bot}=\operatorname{Ker}B$, that is,%
\begin{equation}
\Omega_{1}=P_{B}^{\bot}\Omega P_{B}^{\bot}|_{H_{B}^{\bot}}:H_{B}^{\bot
}\rightarrow H_{B}^{\bot}\text{,} \label{plmodp9}%
\end{equation}
where $P_{B}^{\bot}$ is the orthogonal projection onto $H_{B}^{\bot}$ in the
Hilbert space $H=%
%TCIMACRO{\U{2102} }%
%BeginExpansion
\mathbb{C}
%EndExpansion
^{2N}$ with standard inner product $\left(  \cdot,\cdot\right)  $. \ Denote by
$P_{\mathsf{\tilde{R}}}^{\bot}$ the orthogonal projection onto
$\operatorname{Ker}\mathsf{\tilde{R}}$ in the Hilbert space $H_{\mathrm{p}}=%
%TCIMACRO{\U{2102} }%
%BeginExpansion
\mathbb{C}
%EndExpansion
^{N}$ with the standard inner product. \ Then it follows from the block
representation of $\Omega$ and $B$ in (\ref{pod7}) with respect to the
decomposition $H=H_{\mathrm{p}}\oplus H_{\mathrm{q}}$ that $P_{B}^{\bot}$ and
$P_{B}^{\bot}\Omega P_{B}^{\bot}$ with respect to this decomposition are the
block operators
\begin{equation}
P_{B}^{\bot}=%
\begin{bmatrix}
P_{\mathsf{\tilde{R}}}^{\bot} & 0\\
0 & \mathbf{1}%
\end{bmatrix}
,\text{ }P_{B}^{\bot}\Omega P_{B}^{\bot}=%
\begin{bmatrix}
P_{\mathsf{\tilde{R}}}^{\bot} & 0\\
0 & \mathbf{1}%
\end{bmatrix}%
\begin{bmatrix}
0 & -\mathrm{i}\Phi^{\mathrm{T}}\\
\mathrm{i}\Phi & 0
\end{bmatrix}%
\begin{bmatrix}
P_{\mathsf{\tilde{R}}}^{\bot} & 0\\
0 & \mathbf{1}%
\end{bmatrix}
=%
\begin{bmatrix}
0 & -\mathrm{i}P_{\mathsf{\tilde{R}}}^{\bot}\Phi^{\mathrm{T}}\\
\mathrm{i}\Phi P_{\mathsf{\tilde{R}}}^{\bot} & 0
\end{bmatrix}
. \label{plmodp10}%
\end{equation}
From (\ref{plmodp9}), (\ref{plmodp10}) it follows that%
\[
\operatorname{Ker}\Omega_{1}=\left\{  \left[
\begin{array}
[c]{c}%
\varphi\\
\psi
\end{array}
\right]  \in\operatorname{Ker}\mathsf{\tilde{R}}\oplus H_{\mathrm{q}%
}:-\mathrm{i}P_{\mathsf{\tilde{R}}}^{\bot}\Phi^{\mathrm{T}}\psi=0\text{ and
}\mathrm{i}\Phi P_{\mathsf{\tilde{R}}}^{\bot}\varphi=0\right\}  .
\]
But if $\mathrm{i}\Phi P_{\mathsf{\tilde{R}}}^{\bot}\varphi=0$ with
$\varphi\in\operatorname{Ker}\mathsf{\tilde{R}}$ then $P_{\mathsf{\tilde{R}}%
}^{\bot}\varphi=\varphi$ by definition of $P_{\mathsf{\tilde{R}}}^{\bot}$ and
hence $\mathsf{\tilde{R}}\varphi=0$ and $\Phi\varphi=0$ which implies $\eta
K_{\mathrm{p}}^{\mathrm{T}}\varphi=0$ and\ $RK_{\mathrm{p}}^{\mathrm{T}%
}\varphi=0$ since by (\ref{radis8b}) and (\ref{radis6a1}) we have
$\Phi=K_{\mathrm{q}}K_{\mathrm{p}}^{\mathrm{T}}$, $\mathsf{\tilde{R}%
}=K_{\mathrm{p}}RK_{\mathrm{p}}^{\mathrm{T}}$, $K_{\mathrm{p}}=\sqrt{\alpha
}^{-1}$, and $K_{\mathrm{q}}=\sqrt{\eta}$. By the hypothesis of Condition
\ref{cndnd} this implies that $\varphi=0$. Also, since $\mathsf{\tilde{R}}$ is
a Hermitian matrix it follows that $I-P_{\mathsf{\tilde{R}}}^{\bot}$ is the
orthogonal projection onto the $\operatorname{Ran}\mathsf{\tilde{R}}$ so that
$P_{\mathsf{\tilde{R}}}^{\bot}\Phi^{\mathrm{T}}\psi=0$ is equivalent to
$\Phi^{\mathrm{T}}\psi\in\operatorname{Ran}\mathsf{\tilde{R}}$ and this is
equivalent to $\sqrt{\eta}\psi\in$ $\operatorname{Ran}R$ by properties
(\ref{radis6b}). But $\sqrt{\eta}\psi\in$ $\operatorname{Ran}R$ is equivalent
to $\sqrt{\eta}\psi\in$ $\operatorname{Ran}R\cap\operatorname{Ran}\sqrt{\eta}$
which by (\ref{plmodp1}) is equivalent to $\sqrt{\eta}\psi\in$
$\operatorname{Ran}R\cap\operatorname{Ran}\eta$. This proves (\ref{plmodp6}).

Finally, we will now prove (\ref{plmodp8}). It follows from the fact that
$\eta$ is a Hermitian matrix and (\ref{plmodp1}) that the restriction of the
operator$\sqrt{\eta}$ to $\operatorname{Ran}\eta$, i.e., $\sqrt{\eta
}|_{\operatorname{Ran}\eta}:\operatorname{Ran}\eta\rightarrow
\operatorname{Ran}\eta$, is an invertible operator and that
\begin{equation}
\left\{  \phi\in\operatorname{Ran}\eta:\sqrt{\eta}\phi\in\operatorname{Ran}%
\eta\cap\operatorname{Ran}R\right\}  =\sqrt{\eta}|_{\operatorname{Ran}\eta
}^{-1}\left(  \operatorname{Ran}\eta\cap\operatorname{Ran}R\right)  .
\label{plmodp11}%
\end{equation}
From this it follows that%
\[
\dim\left\{  \phi\in\operatorname{Ran}\eta:\sqrt{\eta}\phi\in
\operatorname{Ran}\eta\cap\operatorname{Ran}R\right\}  =\dim\left(
\operatorname{Ran}\eta\cap\operatorname{Ran}R\right)  .
\]
The proof of (\ref{plmodp8}) immediately follows from this fact and
(\ref{plmodp5}). This completes the proof of the theorem.
\end{proof}

\section{Appendix: Schur complement and the Aitken formula\label{apxsc}}

Let
\begin{equation}
M=\left[
\begin{array}
[c]{cc}%
P & Q\\
R & S
\end{array}
\right]  \label{mab1}%
\end{equation}
be a square matrix represented in block form where $P$ and $S$ are square
matrices with the latter invertible, that is, $\left\Vert S^{-1}\right\Vert
<\infty$. Then the following \emph{Aitken block-diagonalization formula} holds
\cite[\S 0.1]{Zhang}, \cite[\S 29]{Aitken},%
\begin{equation}
M=\left[
\begin{array}
[c]{cc}%
P & Q\\
R & S
\end{array}
\right]  =\left[
\begin{array}
[c]{cc}%
\mathbf{1} & QS^{-1}\\
0 & \mathbf{1}%
\end{array}
\right]  \left[
\begin{array}
[c]{cc}%
P-QS^{-1}R & 0\\
0 & S
\end{array}
\right]  \left[
\begin{array}
[c]{cc}%
\mathbf{1} & 0\\
S^{-1}R & \mathbf{1}%
\end{array}
\right]  , \label{mab2}%
\end{equation}
where the matrix%
\begin{equation}
M/S=P-QS^{-1}R \label{mab3}%
\end{equation}
is known as the \emph{Schur complement} of $S$ in $M$. The Aitken formula
(\ref{mab2}) readily implies%
\begin{equation}
\det M=\det S\det\left(  P-QS^{-1}R\right)  . \label{mab4}%
\end{equation}
From this we can conclude that%
\begin{align}
\text{If }QS  &  =SQ\text{ then }\det M=\det\left(  SP-QR\right)
;\label{mab5}\\
\text{If }RS  &  =SR\text{ then }\det M=\det\left(  PS-QR\right)  .\nonumber
\end{align}
But the latter statement is true even if $\det S=0$ which is proved using a
limiting argument. \ Indeed, by substituting for $S$ in (\ref{mab1}) the small
perturbation $S_{\delta}=S+\delta I$ which is invertible for $0<\left\vert
\delta\right\vert \ll1$ and commutes with any matrix that $S$ commutes with
then statement \ref{mab5} is true for $S_{\delta}$ and the limit as
$\delta\rightarrow0$ yields the desired result.\qquad

\section{Appendix: \ Eigenvalue bounds and dichotomy for non-Hermitian
matrices\label{apxebd}}

In this appendix we discuss results pertaining to bounds on the eigenvalues of
non-Hermitian matrices in terms of their real and imaginary parts. We also
discuss a phenomenon, which we call \emph{modal dichotomy}, that occurs for
any non-Hermitian matrix with a non-invertible imaginary part which is large
in comparison to the real part of the matrix. \ The bounds and dichotomy
described here are extremely useful tools in studying the spectral properties
of dissipative systems especially for composite systems with high-loss and
lossless components.

Recall, any square matrix $M$ can be written as the sum of a Hermitian matrix
and a skew-Hermitian matrix,%
\begin{equation}
M=\operatorname{Re}M+\mathrm{i}\operatorname{Im}M,\text{ }\operatorname{Re}%
M=\frac{M+M^{\ast}}{2},\text{ }\operatorname{Im}M=\frac{M-M^{\ast}%
}{2\mathrm{i}}, \label{apxebd1}%
\end{equation}
the real and imaginary parts of $M$, respectively, and $M^{\ast}$ is the
conjugate transpose of the matrix $M$.

For the set of $n\times n$ matrices, denote by $\left\Vert \cdot\right\Vert $
the operator norm%
\begin{equation}
\left\Vert M\right\Vert =\sup_{x\not =0}\frac{\left\Vert Mx\right\Vert _{2}%
}{\left\Vert x\right\Vert _{2}}, \label{apxebd2}%
\end{equation}
where $\left\Vert \cdot\right\Vert _{2}=\sqrt{\left(  \cdot,\cdot\right)  }$
and $\left(  \cdot,\cdot\right)  $ denotes the standard Euclidean inner
product on $n\times1$ column vectors with entries in $%
%TCIMACRO{\U{2102} }%
%BeginExpansion
\mathbb{C}
%EndExpansion
$. \ Denote the spectrum of a square matrix $M$ by $\sigma\left(  M\right)  $.
Recall in the case $M$ is a Hermitian matrix, i.e., $M^{\ast}=M$, or a normal
matrix, i.e., $M^{\ast}M=MM^{\ast}$, that
\begin{equation}
\left\Vert M\right\Vert =\max_{\zeta\in\sigma\left(  M\right)  }\left\vert
\zeta\right\vert ,\text{ \ \ }\left\Vert \left(  \zeta I-M\right)
^{-1}\right\Vert ^{-1}=\operatorname{dist}\left(  \zeta,\sigma\left(
M\right)  \right)  , \label{apxebd3}%
\end{equation}
where $\operatorname{dist}\left(  \zeta,\sigma\left(  M\right)  \right)
=\inf\limits_{\lambda\in\sigma\left(  M\right)  }\left\vert \zeta
-\lambda\right\vert $ and by convention $\left\Vert \left(  \zeta I-M\right)
^{-1}\right\Vert ^{-1}=0$ if $\zeta\in$ $\sigma\left(  M\right)  $.

The first result which we will prove here is the following proposition which
gives bounds on the eigenvalues of a matrix $M$ in terms of $\operatorname{Re}%
M$ and $\operatorname{Im}M$.

\begin{proposition}
[eigenvalue bounds]\label{apxpebd}Let $M$ be a square matrix. \ Then%
\begin{equation}
\sigma\left(  M\right)  \subseteq\left\{  \zeta\in%
%TCIMACRO{\U{2102} }%
%BeginExpansion
\mathbb{C}
%EndExpansion
:\operatorname{dist}\left(  \zeta,\sigma\left(  \mathrm{i}\operatorname{Im}%
M\right)  \right)  \leq\left\Vert \operatorname{Re}M\right\Vert \right\}
\label{apxebd4}%
\end{equation}
In other words, the eigenvalues of $M$ lie in the union of the closed discs
whose centers are the eigenvalues of $\mathrm{i}\operatorname{Im}M$ with
radius $\left\Vert \operatorname{Re}M\right\Vert $. \ Moreover, the imaginary
part of any eigenvalue $\zeta$ of $M$ satisfies the inequality%
\begin{equation}
\left\Vert \operatorname{Im}M\right\Vert \geq\left\vert \operatorname{Im}%
\zeta\right\vert \geq\left\vert \gamma\right\vert -\left\Vert
\operatorname{Re}M\right\Vert \label{apxebd5}%
\end{equation}
for some $\gamma\in\sigma\left(  \operatorname{Im}M\right)  $ depending on
$\zeta$.
\end{proposition}

Before we proceed with the proof we will need the following perturbation
result from \cite[Sec. V.4, p. 291, Prob. 4.8]{Kato} for non-Hermitian
matrices which we prove here for completeness.

\begin{lemma}
If $M=M_{0}+E$ is a square matrix with $M_{0}$ normal then%
\begin{equation}
\sigma\left(  M\right)  \subseteq\left\{  \zeta\in%
%TCIMACRO{\U{2102} }%
%BeginExpansion
\mathbb{C}
%EndExpansion
:\operatorname{dist}\left(  \zeta,\sigma\left(  M_{0}\right)  \right)
\leq\left\Vert E\right\Vert \right\}  \label{apxebd6}%
\end{equation}

\end{lemma}

\begin{proof}
Let $M=M_{0}+E$ be a square matrix with $M_{0}$ normal. \ Then%
\[
\sigma\left(  M\right)  \cap\sigma\left(  M_{0}\right)  \subseteq\left\{
\zeta\in%
%TCIMACRO{\U{2102} }%
%BeginExpansion
\mathbb{C}
%EndExpansion
:\operatorname{dist}\left(  \zeta,\sigma\left(  M_{0}\right)  \right)
\leq\left\Vert E\right\Vert \right\}  .
\]
Let $\rho\left(  M_{0}\right)  =\left\{  \zeta\in%
%TCIMACRO{\U{2102} }%
%BeginExpansion
\mathbb{C}
%EndExpansion
:\zeta\not \in \sigma\left(  M_{0}\right)  \right\}  $, i.e., the resolvent
set of $M_{0}$. \ Then for any $\zeta\in\rho\left(  M_{0}\right)  $ we have%
\[
\zeta I-M=\zeta I-M_{0}-\left(  M-M_{0}\right)  =\left(  I-E\left(  \zeta
I-M_{0}\right)  ^{-1}\right)  \left(  \zeta I-M_{0}\right)  .
\]
Thus if $\left\Vert E\left(  \zeta I-M_{0}\right)  ^{-1}\right\Vert <1$ then
the matrix $\left(  I-E\left(  \zeta I-M_{0}\right)  ^{-1}\right)  $ is
invertible since the Neumann series $\sum_{j=0}^{\infty}T^{n}$ converges to
$\left(  I-T\right)  ^{-1}$ for any square matrix $T$ satisfying $\left\Vert
T\right\Vert <1$. \ This implies%
\[
\sigma\left(  M\right)  \cap\rho\left(  M_{0}\right)  \subseteq\left\{
\zeta\in\rho\left(  M_{0}\right)  :\left\Vert E\left(  \zeta I-M_{0}\right)
^{-1}\right\Vert \geq1\right\}  .
\]
But the inequality$\left\Vert E\left(  \zeta I-M_{0}\right)  ^{-1}\right\Vert
\leq\left\Vert E\right\Vert \left\Vert \left(  \zeta I-M_{0}\right)
^{-1}\right\Vert $ and the fact $M_{0}$ is a normal matrix implies%
\begin{align*}
\sigma\left(  M\right)  \cap\rho\left(  M_{0}\right)   &  \subseteq\left\{
\zeta\in\rho\left(  M_{0}\right)  :\left\Vert \left(  \zeta I-M_{0}\right)
^{-1}\right\Vert ^{-1}\leq\left\Vert E\right\Vert \right\} \\
&  \subseteq\left\{  \zeta\in%
%TCIMACRO{\U{2102} }%
%BeginExpansion
\mathbb{C}
%EndExpansion
:\operatorname{dist}\left(  \zeta,\sigma\left(  M_{0}\right)  \right)
\leq\left\Vert E\right\Vert \right\}  .
\end{align*}
Therefore,%
\[
\sigma\left(  M\right)  =\sigma\left(  M\right)  \cap\sigma\left(
M_{0}\right)  \cup\sigma\left(  M\right)  \cap\rho\left(  M_{0}\right)
\subseteq\left\{  \zeta\in%
%TCIMACRO{\U{2102} }%
%BeginExpansion
\mathbb{C}
%EndExpansion
:\operatorname{dist}\left(  \zeta,\sigma\left(  M_{0}\right)  \right)
\leq\left\Vert E\right\Vert \right\}  .
\]
This completes the proof.
\end{proof}

\begin{proof}
[Proof of Proposition \ref{apxpebd}]Let $M$ be any square matrix. \ Then
$M=-\mathrm{i}\operatorname{Im}M+\operatorname{Re}M$ and $-\mathrm{i}%
\operatorname{Im}M$ is a normal matrix. By the previous lemma this implies%
\[
\sigma\left(  M\right)  \subseteq\left\{  \zeta\in%
%TCIMACRO{\U{2102} }%
%BeginExpansion
\mathbb{C}
%EndExpansion
:\operatorname{dist}\left(  \zeta,\sigma\left(  \mathrm{i}\operatorname{Im}%
M\right)  \right)  \leq\left\Vert \operatorname{Re}M\right\Vert \right\}  .
\]
Now let $\zeta$ be any eigenvalue of $M$. \ Then for any corresponding
eigenvector $x$ of unit norm%
\begin{equation}
\left\vert \operatorname{Im}\zeta\right\vert =\left\vert \operatorname{Im}%
\left(  x,Mx\right)  \right\vert =\left\vert \left(  x,\operatorname{Im}%
Mx\right)  \right\vert \leq\left\Vert \operatorname{Im}M\right\Vert .
\label{papxpebd1}%
\end{equation}
Moreover, we know there exists an eigenvalue $\gamma$ of $\operatorname{Im}M$
(depending on\ $\zeta$) such that
\[
\left\vert \zeta-\mathrm{i}\gamma\right\vert =\operatorname{dist}\left(
\zeta,\sigma\left(  \mathrm{i}\operatorname{Im}M\right)  \right)
\leq\left\Vert \operatorname{Re}M\right\Vert .
\]
Hence
\begin{equation}
\left\Vert \operatorname{Re}M\right\Vert \geq\left\vert \zeta-\mathrm{i}%
\gamma\right\vert \geq\left\vert \operatorname{Im}\left(  \zeta-\mathrm{i}%
\gamma\right)  \right\vert =\left\vert \operatorname{Im}\zeta-\gamma
\right\vert \geq\left\vert \gamma\right\vert -\left\vert \operatorname{Im}%
\zeta\right\vert . \label{papxpebd2}%
\end{equation}
The proof now follows from (\ref{papxpebd1}) and (\ref{papxpebd2}).
\end{proof}

The next result we prove is a proposition on dichotomy of the spectrum of a
non-Hermitian matrix $M$ which we mentioned in the introduction of this
appendix. \ In particular, the following proposition tells us that the
spectrum of a matrix $M$ will split into two disjoint parts when the imaginary
part $\operatorname{Im}M$ is non-invertible and \textquotedblleft large" in
comparison to the real part $\operatorname{Re}M$. \ The term \textquotedblleft
large" means that the bottom of the nonzero spectrum of $\left\vert
\operatorname{Im}M\right\vert $ must be greater than twice the top of the
spectrum of $\left\vert \operatorname{Re}M\right\vert $ (where for a square
matrix $A$, $\left\vert A\right\vert =\sqrt{A^{\ast}A}$ as defined in
\cite[\S V1.4, p. 196]{ReSi1}), that is,
\begin{equation}
\min\limits_{\gamma\in\sigma\left(  \operatorname{Im}M\right)  ,\text{ }%
\gamma\not =0}\left\vert \gamma\right\vert >2\max\limits_{\lambda\in
\sigma\left(  \operatorname{Re}M\right)  }\left\vert \lambda\right\vert .
\label{apxebd7}%
\end{equation}
\qquad This of course implies $\left\Vert \operatorname{Im}M\right\Vert $ is
\textquotedblleft large" in comparison to $\left\Vert \operatorname{Re}%
M\right\Vert $ since it follows from this inequality that%
\begin{equation}
\left\Vert \operatorname{Im}M\right\Vert >2\left\Vert \operatorname{Re}%
M\right\Vert . \label{apxebd8}%
\end{equation}

\begin{theorem}
[modal dichotomy]\label{apxtmd}Let $M$ be any $n\times n$ matrix which is
non-Hermitian such that its imaginary part $\operatorname{Im}M$ is
non-invertible with rank $m=\operatorname{rank}\left(  \operatorname{Im}%
M\right)  $. \ Denote by $\gamma_{j}$, $j=1,$\ldots, $m $ the nonzero
eigenvalues of $\operatorname{Im}M$. \ If $\min_{1\leq j\leq m}\left\vert
\gamma_{j}\right\vert >2\left\Vert \operatorname{Re}M\right\Vert $ then
\begin{equation}
\sigma\left(  M\right)  =\sigma_{0}\left(  M\right)  \cup\sigma_{1}\left(
M\right)  ,\text{ \ \ }\sigma_{0}\left(  M\right)  \cap\sigma_{1}\left(
M\right)  =\emptyset, \label{apxebd9}%
\end{equation}
where%
\begin{align}
\sigma_{0}\left(  M\right)   &  =\left\{  \zeta\in\sigma\left(  M\right)
:\left\vert \zeta\right\vert \leq\left\Vert \operatorname{Re}M\right\Vert
\right\}  ,\label{apxebd10}\\
\sigma_{1}\left(  M\right)   &  =\left\{  \zeta\in\sigma\left(  M\right)
:\left\vert \zeta-\mathrm{i}\gamma_{j}\right\vert \leq\left\Vert
\operatorname{Re}M\right\Vert \text{ for some }j\right\}  .\nonumber
\end{align}
Furthermore, there exists unique matrices $P_{0}$, $P_{1}$ with the properties%
\begin{align}
&  (i)\text{ \ \ }\mathbf{1=}\text{ }P_{0}+P_{1},\text{ \ \ }P_{i}P_{j}%
=\delta_{ij}P_{j};\label{apxebd11}\\
&  (ii)\text{ \ \ }MP_{0}=P_{0}MP_{0},\text{ \ \ }MP_{1}=P_{1}MP_{1}%
;\nonumber\\
&  (iii)\text{ \ \ }\sigma\left(  MP_{0}|_{\operatorname{Ran}P_{0}}\right)
=\sigma_{0}\left(  M\right)  ,\text{ \ \ }\sigma\left(  MP_{1}%
|_{\operatorname{Ran}P_{1}}\right)  =\sigma_{1}\left(  M\right)  ,\nonumber
\end{align}
where $\delta_{ij}=0$ if $i\not =j$ and $\delta_{ij}=1$ if $i=j$ ($i,j=1,2$).
Moreover, these projection matrices have rank satisfying%
\begin{equation}
\operatorname{rank}P_{1}=m,\ \ \operatorname{rank}P_{0}=n-m. \label{apxebd12}%
\end{equation}

\end{theorem}

\begin{proof}
All the eigenvalues of $\operatorname{Im}M$ by hypothesis are $\gamma_{j}$,
$j=1,$\ldots, $m$ and $\gamma_{0}=0$. It follows from Proposition
\ref{apxpebd} that%
\[
\sigma\left(  M\right)  \subseteq\bigcup_{j=0}^{n}\left\{  \zeta\in%
%TCIMACRO{\U{2102} }%
%BeginExpansion
\mathbb{C}
%EndExpansion
:\left\vert \zeta-\mathrm{i}\gamma_{j}\right\vert \leq\left\Vert
\operatorname{Re}M\right\Vert \right\}
\]
and the imaginary part of any eigenvalue $\zeta$ of $M$ satisfies the
inequality%
\[
\left\Vert \operatorname{Im}M\right\Vert \geq\left\vert \operatorname{Im}%
\zeta\right\vert \geq\left\vert \gamma_{j^{\prime}}\right\vert -\left\Vert
\operatorname{Re}M\right\Vert
\]
for some $j^{\prime}\in\left\{  1,\ldots,n\right\}  $ depending on $\zeta$.
\ In particular, it follows that%
\[
\sigma\left(  M\right)  =\sigma_{0}\left(  M\right)  \cup\sigma_{1}\left(
M\right)
\]
where
\begin{align*}
\sigma_{0}\left(  M\right)   &  =\left\{  \zeta\in\sigma\left(  M\right)
:\left\vert \zeta\right\vert \leq\left\Vert \operatorname{Re}M\right\Vert
\right\}  ,\\
\sigma_{1}\left(  M\right)   &  =\left\{  \zeta\in\sigma\left(  M\right)
:\left\vert \zeta-\mathrm{i}\gamma_{j}\right\vert \leq\left\Vert
\operatorname{Re}M\right\Vert \text{ for some }j\not =0\right\}  .
\end{align*}
Now suppose that $\min_{1\leq j\leq m}\left\vert \gamma_{j}\right\vert
>2\left\Vert \operatorname{Re}M\right\Vert $. \ We will now show that
$\sigma_{0}\left(  M\right)  \cap\sigma_{1}\left(  M\right)  =\emptyset$.
\ Suppose this were not true then we could find $\zeta\in\sigma_{0}\left(
M\right)  \cap\sigma_{1}\left(  M\right)  $ which would imply for some
$\gamma_{j^{\prime}}$ with $j^{\prime}\not =0$ we have
\[
\left\Vert \operatorname{Re}M\right\Vert \geq\left\vert \zeta-\mathrm{i}%
\gamma_{j^{\prime}}\right\vert \geq\left\vert \gamma_{j^{\prime}}\right\vert
-\left\Vert \operatorname{Re}M\right\Vert \geq\min_{1\leq j\leq m}\left\vert
\gamma_{j}\right\vert -\left\Vert \operatorname{Re}M\right\Vert >\left\Vert
\operatorname{Re}M\right\Vert ,
\]
a contradiction. \ Thus $\sigma_{0}\left(  M\right)  \cap\sigma_{1}\left(
M\right)  =\emptyset$.

Now it follows from the spectral theory of matrices that the
finite-dimensional vector space $%
%TCIMACRO{\U{2102} }%
%BeginExpansion
\mathbb{C}
%EndExpansion
^{n}$ of $n\times1$ column vectors with entries in $%
%TCIMACRO{\U{2102} }%
%BeginExpansion
\mathbb{C}
%EndExpansion
$ can be written as the direct sum of two invariant subspaces for the matrix
$M$,%
\[%
%TCIMACRO{\U{2102} }%
%BeginExpansion
\mathbb{C}
%EndExpansion
^{n}=H_{0}\oplus H_{1},\text{ }H_{0}=\cup_{\zeta\in\sigma_{0}\left(  M\right)
}\operatorname{Ker}\left(  M-\zeta\mathbf{1}\right)  ^{n},\text{ }H_{1}%
=\cup_{\zeta\in\sigma_{1}\left(  M\right)  }\operatorname{Ker}\left(
M-\zeta\mathbf{1}\right)  ^{n}.
\]
That is, $H_{0}$ is the union of the generalized eigenspaces of $M$
corresponding to eigenvalues in $\sigma_{0}\left(  M\right)  $ and $H_{1}$ is
the union of generalized eigenspaces of $M$ corresponding to the eigenvalues
in $\sigma_{1}\left(  M\right)  =\sigma\left(  M\right)  \backslash\sigma
_{0}\left(  M\right)  $. \ Denote the projection matrix onto $H_{0}$ along
$H_{1}$ by $P_{0}$, that is, the matrix satisfying $P_{0}^{2}=P_{0}$,
$\operatorname{Ran}P_{0}=H_{0}$, and $\operatorname{Ran}\left(  \mathbf{1-}%
P_{0}\right)  =H_{1}$. \ Then it follows from the spectral theory of matrices
that $P_{0}$ and $P_{1}=\mathbf{1}-P_{0}$ are the unique matrices with the
properties%
\begin{align*}
&  (i)\text{ \ \ }\mathbf{1=}\text{ }P_{0}+P_{1},\text{ \ \ }P_{i}P_{j}%
=\delta_{ij}P_{j};\\
&  (ii)\text{ \ \ }MP_{0}=P_{0}MP_{0},\text{ \ \ }MP_{1}=P_{1}MP_{1};\\
&  (iii)\text{ \ \ }\sigma\left(  MP_{0}|_{\operatorname{Ran}P_{0}}\right)
=\sigma_{0}\left(  M\right)  ,\text{ \ \ }\sigma\left(  MP_{1}%
|_{\operatorname{Ran}P_{1}}\right)  =\sigma_{1}\left(  M\right)  .
\end{align*}
Now it follows that
\[
n=\dim H_{1}+\dim H_{0}=\dim\operatorname{Ran}P_{1}+\dim\operatorname{Ran}%
P_{0}=\operatorname{rank}P_{1}+\operatorname{rank}P_{0}%
\]
so that to complete the proof of this proposition we need only show that%
\[
\operatorname{rank}P_{0}=n-m\text{.}%
\]

In order to prove this we begin by giving an explicit representation of the
matrix $P_{0}$ in terms of the resolvent of $M$ which is defined as
\[
R\left(  \zeta\right)  =\left(  \zeta\mathbf{1-}M\right)  ^{-1},\text{
\ \ }\zeta\not \in \sigma\left(  M\right)  .
\]
To do this we define
\[
r_{0}=\frac{1}{2}\min_{1\leq j\leq m}\left\vert \gamma_{j}\right\vert ,\text{
\ \ }D\left(  0,r_{0}\right)  =\left\{  \zeta\in%
%TCIMACRO{\U{2102} }%
%BeginExpansion
\mathbb{C}
%EndExpansion
:\left\vert \zeta\right\vert \leq r_{0}\right\}  .
\]
We will now prove that%
\[
\sigma_{0}\left(  M\right)  \subseteq D\left(  0,r_{0}\right)  ,\text{
\ \ }\sigma_{1}\left(  M\right)  \cap D\left(  0,r_{0}\right)  =\emptyset.
\]
First, we notice that
\[
r_{0}=\left\Vert \operatorname{Re}M\right\Vert +\frac{1}{2}\left(  \min_{1\leq
j\leq m}\left\vert \gamma_{j}\right\vert -2\left\Vert \operatorname{Re}%
M\right\Vert \right)  >\left\Vert \operatorname{Re}M\right\Vert
\]
which proves $\sigma_{0}\left(  M\right)  \subseteq D\left(  0,r_{0}\right)
$. \ Next, we have that if $\zeta\in D\left(  0,r_{0}\right)  $ then for any
$\gamma_{j}$ with $j\not =0$,
\[
\left\vert \zeta-\mathrm{i}\gamma_{j}\right\vert \geq\left\vert \gamma
_{j}\right\vert -\left\vert \zeta\right\vert \geq\min_{1\leq j\leq
m}\left\vert \gamma_{j}\right\vert -\left\vert \zeta\right\vert \geq
\min_{1\leq j\leq m}\left\vert \gamma_{j}\right\vert -r_{0}=\frac{1}{2}%
\min_{1\leq j\leq m}\left\vert \gamma_{j}\right\vert >\left\Vert
\operatorname{Re}M\right\Vert
\]
implying that $\zeta\not \in $\ $\sigma_{1}\left(  M\right)  $. \ This proves
$\sigma_{1}\left(  M\right)  \cap D\left(  0,r_{0}\right)  =\emptyset$. \ It
now follows from the Cauchy-Riesz functional calculus \cite{Bau85} that%
\[
P_{0}=\frac{1}{2\pi\mathrm{i}}\int\limits_{\left\vert \zeta\right\vert =r_{0}%
}R\left(  \zeta\right)  d\zeta
\]
where the contour integral is over the simply closed positively oriented path
$\zeta\left(  \theta\right)  =r_{0}e^{\mathrm{i}\theta}$, $0\leq\theta\leq1$.

Now we consider the family of matrices%
\[
M\left(  t\right)  =(1-t)\operatorname{Re}M+\mathrm{i}\operatorname{Im}%
M,\text{ \ \ }0\leq t\leq1.
\]
Notice that for $0\leq t\leq1$ we have%
\begin{align*}
M\left(  0\right)   &  =M,\text{ \ \ }M\left(  1\right)  =\mathrm{i}%
\operatorname{Im}M,\\
\operatorname{Im}M\left(  t\right)   &  =\operatorname{Im}M,\text{
\ \ }\operatorname{Re}M\left(  t\right)  =(1-t)\operatorname{Re}M,\\
\min_{1\leq j\leq m}\left\vert \gamma_{j}\right\vert  &  >2\left\Vert
\operatorname{Re}M\right\Vert \geq2\left\Vert \operatorname{Re}M\left(
t\right)  \right\Vert .
\end{align*}
It follows from this that the results proven so far in this theorem for the
matrix $M$ apply to the matrix $M\left(  t\right)  $ for each $t\in\left[
0,1\right]  $. \ In particular, for each $t\in\left[  0,1\right]  $ we have
\[
\sigma\left(  M\left(  t\right)  \right)  =\sigma_{0}\left(  M\left(
t\right)  \right)  \cup\sigma_{1}\left(  M\left(  t\right)  \right)  ,\text{
\ \ }\sigma_{0}\left(  M\left(  t\right)  \right)  \cap\sigma_{1}\left(
M\left(  t\right)  \right)  =\emptyset,
\]
where%
\begin{align*}
\sigma_{0}\left(  M\left(  t\right)  \right)   &  =\left\{  \zeta\in
\sigma\left(  M\left(  t\right)  \right)  :\left\vert \zeta\right\vert
\leq\left\Vert \operatorname{Re}M\left(  t\right)  \right\Vert \right\}  ,\\
\sigma_{1}\left(  M\left(  t\right)  \right)   &  =\left\{  \zeta\in
\sigma\left(  M\left(  t\right)  \right)  :\left\vert \zeta-\mathrm{i}%
\gamma_{j}\right\vert \leq\left\Vert \operatorname{Re}M\left(  t\right)
\right\Vert \text{ for some }j\right\}  .
\end{align*}
Furthermore, there exists unique matrices $P_{0}\left(  t\right)  $,
$P_{1}\left(  t\right)  $ with the properties%
\begin{align*}
&  (i)\text{ \ \ }\mathbf{1=}\text{ }P_{0}\left(  t\right)  +P_{1}\left(
t\right)  ,\text{ \ \ }P_{i}\left(  t\right)  P_{j}\left(  t\right)
=\delta_{ij}P_{j}\left(  t\right)  ;\\
&  (ii)\text{ \ \ }M\left(  t\right)  P_{0}\left(  t\right)  =P_{0}\left(
t\right)  M\left(  t\right)  P_{0}\left(  t\right)  ,\text{ \ \ }M\left(
t\right)  P_{1}\left(  t\right)  =P_{1}\left(  t\right)  MP_{1}\left(
t\right)  ;\\
&  (iii)\text{ \ \ }\sigma\left(  M\left(  t\right)  P_{0}\left(  t\right)
|_{\operatorname{Ran}P_{0}\left(  t\right)  }\right)  =\sigma_{0}\left(
M\left(  t\right)  \right)  ,\text{ \ \ }\sigma\left(  M\left(  t\right)
P_{1}\left(  t\right)  |_{\operatorname{Ran}P_{1}\left(  t\right)  }\right)
=\sigma_{1}\left(  M\left(  t\right)  \right)  .
\end{align*}
Moreover,
\begin{gather*}
r_{0}=\frac{1}{2}\min_{1\leq j\leq m}\left\vert \gamma_{j}\right\vert ,\text{
\ \ }D\left(  0,r_{0}\right)  =\left\{  \zeta\in%
%TCIMACRO{\U{2102} }%
%BeginExpansion
\mathbb{C}
%EndExpansion
:\left\vert \zeta\right\vert \leq r_{0}\right\}  ,\\
\sigma_{0}\left(  M\left(  t\right)  \right)  \subseteq D\left(
0,r_{0}\right)  ,\text{ \ \ }\sigma_{1}\left(  M\left(  t\right)  \right)
\cap D\left(  0,r_{0}\right)  =\emptyset,\\
P_{0}\left(  t\right)  =\frac{1}{2\pi\mathrm{i}}\int\limits_{\left\vert
\zeta\right\vert =r_{0}}R\left(  \zeta,t\right)  d\zeta,\text{ \ \ }R\left(
\zeta,t\right)  =\left(  \zeta\mathbf{1-}M\left(  t\right)  \right)  ^{-1}.
\end{gather*}
It follows from these facts that $P_{0}\left(  t\right)  $ is a continuous
matrix projection-valued function of $t$ in $\left[  0,1\right]  $ and%
\[
\operatorname{tr}P_{0}\left(  t\right)  =\operatorname{rank}P_{0}\left(
t\right)  ,\text{ \ \ }0\leq t\leq1,
\]
where $\operatorname{tr}\left(  \cdot\right)  $ denotes the trace of a square
matrix. \ But $\operatorname{tr}P_{0}\left(  t\right)  $ is a continuous
function of $t$ since $P_{0}\left(  t\right)  $ is a continuous matrix-valued
function implying $\operatorname{rank}P_{0}\left(  t\right)  $ is a continuous
function of $t$ taking values in the nonnegative integers only. \ From this we
conclude $\operatorname{rank}P_{0}\left(  t\right)  $ is constant for $t\in$
$\left[  0,1\right]  $. \ In particular,
\[
\operatorname{rank}P_{0}=\operatorname{rank}P_{0}\left(  0\right)
=\operatorname{rank}P_{0}\left(  1\right)  .
\]
Moreover,%
\begin{align*}
P_{0}\left(  1\right)   &  =\frac{1}{2\pi\mathrm{i}}\int\limits_{\left\vert
\zeta\right\vert =r_{0}}R\left(  \zeta,1\right)  d\zeta,\text{ \ \ }R\left(
\zeta,1\right)  =\left(  \zeta\mathbf{1-}\mathrm{i}\operatorname{Im}M\right)
^{-1},\\
\sigma\left(  \mathrm{i}\operatorname{Im}M\right)   &  =\sigma_{0}\left(
M\left(  1\right)  \right)  =\left\{  \zeta\in\sigma\left(  M\left(  1\right)
\right)  :\left\vert \zeta\right\vert \leq\left\Vert \operatorname{Re}M\left(
1\right)  \right\Vert \right\}  =\left\{  0\right\} \\
\operatorname{Ran}P_{0}\left(  1\right)   &  =\cup_{\zeta\in\sigma_{0}\left(
M\left(  1\right)  \right)  }\operatorname{Ker}\left(  M\left(  1\right)
-\zeta\mathbf{1}\right)  ^{n}=\operatorname{Ker}\left(  \mathrm{i}%
\operatorname{Im}M\right)  ^{n}=\operatorname{Ker}\left(  \operatorname{Im}%
M\right)  .
\end{align*}
But this implies that%
\[
\operatorname{rank}P_{0}=\operatorname{rank}P_{0}\left(  1\right)
=\dim\operatorname{Ran}P_{0}\left(  1\right)  =\dim\operatorname{Ker}\left(
\operatorname{Im}M\right)  =n-\operatorname{rank}\left(  \operatorname{Im}%
M\right)  =n-m.
\]
This completes the proof.
\end{proof}

We conclude this appendix by giving an alternative characterization of the
subsets $\sigma_{0}\left(  M\right)  $ and $\sigma_{1}\left(  M\right)  $ from
Theorem \ref{apxtmd} in terms of the magnitude of the imaginary parts of the
eigenvalues of $M$. \ In particular, the next corollary tells us that if the
$\operatorname{Im}M$ is associated with losses and if the imaginary part of
the eigenvalues of $M$ are associated with damping then only the eigenmodes
with eigenvalues in $\sigma_{1}\left(  M\right)  $ are susceptible to large
damping, relative to the norm of $\operatorname{Re}M$, when losses are large.

\begin{corollary}
\label{apxcmd}If $\min_{1\leq j\leq m}\left\vert \gamma_{j}\right\vert
>2\left\Vert \operatorname{Re}M\right\Vert $ then%
\begin{align}
\sigma_{0}\left(  M\right)   &  =\left\{  \zeta\in\sigma\left(  M\right)
:\left\vert \operatorname{Im}\zeta\right\vert \leq\left\Vert \operatorname{Re}%
M\right\Vert \right\}  ,\text{ }\label{apxebd13}\\
\sigma_{1}\left(  M\right)   &  =\left\{  \zeta\in\sigma\left(  M\right)
:\left\vert \operatorname{Im}\zeta\right\vert \geq\min_{1\leq j\leq
m}\left\vert \gamma_{j}\right\vert -\left\Vert \operatorname{Re}M\right\Vert
\right\}  .\nonumber
\end{align}

\end{corollary}

\begin{proof}
By Theorem \ref{apxtmd} we know that%
\[
\sigma\left(  M\right)  =\sigma_{0}\left(  M\right)  \cup\sigma_{1}\left(
M\right)  ,\text{ \ \ }\sigma_{0}\left(  M\right)  \cap\sigma_{1}\left(
M\right)  =\emptyset\text{,}%
\]
where%
\begin{align*}
\sigma_{0}\left(  M\right)   &  =\left\{  \zeta\in\sigma\left(  M\right)
:\left\vert \zeta\right\vert \leq\left\Vert \operatorname{Re}M\right\Vert
\right\}  ,\text{ }\\
\sigma_{1}\left(  M\right)   &  =\left\{  \zeta\in\sigma\left(  M\right)
:\left\vert \zeta-\mathrm{i}\gamma_{j}\right\vert \leq\left\Vert
\operatorname{Re}M\right\Vert \text{ for some }j\right\}  ,
\end{align*}
and $\gamma_{j}$, $j=1,$\ldots, $m$ are the nonzero eigenvalues of
$\operatorname{Im}M$. \ In particular, these facts imply immediately that%
\[
\sigma_{1}\left(  M\right)  \supseteq\left\{  \zeta\in\sigma\left(  M\right)
:\left\vert \operatorname{Im}\zeta\right\vert >\left\Vert \operatorname{Re}%
M\right\Vert \right\}  ,\text{ }\sigma_{0}\left(  M\right)  \subseteq\left\{
\zeta\in\sigma\left(  M\right)  :\left\vert \operatorname{Im}\zeta\right\vert
\leq\left\Vert \operatorname{Re}M\right\Vert \right\}  .
\]
Now if $\zeta\in\sigma_{1}\left(  M\right)  $ then there exists a nonzero
eigenvalue $\gamma$ of $\operatorname{Im}M$ such that%
\[
\left\Vert \operatorname{Re}M\right\Vert \geq\left\vert \zeta-\mathrm{i}%
\gamma\right\vert =\left\vert \operatorname{Re}\zeta+\mathrm{i}\left(
\operatorname{Im}\zeta-\gamma\right)  \right\vert \geq\left\vert
\operatorname{Im}\zeta-\gamma\right\vert \geq\left\vert \gamma\right\vert
-\left\vert \operatorname{Im}\zeta\right\vert \geq\min_{1\leq j\leq
m}\left\vert \gamma_{j}\right\vert -\left\vert \operatorname{Im}%
\zeta\right\vert
\]
which by the hypothesis $\min_{1\leq j\leq m}\left\vert \gamma_{j}\right\vert
>2\left\Vert \operatorname{Re}M\right\Vert $ implies%
\[
\sigma_{1}\left(  M\right)  \subseteq\left\{  \zeta\in\sigma\left(  M\right)
:\left\vert \operatorname{Im}\zeta\right\vert \geq\min_{1\leq j\leq
m}\left\vert \gamma_{j}\right\vert -\left\Vert \operatorname{Re}M\right\Vert
\right\}  \subseteq\left\{  \zeta\in\sigma\left(  M\right)  :\left\vert
\zeta\right\vert >\left\Vert \operatorname{Re}M\right\Vert \right\}  \text{.}%
\]
The corollary follows immediately now from these facts.
\end{proof}

\section{Appendix: Energetics\label{apenerg}}

The term of \emph{energetics} refers to a fundamental set up for a system
evolution when two energies are defined for any system configuration, namely
the \emph{kinetic energy} and the \emph{potential energy}. Here is a concise
description of the energetics given by H. Poincar\'{e}, \cite[p.
115-116]{Poincare FS}:

\begin{quote}
\textquotedblleft The difficulties inherent in the classic mechanics have led
certain minds to prefer a new system they call \emph{energetics}.

Energetics took its rise as an outcome of the discovery of the principle of
the conservation of energy. Helmholtz gave it its final form.

It begins by defining two quantities which play the fundamental role in this
theory. They are \emph{kinetic energy, or vis viva}, and \emph{potential
energy}.

All the changes which bodies in nature can undergo are regulated by two
experimental laws:

1$^{\circ}$. The sum of kinetic energy and potential energy is constant. This
is the principle of the conservation of energy.

2$^{\circ}$. If a system of bodies is at $A$ at the time $t_{0}$ and at $B$ at
the time $t_{1}$, it always goes from the first situation to the second in
such a way that the mean value of the difference between the two sorts of
energy, in the interval of time which separates the two epochs $t_{0}$ and
$t_{1}$, may be as small as possible.

This is Hamilton's principle, which is one of the forms of the principle of
least action."
\end{quote}

In other words according to energetics the kinetic and potential energies
$\mathcal{T}$ and $\mathcal{V}$ that are related to the Lagrangian
$\mathcal{L}$ and the total energy $\mathcal{H}$ by the following relations%
\begin{equation}
\mathcal{L}=\mathcal{T}-\mathcal{V},\qquad\mathcal{H}=\mathcal{T}+\mathcal{V},
\label{lagham1}%
\end{equation}
implying%
\begin{equation}
\mathcal{T}=\frac{1}{2}\left(  \mathcal{L}+\mathcal{H}\right)  ,\qquad
\mathcal{V}=\frac{1}{2}\left(  \mathcal{H-L}\right)  . \label{lagham2}%
\end{equation}
We assume the total energy to be equal to the Hamiltonian, that is to be
defined by the Lagrangian through the Legendre transformation%
\begin{equation}
\mathcal{H}=\frac{\partial\mathcal{L}}{\partial\dot{Q}}\dot{Q}-\mathcal{L}.
\label{lagham3}%
\end{equation}
Then we arrive at the following uniquely defined expressions for the kinetic
and potential energies%
\begin{equation}
\mathcal{T}=\frac{1}{2}\frac{\partial\mathcal{L}}{\partial\dot{Q}}\dot
{Q},\qquad\mathcal{V}=\frac{1}{2}\frac{\partial\mathcal{L}}{\partial\dot{Q}%
}\dot{Q}-\mathcal{L}. \label{lagham4}%
\end{equation}
We are particularly interested in the case of the quadratic Lagrangian%
\begin{equation}
\mathcal{L}=\mathcal{L}\left(  Q,\dot{Q}\right)  =\frac{1}{2}\dot
{Q}^{\mathrm{T}}\alpha\dot{Q}+\dot{Q}^{\mathrm{T}}\theta Q-\frac{1}%
{2}Q^{\mathrm{T}}\eta Q, \label{lagham5}%
\end{equation}
where $\alpha$ and $\eta$ are real symmetric matrices which are positive
definite and positive semidefinite, respectively, and $\theta$ is a real
skew-symmetric matrix, that is%
\begin{equation}
\alpha^{\mathrm{T}}=\alpha>0,\qquad\eta^{\mathrm{T}}=\eta\geq0,\qquad
\theta^{\mathrm{T}}=-\theta, \label{lagham6}%
\end{equation}
Then according to (\ref{lagham3}) the total energy $\mathcal{H}$ the takes the
form%
\begin{equation}
\mathcal{H}=\frac{1}{2}\dot{Q}^{\mathrm{T}}\alpha\dot{Q}+\frac{1}%
{2}Q^{\mathrm{T}}\eta Q\geq0. \label{lagham7}%
\end{equation}
Consequently, as it follows from (\ref{lagham2}) the kinetic and the potential
energies are defined by%
\begin{equation}
\mathcal{T}=\frac{1}{2}\dot{Q}^{\mathrm{T}}\alpha\dot{Q}+\frac{1}{2}\dot
{Q}^{\mathrm{T}}\theta Q,\quad\mathcal{V}=\frac{1}{2}Q^{\mathrm{T}}\eta
Q-\frac{1}{2}\dot{Q}^{\mathrm{T}}\theta Q. \label{lagham8}%
\end{equation}
The Euler-Lagrange equations corresponding to the quadratic Lagrangian
(\ref{lagham5}) are
\begin{equation}
\alpha\ddot{Q}+2\theta\dot{Q}+\eta Q=0. \label{lagham9}%
\end{equation}
Multiplying the above equation by $Q^{\mathrm{T}}$ from the left yields%
\begin{equation}
Q^{\mathrm{T}}\alpha\ddot{Q}+2Q^{\mathrm{T}}\theta\dot{Q}+Q^{\mathrm{T}}\eta
Q=0, \label{lagham10}%
\end{equation}
and this equation can be recast recast as%
\begin{equation}
\partial_{t}\left(  Q^{\mathrm{T}}\alpha\dot{Q}\right)  -\dot{Q}^{\mathrm{T}%
}\alpha\dot{Q}-2\dot{Q}^{\mathrm{T}}\theta Q+Q^{\mathrm{T}}\eta Q=0.
\label{lagham11}%
\end{equation}
Combining the above equation with (\ref{lagham2}) and (\ref{lagham5}) we
arrive at the following important identity that holds for any solution $Q$ to
the Euler-Lagrange equation (\ref{lagham9})%
\begin{equation}
\mathcal{L}=\mathcal{T}-\mathcal{V}=\frac{1}{2}\partial_{t}\mathrm{G},\text{
\ \ }\mathrm{G}=Q^{\mathrm{T}}\alpha\dot{Q}, \label{lagham12}%
\end{equation}
where the term $\mathrm{G}$ is often referred to as the \emph{virial}.

\section{Appendix: Virial theorem\label{ApdxVirThm}}

Let us introduce now the time average $\left\langle f\right\rangle $ for a
time dependent quantity $f=f\left(  t\right)  $ defined by%
\begin{equation}
\left\langle f\right\rangle =\lim_{T\rightarrow\infty}\frac{1}{T}%
%TCIMACRO{\dint _{0}^{T}}%
%BeginExpansion
{\displaystyle\int_{0}^{T}}
%EndExpansion
f\left(  t\right)  \,\mathrm{d}t, \label{ftaver1}%
\end{equation}
and observe that if $f\left(  t\right)  $ is bounded for $-\infty<t<\infty$
then
\begin{equation}
\left\langle \partial_{t}f\right\rangle =\lim_{T\rightarrow\infty}\frac{1}%
{T}\left(  f\left(  T\right)  -f\left(  0\right)  \right)  =0. \label{ftaver2}%
\end{equation}
Notice now that in the case of interest when $\alpha>0,\eta\geq0$ and
consequently $\mathcal{H}\geq0$, all eigenfrequencies of the system are real
and hence all solutions to the Euler-Lagrange equation (\ref{lagham10}) are
bounded. Consequently, applying the relation (\ref{ftaver2}), using the
identity (\ref{lagham12}) we obtain the virial theorem:%
\begin{equation}
\left\langle \mathcal{L}\right\rangle =\left\langle \mathcal{T}\right\rangle
-\left\langle \mathcal{V}\right\rangle =\frac{1}{2}\left\langle \partial
_{t}\mathrm{G}\right\rangle =0, \label{ftaver3}%
\end{equation}

\textbf{Acknowledgment and Disclaimer:} The research of A. Figotin was
supported through Dr. A. Nachman of the U.S. Air Force Office of Scientific
Research (AFOSR), under grant number FA9550-11-1-0163. Both authors are
indebted to the referee for his valuable comments on our original paper.


\begin{thebibliography}{99999999999}                                                                                      %


\bibitem[Aitken]{Aitken}A. Aitken, \textsl{Determinants and Matrices}, 3rd
Ed., Oliver and Boyd, 1944.

\bibitem[Arnold]{Arnold}V. I. Arnold, \textsl{Mathematical methods of
classical mechanics}, Translated from the Russian by K. Vogtmann and A.
Weinstein, Second edition, Graduate Texts in Mathematics, 60, Springer-Verlag,
New York, 1989.

\bibitem[BarLan92]{BarLan92}L. Barkwell and P. Lancaster, \textsl{Overdamped
and gyroscopic vibrating systems}, Trans. ASME J. Appl. Mech. 59, no. 1,
176--181, (1992).

\bibitem[Bau85]{Bau85}H. Baumgartel, \textsl{Analytic Perturbation Theory for
Matrices and Operators}, Birkh\"{a}user Verlag, Basel 1985.

\bibitem[Duff55]{Duff55}R. J. Duffin, \textsl{A minimax theory for overdamped
networks}, J. Rational Mech. Anal. \textbf{4}, 221--233, (1955).

\bibitem[FigSch2]{FigSch2}A. Figotin and J. H. Schenker, \textsl{Hamiltonian
structure for dispersive and dissipative dynamical systems}, J. Stat. Phys.,
\textbf{128}, 969-1056, (2007).

\bibitem[FigVit1]{FigVit1}A. Figotin and I. Vitebskiy, \textsl{Spectra of
Periodic Nonreciprocal Electric Circuits}, SIAM, \textbf{61}, 6, 2008--2035, (2001).

\bibitem[FigVit8]{FigVit8}A. Figotin and I. Vitebskiy, \textsl{Absorption
suppression in photonic crystals}, Phys. Rev. B, \textbf{77}, 104421 (2008).

\bibitem[Gant]{Gantmacher}F. Gantmacher, \textsl{Lectures in Analytical
Mechanics}, Mir, 1975.

\bibitem[Gold]{Goldstein}H. Goldstein, C. Poole, and J. Safko,
\textsl{Classical Mechanics}, 3rd Ed., Addison-Wesley, 2001.

\bibitem[FigWel1]{FigWel1}A. Figotin and A. Welters, \textsl{Dissipative
properties of systems composed of high-loss and lossless components}, J. Math.
Phys. \textbf{53}, 123508 (2012).

\bibitem[Kato]{Kato}T. Kato, \textsl{Perturbation theory for linear
operators}, reprint of the 1980 edition, classics in mathematics,
Springer-Verlag, Berlin, 1995.

\bibitem[Lan66]{Lan66}P. Lancaster, \textsl{Lambda-matrices and vibrating
systems}, Dover Publications, Inc., Mineola, NY, 2002.

\bibitem[Mar88]{Mar88}A. S. Markus, \textsl{Introduction to the spectral
theory of polynomial operator pencils}, Translations of Mathematical
Monographs, 71, American Mathematical Society, Providence, RI, 1988.

\bibitem[Pain]{Pain}J. Pain, \textsl{The Physics of Vibrations and Waves}, 6th
Ed., Wiley, 2004.

\bibitem[Pars]{Pars}L. Pars, \textsl{Treatise on Analytical Mechanics},
Heinemann, 1965.

\bibitem[Poincare FS]{Poincare FS}H. Poincar\'{e}, \textsl{The Foundations of
Science}, authorized translation to English by B. Halsted, The Science Press, 1913.

\bibitem[ReSi1]{ReSi1}M. Reed and B. Simon, \textsl{Methods of modern
mathematical physics, Vol. 1: Functional analysis}, Second edition, Academic
Press, Inc., New York, 1980.

\bibitem[Yak]{Yak}V. A. Yakubovich, \textsl{The work of M. G. Kre\u{\i}n in
the theory of linear periodic Hamiltonian systems}, Ukrainian Math. J. 46, no.
1-2, 133--148 (1994).

\bibitem[YakSta]{YakSta}V. A. Yakubovich and V. M. Starzhinskii,
\textsl{Linear differential equations with periodic coefficients}, Vol. 1,
John Wiley \& Sons, 1975.

\bibitem[Zhang]{Zhang}F. Zhang, \textsl{Schur Complement and Its
Applications}, Springer, 2005.
\end{thebibliography}
\end{document}